\documentclass{amsart}

\usepackage{amsfonts, amssymb, amsmath, amsthm}                               
\usepackage{mathrsfs}                                                         
\usepackage{setspace}                                                         
\usepackage{geometry}                                                         
\usepackage{xcolor}                                                           
\usepackage{graphicx}                                                         
\usepackage{slashed}                                                          
\usepackage{comment}                                                          

\newtheorem{theorem}{Theorem} 	      	      	                              
\newtheorem{corollary}[theorem]{Corollary}     	      	      	      	      
\newtheorem{lemma}[theorem]{Lemma}     	       	      	      	      	      
\newtheorem{proposition}[theorem]{Proposition} 	      	      	      	      
\newtheorem{definition}[theorem]{Definition} 	      	      	                
\newtheorem{question}[theorem]{Question}     	      	      	      	        
\newtheorem{remark}[theorem]{Remark}                                          

\numberwithin{equation}{section}                                              
\numberwithin{theorem}{section}                                               

\newcommand{\mf}[1]{\mathfrak{#1}}                                            
\newcommand{\mi}[1]{\mathscr{#1}}                                             
\newcommand{\mc}[1]{\mathcal{#1}}                                             
\newcommand{\ms}[1]{\mathsf{#1}}                                              
\newcommand{\mr}[1]{\mathrm{#1}}                                              

\newcommand{\N}{\mathbb{N}}                                                   
\newcommand{\Z}{\mathbb{Z}}                                                   
\newcommand{\R}{\mathbb{R}}                                                   
\newcommand{\Sph}{\mathbb{S}}                                                 

\newcommand{\trace}[1]{\mr{tr}_{#1}}                                          

\newcommand{\gv}{\mathsf{g}}                                                  
\newcommand{\kv}{\mathsf{k}}                                                  
\newcommand{\Rv}{\mathsf{R}}                                                  
\newcommand{\Rcv}{\mathsf{Rc}}                                                
\newcommand{\Rsv}{\mathsf{Rs}}                                                
\newcommand{\Dv}{\mathsf{D}}                                                  

\newcommand{\gb}[1]{\mathfrak{g}^{\scriptscriptstyle (#1)}}                   
\newcommand{\gm}{\mf{g}}                                                      
\newcommand{\Dm}{\mf{D}}                                                      
\newcommand{\Rb}[1]{\mathfrak{r}^{\scriptscriptstyle (#1)}}                   
\newcommand{\Rm}{\mf{R}}                                                      
\newcommand{\Rcm}{\mf{Rc}}                                                    
\newcommand{\Rsm}{\mf{Rs}}                                                    
\newcommand{\bb}[1]{\mathfrak{b}^{\scriptscriptstyle (#1)}}                   

\newcommand{\sch}[1]{\mi{S} ( #1 )}                                           

\allowdisplaybreaks
\def\comp{1}

\begin{document}

\title[The Near-Boundary Geometry of Vacuum aAdS Spacetimes]{The Near-Boundary Geometry of Einstein-Vacuum Asymptotically Anti-de Sitter Spacetimes}
\author{Arick Shao}
\address{School of Mathematical Sciences\\
Queen Mary University of London\\
London E1 4NS\\ United Kingdom}
\email{a.shao@qmul.ac.uk}

\begin{abstract}
We study the geometry of a general class of vacuum asymptotically Anti-de Sitter spacetimes near the conformal boundary.
In particular, the spacetime is only assumed to have finite regularity, and it is allowed to have arbitrary boundary topology and geometry.
For the main results, we derive limits at the conformal boundary of various geometric quantities, and we use these limits to construct partial Fefferman--Graham expansions from the boundary.
The results of this article will be applied, in upcoming papers, toward proving symmetry extension and gravity--boundary correspondence theorems for vacuum asymptotically Anti-de Sitter spacetimes.
\end{abstract}


\maketitle

\section{Introduction} \label{sec.intro}

The main objective of this article is to study the geometry, near the timelike conformal boundary, of a wide class of asymptotically Anti-de Sitter spacetimes that also satisfy the Einstein-vacuum equations.
More specfically, assuming that such a spacetime, which can have arbitrary boundary topology and geometry, satisfies generic and finite regularity assumptions near the boundary:
\begin{enumerate}
\item We derive limits of various geometric quantities at the conformal boundary.

\item We then apply these limits in order to construct partial Fefferman--Graham expansions from the conformal boundary for these same geometric quantities.
\end{enumerate}
Moreover, for each such partial expansion, we obtain enough terms so that both ``free" coefficients---those undetermined by the Einstein-vacuum equations---are captured.

In addition, this article supplements the recent results obtained in \cite{hol_shao:uc_ads, hol_shao:uc_ads_ns} on unique continuation of waves from the conformal boundary of asymptotically Anti-de Sitter spacetimes, as well as their upcoming applications.
The estimates from \cite{hol_shao:uc_ads_ns} will be applied toward proving:
\begin{enumerate}
\item Symmetry extension results from the conformal boundary for Einstein-vacuum spacetimes.

\item Correspondence results between vacuum spacetimes and data at the conformal boundary.
\end{enumerate}
However, for these results, one needs a detailed understanding of the boundary asymptotics implied by the Einstein-vacuum equations.
This task is fulfilled by the present paper.

\subsection{Background}

\emph{Anti-de Sitter} (abbreviated \emph{AdS}) spacetime can be viewed as the prototypical, and maximally symmetric, solution of the \emph{Einstein-vacuum equations},
\begin{equation}
\label{eq.intro_vacuum} Rc_{ \alpha \beta } - \frac{1}{2} Rs \cdot g_{ \alpha \beta } + \Lambda \cdot g_{ \alpha \beta } = 0 \text{,}
\end{equation}
with negative \emph{cosmological constant} $\Lambda$.
Given the appropriate normalization of $\Lambda$ (see Definition \ref{def.aads_vacuum}), AdS spacetime can be globally represented as the manifold $\R^{1+n}$, with metric
\begin{equation}
\label{eq.intro_ads} g_0 := ( r^2 + 1 )^{-1} dr^2 - ( r^2 + 1 ) dt^2 + r^2 \mathring{\gamma} \text{,}
\end{equation}
where $\mathring{\gamma}$ denotes the round metric on $\Sph^{n-1}$.

One fundamental feature of AdS spacetime is the presence of a timelike conformal boundary.
For instance, if the radial coordinate $r$ is replaced by $\rho := r^{-1}$, then $g_0$ can be written as
\begin{equation}
\label{eq.intro_ads_inv} g_0 = \rho^{-2} [ ( 1 + \rho^2 )^{-1} d \rho^2 - ( 1 + \rho^2 ) dt^2 + \mathring{\gamma} ] \text{.}
\end{equation}
Thus, if we disregard the factor $\rho^{-2}$ in \eqref{eq.intro_ads_inv}, then we can formally attach a boundary
\begin{equation}
\label{eq.intro_boundary} ( \mi{I}, \gm_0 ) := ( \R \times \Sph^{ n - 1 }, - dt^2 + \mathring{\gamma} )
\end{equation}
at $\rho = 0$ (or equivalently, at $r = \infty$).  In particular, $( \mi{I}, \gm_0 )$ is an $n$-dimensional Lorentzian manifold, and $\gm$ represents the limit of $\rho^2 g_0$ as $\rho \searrow 0$ (or $r \nearrow \infty$).

In general, we loosely refer to spacetimes (not necessarily vacuum) that have a similar conformal boundary $( \mi{I}, \gm )$ as \emph{asympotically Anti-de Sitter}, which we will also abbreviate as \emph{aAdS}.
(A much more precise definition of aAdS spacetimes for this paper will be given in Section \ref{sec.aads}.)

In the past two decades, aAdS spacetimes have been a topic of significant interest in theoretical physics due to the AdS/CFT conjecture \cite{malda:ads_cft}, which roughly posits a holographic correspondence between gravitational theory in the ($(n+1)$-dimensional) spacetime and a conformal field theory on its ($n$-dimensional) boundary.
Despite its substantial influence in physics, there are scant rigorously formulated mathematical results that support or refute this conjecture; see, e.g., \cite{and_herz:uc_ricci, and_herz:uc_ricci_err, biq:uc_einstein, chru_delay:uc_killing, witt_yau:ads_cft}.

In the context of classical relativity, one may attempt to approach this conjecture by studying the relationship between vacuum aAdS spacetimes---representing the interior gravitational theory---and appropriate data posed at the conformal boundary $\mi{I}$.
In this direction, an essential task would be to gain a precise understanding of the geometries of such spacetimes and their asymptotics near $\mi{I}$.

\subsubsection{Fefferman--Graham Expansions}

Much of the understanding around the above question has arisen from the works of Fefferman and Graham \cite{fef_gra:conf_inv, fef_gra:amb_met}, which were concerned with constructing Einstein-vacuum spacetimes given conformal data on a null hypersurface.
From this data and the Einstein-vacuum equations, one then derives a formal series expansion; when the data is sufficiently ``analytic", this series then converges to a genuine Einstein-vacuum spacetime.

These ideas have been widely adapted by the physics community to the present aAdS setting.
Here, one first simplifies the expression for the metric by adopting special coordinates, so that the aAdS spacetime metric $g$ has a form that trivializes the (inverse) radial component:
\begin{equation}
\label{eq.intro_fg_gauge} g := \rho^{-2} ( d \rho^2 + \gv_{ a b } dx^a dx^b ) \text{.}
\end{equation}
(Here, the $x^a$'s denote non-radial coordinates that are tangent to the boundary manifold $\mi{I}$.)
Using the Einstein-vacuum equations, one then derives a formal series expansion for $\gv$ from $\mi{I}$:
\begin{equation}
\label{eq.intro_fg_exp} \gv = \begin{cases} \sum_{ k = 0 }^{ \frac{ n - 1 }{2} } \gb{ 2k } \rho^{ 2 k } + \gb{n} \rho^n + \dots & \text{$n$ odd,} \\ \sum_{ k = 0 }^{ \frac{ n - 2 }{2} } \gb{ 2k } \rho^{ 2 k } + \gb{\star} \rho^n \log \rho + \gb{n} \rho^n + \dots & \text{$n$ even,} \end{cases}
\end{equation}
When $n$ is odd, the above continues as a power series in $\rho$.
However, when $n$ is even, the series becomes polyhomogeneous beyond $n$-th order, containing terms of the form $\mf{h} \cdot \rho^k ( \log \rho )^l$.

In addition, one can observe the following features of the expansion \eqref{eq.intro_fg_exp}:
\begin{itemize}
\item The zero-order coefficient $\gb{0}$ and the $n$-th order coefficient $\gb{n}$ represent the ``free" data that are not determined by the Einstein-vacuum equations.
\footnote{More accurately, only the trace-free and divergence-free parts of $\gb{n}$ are ``free" in this sense.}

\item The coefficients $\gb{1}, \gb{2}, \dots, \gb{\star}$ that arise before $\gb{n}$ are formally determined only by $\gb{0}$.

\item All the coefficients arising after $\gb{n}$ are formally determined by both $\gb{0}$ and $\gb{n}$.
\end{itemize}
Again, under appropriate ``analyticity" (taking into account the possible polyhomogeneity) assumptions, these expansions can be shown to converge to Einstein-vacuum spacetimes \cite{fef_gra:conf_inv, fef_gra:amb_met, kichen:fg_log}.

We remark that $\gb{0}$ is precisely the conformal boundary metric $\gm$ from the preceding discussions.
Furthermore, $\gb{n}$ is closely related to the \emph{boundary stress-energy tensor} that is often studied in the physics literature \cite{deharo_sken_solod:holog_adscft, sken:aads_set} and is expected to describe the boundary conformal field theory.

While Fefferman--Graham expansions have proven quite fruitful for inferring the near-boundary heuristics of vacuum aAdS geometries, they do not provide a robust---or even generic---means of constructing vacuum aAdS spacetimes.
In particular, given ``free" boundary data $( \gb{0}, \gb{n} )$, one can solve the Einstein-vacuum equations for series solutions of the form \eqref{eq.intro_fg_exp} only within the excessively restrictive class of sufficiently ``analytic" metrics.
\footnote{Of course, this notion of ``analyticity" must be modified to account for the possible polyhomogeneity in \eqref{eq.intro_fg_gauge}.}
Therefore, if our aim is to obtain a general understanding of aAdS spacetimes near the conformal boundary, then we would need to dramatically broaden the class of metrics under consideration.

Unfortunately, the expectation is that the Einstein-vacuum equations are \emph{ill-posed} given non-``analytic" data $( \gb{0}, \gb{n})$ at $\mi{I}$.
This arises from the observation that one generally cannot solve hyperbolic PDEs from non-analytic Cauchy data imposed on a timelike hypersurface.
Continuing this analogy with hyperbolic equations, we then expect that the appropriate venue for solving the Einstein-vacuum equations in the aAdS context is as an \emph{initial-boundary value problem}:
\begin{itemize}
\item One imposes the usual initial data on a spacelike hypersurface.

\item One also imposes \emph{one} piece of boundary data (say, either $\gb{0}$ or $\gb{n}$) on $\mi{I}$.
\end{itemize}
These principles were affirmed in \cite{enc_kam:aads_hol, frie:aads_conf}, where the authors solved, under various assumptions, the Einstein-vacuum equations in aAdS settings via this initial-boundary value problem.

\subsubsection{Partial Expansions}

While we cannot rely on Fefferman--Graham expansions as a means for solving the Einstein-vacuum equations in the smooth (or finitely regular) class, the ideas could still provide powerful insights for characterizing the boundary asymptotics of vacuum aAdS spacetimes.
Even in these more general situations, one might expect that the Einstein-vacuum equations could still be exploited as before to obtain \emph{partial expansions}, up to a finite order, for the metric from the conformal boundary.
In terms of elementary analysis, this is comparable to the distinction between Taylor series (which are applicable only to analytic functions) and the finite expansions from Taylor's theorem (which apply to finitely differentiable functions).

One objective of the present article is to give a rigorous affirmation of the above expectations.
To be more specific, we will address the following question:

\begin{question} \label{q.fg_exp}
Suppose we are given a ``generic" aAdS spacetime $( \mi{M}, g )$ that also satisfies the Einstein-vacuum equations.
Then, how does $g$ behave near the conformal boundary $\mi{I}$, and what are the asymptotic properties of $g$ at $\mi{I}$?
Moreover, can we rigorously derive a partial Fefferman--Graham expansion for $g$---as well as other associated geometric quantities---from $\mi{I}$?
\end{question}

By ``generic", we mean that we make minimal assumptions on our aAdS spacetime:
\begin{enumerate}
\item The conformal boundary $( \mi{I}, \gm )$ is allowed to have arbitrary topology and geometry.

\item $g$ satisfies the Einstein-vacuum equations \eqref{eq.intro_vacuum}.

\item We also presume a minimal amount of regularity for the spacetime metric $g$.
\end{enumerate}
Regarding the last point, the rough idea is that we only assume that the $\mi{I}$-tangent (i.e., non-radial) components of $g$ remain regular up to a sufficiently high but finite order.
(See \eqref{eq.fg_main_ass_g} and \eqref{eq.fg_main_ass_Lg} for the precise conditions.)
In particular, these assumptions are satisfied by the well-known aAdS spacetimes (such as AdS spacetime itself and the Kerr-AdS families).
Moreover, the aAdS spacetimes obtained by solving the initial-boundary value problem for the Einstein-vacuum equations (for instance, in \cite{enc_kam:aads_hol, frie:aads_conf}) also satisfy these regularity criteria.

Given these very weak assumptions, the main results of this paper achieve the following:
\begin{itemize}
\item We derive limits at the conformal boundary for the metric $g$ and its $\rho$-derivatives, up to order $n$.
We also show that these limits depend only on the boundary metric $\gm$.

\item We prove the existence of the boundary limit at the $n$-th order (which corresponds to $\gb{n}$), which is not determined by the Einstein-vacuum equations.

\item We also obtain analogous limits for other geometric quantities, such as the curvature.
\end{itemize}
Using these boundary limits, we then construct the expected partial Fefferman--Graham expansions, up to and including the $n$-th order, for the above geometric quantities.
In particular, these capture both branches of ``free" data at the conformal boundary.

The results here can also be viewed as complementary to the well-posedness results for the initial-boundary value problem.
In \cite{enc_kam:aads_hol}, the authors assume initial data that solve the constraint equations and have the boundary asymptotics dictated by a partial Fefferman--Graham expansion.
They then solve the Einstein-vacuum equations, in turn showing that these boundary asymptotics are propagated to all times at which the solution exists.
In contrast, here we assume a priori (rather than solve for) a vacuum aAdS spacetime, but we impose far weaker regularity assumptions around the conformal boundary.
We then proceed to \emph{derive} the boundary asymptotics (up to and including the $n$-order) that are suggested by the Fefferman--Graham expansions.

A Riemannian analogue of our results was established by Chru\'sciel, Delay, Lee, and Skinner in \cite{chru_delay_lee_skin:fg_einstein}.
Roughly, they showed that a conformally compact Einstein manifold has boundary regularity consistent with a partial Fefferman--Graham expansion from the conformal boundary.

We also mention here the recent results of \cite{rod_shlar:selfsimilar_fg}, which is in some ways similar in nature to \cite{enc_kam:aads_hol}, but lies closer to the original context of Fefferman--Graham expansions.
Here, the authors solve, from characteristic initial data, for asymptotically self-similar solutions of the Einstein-vacuum equations (with zero cosmological constant).
In particular, the authors impose, on the initial outgoing null hypersurface, geometric data that is consistent with a partial Fefferman--Graham expansion, again up to the $n$-th order.
They then proceed to rigorously prove that this form of the partial expansion is propagated along the incoming null direction.

\subsubsection{Connections to Unique Continuation}

Though we cannot generally solve the Einstein-vacuum equations given boundary data $( \gb{0}, \gb{n} )$, we could instead formulate a slightly different problem---that of \emph{unique continuation} for the Einstein-vacuum equations from $\mi{I}$:

\begin{question} \label{q.correspondence}
Suppose we are given a vacuum aAdS spacetime $( \mi{M}, g )$.
Does its ``free" boundary values $( \gb{0}, \gb{n} )$ uniquely determine $g$?
In other words, is there a one-to-one correspondence between aAdS solutions of the Einstein-vacuum equations and some space of boundary data?
\end{question}

A related, though generally simpler, problem is whether a suitably defined symmetry of the conformal boundary must necessarily be inherited by the spacetime:

\begin{question} \label{q.symmetry}
Assume $( \mi{M}, g )$ is a vacuum aAdS spacetime, and suppose its boundary $( \mi{I}, \gm )$ has a symmetry (in an appropriate sense).
Must this symmetry necessarily extend into $( \mi{M}, g )$?
\end{question}

A crucial step toward resolving both questions has been given in \cite{hol_shao:uc_ads, hol_shao:uc_ads_ns}, which establish a unique continuation property from $\mi{I}$ for wave equations on a fixed aAdS spacetime.
However, in order to apply the results of \cite{hol_shao:uc_ads, hol_shao:uc_ads_ns} toward Questions \ref{q.correspondence} and \ref{q.symmetry}, we must first have a detailed understanding of the boundary asymptotics of vacuum aAdS spacetimes.
In particular:
\begin{itemize}
\item For Question \ref{q.correspondence}, we must establish the coefficients $( \gb{0}, \gb{n} )$ as the appropriate boundary data, and we must understand how they relate to the spacetime geometry.
We must also connect them with applications of the unique continuation results of \cite{hol_shao:uc_ads, hol_shao:uc_ads_ns}.

\item Similarly, for Question \ref{q.symmetry}, the appropriate notion of symmetry of $( \mi{I}, \gm )$ is to posit that both $\gb{0}$ and $\gb{n}$ are invariant with respect to the corresponding transformation.
Once again, we must then connect this condition with the relevant unique continuation results.
\end{itemize}
The results of the present paper provide the foundations for the above tasks.
Therefore, a second objective---and a key application---of this paper is as an important step toward answering Questions \ref{q.correspondence} and \ref{q.symmetry}.
These tasks will be addressed in future papers.

\subsection{The Main Results}

To conclude our introductory discussions, we give informal statements of the main results of this paper, as well as some of the main ideas of their proofs.

\subsubsection{The Vertical Formulation}

We begin with a few key points on the formalisms that we will use for our results and their proofs.
The main objects in our analysis are \emph{vertical tensor fields}, which can be viewed as one-parameter families, indexed by the radial coordinate $\rho$, of tensor fields on $\mi{I}$; see Definition \ref{def.aads_vertical} and the subsequent discussion for further details.

Throughout, we assume the spacetime metric $g$ is expressed in a \emph{Fefferman--Graham gauge}---that is, of the form \eqref{eq.intro_fg_gauge}.
In particular, the object $\gv$ in \eqref{eq.intro_fg_gauge} is a vertical tensor field---which we call the \emph{vertical metric}---that is also a Lorentzian metric for each fixed $\rho$-value.
One advantage of adopting a Fefferman--Graham gauge is that the entire geometric content of $g$ lies in the vertical metric $\gv$.

From $\gv$, we proceed to define other geometric objects associated with it, such as the Levi--Civita connection and the Riemann, Ricci, and scalar curvatures; for details, see Definition \ref{def.aads_vertical_metric}.

Next, we can make sense of boundary limits (as $\rho \searrow 0$) of vertical tensor fields as tensor fields on $\mi{I}$; a formal description of this is given in Definition \ref{def.aads_limit}.
Therefore, as long as the appropriate limits exist, we can express vertical tensor fields in terms of (finite) expansions from the boundary in powers of $\rho$, with each coefficient given by a tensor field on $\mi{I}$.

With the above in mind, we can now state a rough version of the main result:

\begin{theorem} \label{thm.intro_main}
Let $( \mi{M}, g )$ be an $( n + 1 )$-dimensional vacuum aAdS spacetime, and suppose that $g$ is expressed in terms of a Fefferman--Graham gauge \eqref{eq.intro_fg_gauge}.
In addition, assume:
\begin{itemize}
\item The vertical metric $\gv$ has a boundary limit $\gm$ that is a Lorentzian metric on $\mi{I}$.

\item A sufficient number of non-radial derivatives of $\gv$ remain locally bounded on $\mi{M}$.

\item $\partial_\rho \gv$ is locally bounded on $\mi{M}$.
\end{itemize}
Then, the following statements hold:
\begin{itemize}
\item $\gv$ and its (radial and non-radial) derivatives up to and including order $n$ satisfy boundary limits at $\mi{I}$.
Similar boundary limits also hold for the curvature $\Rv$ associated to $\gv$.

\item $\gv$ can be expressed as the following partial Fefferman--Graham expansion:
\begin{equation}
\label{eq.intro_main} \gv = \begin{cases} \sum_{ k = 0 }^{ \frac{ n - 1 }{2} } \gb{ 2k } \rho^{ 2 k } + \gb{n} \rho^n + \ms{r}_{ \gv } \rho^n & \text{$n$ odd,} \\ \sum_{ k = 0 }^{ \frac{ n - 2 }{2} } \gb{ 2k } \rho^{ 2 k } + \gb{\star} \rho^n \log \rho + \gb{n} \rho^n + \ms{r}_{ \gv } \rho^n & \text{$n$ even,} \end{cases}
\end{equation}
where the coefficients $\gb{0}, \gb{2}, \dots, \gb{\star}, \gb{n}$ are tensor fields on $\mi{I}$, and where the remainder $\ms{r}_{ \gv }$ is a vertical tensor field that has vanishing boundary limit.
Similar expansions, also up to $n$-th order, hold for $\Rv$ and other associated geometric objects.

\item The following coefficients depend only on $\gb{0} = \gm$: the coefficients $\gb{k}$ for all $0 \leq k < n$, the anomalous coefficient $\gb{\star}$, and both the ($\gm$-)trace and the ($\gm$-divergence) of $\gb{n}$. 
\end{itemize}
\end{theorem}

In short, we assume in Theorem \ref{thm.intro_main} that sufficiently many \emph{non-radial} derivatives of $\gv$ are well-behaved up to $\mi{I}$.
We then use the Einstein-vacuum equations to deduce the expected asymptotics for $\rho$-derivatives of $\gv$, as dictated by the Fefferman--Graham expansion, up to order $n$.

As previously mentioned, the fact that the partial expansions \eqref{eq.intro_main} hold for ``analytic" (in a polyhomogeneous sense) solutions of the Einstein-vacuum equations is already well-known.
The main novel contribution of present article is to show that even in finite regularity, the form \eqref{eq.intro_main} is, in fact, still forced by the Einstein-vacuum equations.
One particuilar point of emphasis is that the presence of the undetermined $\gb{n}$-term also needs not be assumed a priori, as the existence of the boundary limit corresponding to $\gb{n}$ is derived as part of the proof of Theorem \ref{thm.intro_main}.

\begin{remark}
While we also recover the expected formulas detailing how the coefficients $\gb{k}$ (with $0 < k < n$) and $\gb{\star}$ depend on $\gb{0}$, this part of the argument more closely resembles how the Einstein-vacuum equations are used in the ``analytic" setting.
\end{remark}

In addition, Theorem \ref{thm.intro_main} can be connected to the rigidity results of \cite{and:uc_ads} in aAdS settings.
In particular, the results of \cite{and:uc_ads} assume a vacuum aAdS spacetime that also satisfies the partial expansion \eqref{eq.intro_main}.
As a result, in this context, Theorem \ref{thm.intro_main} demonstrates that the assumption of \eqref{eq.intro_main} is unnecessary and can be replaced by much weaker conditions.

\begin{remark}
Here, we only perform our analysis up to order $n$, as this is the minimal regularity required in order to capture both parts of the ``free" boundary data for $\gv$.
If we assume additional regularity for $\gv$, then it is also possible to deduce additional terms beyond $n$-th order in \eqref{eq.intro_main}.
\end{remark}

We also stress that Theorem \ref{thm.intro_main} only represents the most notable portion of our main results.
The precise and full statements of the boundary limits for $\gv$ and its associated geometric objects are given in Theorem \ref{thm.fg_main} and Corollary \ref{thm.fg_covar}.
Formal statements of the partial Fefferman--Graham expansions are given in Theorem \ref{thm.fg_exp}.
The reader is referred to these statements for additional details, such as the precise sense of convergence for the various vertical quantities.

In addition, we obtain analogous partial expansions for components of the spacetime Weyl curvature in Corollary \ref{thm.fg_exp_W}; these will be needed in upcoming results regarding Questions \ref{q.correspondence} and \ref{q.symmetry}.
Finally, for completeness, we recall the standard partial expansion for Schwarzschild-AdS metrics (which is well-known in the physics literature; see, e.g., \cite{bautier_englert_rooman_spindel:fg_ads, tetradis:bh_holog}) in Corollaries \ref{thm.schw_ads_fg} and \ref{thm.schw_ads_exp}.

\subsubsection{Ideas of the Proof}

The proof of Theorem \ref{thm.intro_main}, given in Section \ref{sec.fg_proof}, can be divided into three steps.
In the first step, we derive limits for $\gv$ and $\partial_\rho \gv$, as well as their non-radial derivatives.
The main idea is to recast the Einstein-vacuum equations as a system of equations satisfied by the vertical metric $\gv$ and its curvature $\Rv$.
The appropriate limits can then be obtained by combining a careful analysis of these equations with the assumptions we have imposed on $\gv$.

The second step is to obtain boundary limits for $\partial_\rho^k \gv$ and $\partial_\rho^k \Rv$ for all $0 \leq k < n$.
This is systematically accomplished by differentiating the Einstein-vacuum equations and then analyzing the resulting system.
(In particular, this is equivalent to the observations that first led to the full Fefferman--Graham expansions.)
However, this process does not continue up to $n$-th order, as the Einstein-vacuum equations do not determine limits $\partial_\rho^n \gv$ and $\partial_\rho^n \Rv$.

At its most basic level, this analysis revolves around the study of ODEs of the form
\begin{equation}
\label{eq.intro_ode} x y' - c y = f \text{,} \qquad c \in \R \text{,}
\end{equation}
where $y$ is the unknown, $x$ is the independent variable, and $f$ is a fixed forcing term.
Note that \eqref{eq.intro_ode} becomes singular as $x \searrow 0$.
The key observation here is that when $c > 0$, the solution $y$ becomes determined at $x = 0$, via the relation $- c \cdot y (0) = f (0)$; see Proposition \ref{thm.limit_positive}.

Since the Einstein-vacuum equations and its $\rho$-derivatives, formulated in terms of vertical tensor fields, have the basic form \eqref{eq.intro_ode}, we can use the above observation in order to derive boundary limits for $\gv$ and its derivatives.
Moreover, in these equations (see \eqref{eq.einstein_vertical_deriv}), the corresponding constant $c$ remains positive until one reaches the equation for $\partial_\rho^n \gv$, for which $c = 0$.
This accounts for the fact that we can obtain boundary limits up to, but not including, the $n$-th order.

The final step of the proof is to examine the $n$-th order behavior.
The first observation is that the Einstein-vacuum equation do imply boundary limits for $\rho \partial_\rho^{ n + 1 } \gv$ (and, by extension, $\rho \partial_\rho^{ n + 1 } \Rv$).
These limits yield the anomalous logarithmic terms in \eqref{eq.intro_fg_exp}.
In addition, we show that these limits vanish when $n$ is odd, so that the logarithmic terms are not present in this case.
(Again, this is analogous to the derivation of the $\log$-term in the original Fefferman--Graham expansion.)

Now, the above observations yield the coefficients $\gb{k}$ and $\gb{\star}$ in \eqref{eq.intro_fg_exp}, for all $0 \leq k < n$.
However, this does not yet imply the ``free" $\gb{n}$-term, which is not determined by the Einstein-vacuum equations.
To obtain this, we require a stronger type of convergence for all the preceding limits that also establishes a minimal convergence rate; see \eqref{eq.aads_limit_ex} for details.

To be more specific, this analysis at the $n$-th order is now based on the ODE
\begin{equation}
\label{eq.intro_ode_log} x y' = f \text{.}
\end{equation}
From \eqref{eq.intro_ode_log}, we see that as $x \searrow 0$, the leading order behavior of $y$ is logarithmic:
\begin{equation}
\label{eq.intro_ode_log_2} y (x) = f (0) \cdot \log x + \text{error.}
\end{equation}
This accounts for the anomalous coefficient $\gb{\star}$.
The additional observation that is responsible for $\gb{n}$ is the following: if we also assume $f$ converges to $f (0)$ in a slightly stronger sense, then the error term in \eqref{eq.intro_ode_log_2} will have a \emph{finite} limit as $x \searrow 0$.
(However, this limit is no longer determined by $f$, which leads to the fact that $\gb{n}$ is ``free".)
See Proposition \ref{thm.limit_zero} for details.

Furthermore, the above considerations force us to modify all the preceding steps of this proof.
We have to derive this stronger convergence for $\gv$ and $\partial_\rho \gv$ in the first step of the above.
Moreover, we must then propagate this stronger convergence to higher derivatives in the second step, and then also to the logarithmic term.
These considerations further complicate the technical analysis.

\subsubsection{Outline}

In Section \ref{sec.aads_spacetime}, we give a precise description of the aAdS spacetimes that we will study.
The goal of Section \ref{sec.aads_vacuum} is to reformulate the Einstein-vacuum equations in terms of vertical tensor fields.
Some basic estimates involving these spacetimes are given in Section \ref{sec.aads_estimates}.

The main results of the paper are stated in Section \ref{sec.fg_result}, while the proof of Theorem \ref{thm.fg_main} (where the boundary limits are obtained) are given in Section \ref{sec.fg_proof}.
In addition, in Section \ref{sec.fg_schw_ads}, we recall the partial Fefferman--Graham expansion for the usual Schwarzschild-AdS metrics.

\if\comp1
Finally, Appendix \ref{sec.comp} contains the proofs of a number of propositions found within the main text.
These contain additional details and computations that are of benefit to especially interested readers, but would also be distractions from the main presentation.
\fi

\subsection{Acknowledgments}

The author thanks Gustav Holzegel and Alex McGill for numerous discussions, as well as Claude Warnick for generously providing some preliminary notes and computations.
The work in this paper is supported by EPSRC grant EP/R011982/1.

\section{Asymptotically AdS Spacetimes} \label{sec.aads}

The aim of this section is to precisely describe the spacetimes that we will study in this paper: near-boundary regions of asymptotically AdS spacetimes satisfying the Einstein-vacuum equations with negative cosmological constant.
We will accomplish this through three steps:
\begin{enumerate}
\item In Section \ref{sec.aads_spacetime_manifold}, we define the manifolds on which our spacetimes will lie.

\item In Section \ref{sec.aads_spacetime_metric}, we define the geometry of the spacetime.

\item In Section \ref{sec.aads_vacuum}, we impose the Einstein-vacuum equations on our spacetime.
\end{enumerate}
Moreover, in Sections \ref{sec.aads_vacuum} and \ref{sec.aads_estimates}, we discuss some basic properties of these spacetimes.

\subsection{aAdS Spacetimes} \label{sec.aads_spacetime}

In this subsection, we give a precise geometric definition of the asymptotically AdS spacetimes that we will study, as well as the Fefferman--Graham gauge condition.
In addition, we define the main tensorial objects that we will use for our analysis.

\subsubsection{The Spacetime Manifold} \label{sec.aads_spacetime_manifold}

The first step is to construct the spacetime manifold $\mi{M}$:

\begin{definition} \label{def.aads_manifold}
We define an \emph{aAdS region} to be a manifold (with boundary) of the form 
\begin{equation}
\label{eq.aads_manifold} \mi{M} := ( 0, \rho_0 ] \times \mi{I} \text{,} \qquad \rho_0 > 0 \text{,}
\end{equation}
where $\mi{I}$ is a smooth $n$-dimensional manifold, and where $n \in \N$.

For brevity, we will, in the context of the above, refer only to $\mi{M}$ as the aAdS region; the remaining quantities $n$, $\mi{I}$, $\rho_0$ are always assumed to be implicitly associated with $\mi{M}$.
\end{definition}

\begin{remark}
In particular, $\mi{I}$ and $\mi{M}$ in Definition \ref{def.aads_manifold} represent a \emph{boundary segment} and a \emph{spacetime segment}, respectively, of a general asymptotically AdS spacetime.
\end{remark}

\begin{definition} \label{def.aads_rho}
Given an aAdS region $\mi{M}$ as in Definition \ref{def.aads_manifold}:
\begin{itemize}
\item We let $\rho$ denote the function on $\mi{M}$ projecting onto the $( 0, \rho_0 ]$-component of $\mi{M}$.

\item We let $\partial_\rho$ denote the lift of the canonical vector field $d_\rho$ on $( 0, \rho_0 ]$ to $\mi{M}$.
\end{itemize}
\end{definition}

In our analysis, we will encounter three types of tensorial objects on an aAdS region $\mi{M}$:
\begin{enumerate}
\item Tensor fields on $\mi{M}$.

\item Vertical tensor fields on $\mi{M}$, representing $\rho$-parametrized families of tensor fields on $\mi{I}$.

\item Tensor fields on $\mi{I}$, representing limits of vertical tensor fields as $\rho \searrow 0$.
\end{enumerate}
While tensor fields on $\mi{M}$ and $\mi{I}$ are self-explanatory, the notion of vertical tensor fields was not explicitly used in \cite{hol_shao:uc_ads, hol_shao:uc_ads_ns}.
These objects can be described more precisely as follows:

\begin{definition} \label{def.aads_vertical}
Let $\ms{V}^k_l \mi{M}$ denote the \emph{vertical bundle} of rank $( k, l )$ over $\mi{M}$, which we define to be the manifold of all tensors of rank $( k, l )$ on each level set of $\rho$ in $\mi{M}$:
\footnote{Here, $T^k_l \{ \rho = \sigma \}$ denotes the usual tensor bundle of rank $( k, l )$ over the level set $\{ \rho = \sigma \}$ of $\mi{M}$.
Keeping with standard conventions, $k$ represents the contravariant rank, while $l$ denotes the covariant rank.}
\begin{equation}
\label{eq.aads_vertical} \ms{V}^k_l \mi{M} = \bigcup_{ \sigma \in ( 0, \rho ] } T^k_l \{ \rho = \sigma \} \text{.}
\end{equation}
In addition, we refer to smooth sections of $\ms{V}^k_l \mi{M}$ as \emph{vertical tensor fields} of rank $( k, l )$.
\end{definition}

Notice that a vertical tensor field of rank $( k, l )$ on an aAdS region $\mi{M}$ can be viewed as a one-parameter family, indexed by $\rho \in ( 0, \rho_0 ]$, of rank $( k, l )$ tensor fields on $\mi{I}$.
Throughout, we will make use of this interpretation of vertical tensor fields whenever convenient.

\begin{definition} \label{def.aads_tensor}
We adopt the following conventions for tensor fields on an aAdS region $\mi{M}$:
\begin{itemize}
\item We use the normal italicized font (for example, $g$) to denote tensor fields on $\mi{M}$.

\item We use serif font (for example, $\gv$) to denote vertical tensor fields on $\mi{M}$.

\item We use Fraktur font (for example, $\mf{g}$) to denote tensor fields on $\mi{I}$.
\end{itemize}
Moreover, unless otherwise stated, we always assume that a given tensor field is smooth.
\end{definition}

\begin{definition} \label{def.aads_tensor_id}
We also note the following natural identifications of tensor fields:
\begin{itemize}
\item Given a vertical tensor field $\ms{A}$ and $\sigma \in ( 0, \rho_0 ]$, we let $\ms{A} |_\sigma$ denote the tensor field on $\mi{I}$ that is obtained by restricting $\ms{A}$ to the level set $\{ \rho = \sigma \}$ (and identifying $\{ \rho = \sigma \}$ with $\mi{I}$).

\item Given a tensor field $\mf{A}$ on $\mi{I}$, we will also use $\mf{A}$ to denote the vertical tensor field on $\mi{M}$ obtained by extending $\mf{A}$ as a $\rho$-independent family of tensor fields on $\mi{I}$:
\begin{equation}
\label{eq.aads_vertical_const} \mf{A} |_\sigma := \mf{A} \text{,} \qquad \sigma \in [ 0, \rho_0 ) \text{.}
\end{equation}

\item Any vertical tensor field $\ms{A}$ can be uniquely extended to a tensor field on $\mi{M}$ of the same rank, via the condition that the contraction of any component of $\ms{A}$ with either $\partial_\rho$ or $d \rho$ (whichever is appropriate) is defined to vanish identically.
\end{itemize}
\end{definition}

More succinctly, the last part of Definition \ref{def.aads_tensor_id} notes that tensor fields on $\mi{I}$ can be equivalently viewed as vertical tensor fields that do not depend on $\rho$.
Once again, we will use both interpretations of tensor fields on $\mi{I}$ interchangeably, depending on context.

\begin{remark} \label{rmk.aads_scalar_const}
As a special case of Definition \ref{def.aads_tensor_id}, any scalar function on $\mi{I}$ can, at the same time, be viewed as a $\rho$-independent scalar function on $\mi{M}$.
\end{remark}

\begin{definition} \label{def.aads_vertical_lie}
Let $\mi{M}$ be an aAdS region, and let $\ms{A}$ be a vertical tensor field.
We define the \emph{Lie derivative} with respect to $\rho$ of $\ms{A}$, denoted $\mi{L}_\rho \ms{A}$, to be the vertical tensor field given by
\begin{equation}
\label{eq.aads_vertical_lie} \mi{L}_\rho \ms{A} |_\sigma = \lim_{ \sigma' \rightarrow \sigma } ( \sigma' - \sigma )^{-1} ( \ms{A} |_{ \sigma' } - \ms{A} |_\sigma ) \text{,} \qquad \sigma \in ( 0, \rho_0 ] \text{.}
\end{equation}
\end{definition}

Next, we set the coordinate conventions that we will use throughout this article:

\begin{definition} \label{def.aads_index}
Let $\mi{M}$ be an aAdS region, and let $( U, \varphi )$ be a coordinate system on $\mi{I}$.
\begin{itemize}
\item Let $\varphi_\rho := ( \rho, \varphi )$ denote the coordinates on $( 0, \rho_0 ] \times U$ obtained by combining $\rho$ and the $\rho$-independent extensions of the components $\varphi$ (see Remark \ref{rmk.aads_scalar_const}).

\item We use lower-case Latin indices $a, b, c, \dots$ to denote $\varphi$-coordinate components, and we use $x^a, x^b, x^c, \dots$ to denote $\varphi$-coordinate functions.

\item We use lower-case Greek indices $\alpha, \beta, \mu, \nu, \dots$ to denote $\varphi_\rho$-coordinate components.

\item As usual, repeated indices will indicate summations over the appropriate components.
\end{itemize}
\end{definition}

\begin{definition} \label{def.aads_coord}
Let $\mi{M}$ be an aAdS region.
A coordinate system $( U, \varphi )$ on $\mi{I}$ is \emph{compact} iff:
\begin{itemize}
\item $\bar{U}$ is a compact subset of $\mi{I}$.

\item $\varphi$ extends smoothly to (an open neighborhood of) $\bar{U}$.
\end{itemize}
\end{definition}

We use compact coordinates to define local notions of convergence for vertical tensor fields:

\begin{definition} \label{def.aads_norm}
Let $\mi{M}$ be an aAdS region, and let $M \geq 0$.
Moreover, let $( U, \varphi )$ be a compact coordinate system on $\mi{I}$, and let $\ms{A}$ be a vertical tensor field of rank $(k, l)$.
\begin{itemize}
\item We define the following local pointwise norm for $\ms{A}$ (with respect to $\varphi$-coordinates):
\begin{equation}
\label{eq.aads_size} | \ms{A} |_{ M, \varphi }: ( 0, \rho_0 ] \times U \rightarrow \R \text{,} \qquad | \ms{A} |_{ M, \varphi } := \sum_{ m = 0 }^M \sum_{ \substack{ a_1, \dots, a_m \\ b_1, \dots, b_k \\ c_1, \dots, c_l } } | \partial^m_{ a_1 \dots a_m } \ms{A}^{ b_1 \dots b_k }_{ c_1 \dots c_l } | \text{.}
\end{equation}

\item We also define the following local uniform norm of $\ms{A}$:
\begin{equation}
\label{eq.aads_norm} \| \ms{A} \|_{ M, \varphi } := \sup_{ ( 0, \rho_0 ] \times U } | \ms{A} |_{ M, \varphi } \text{.}
\end{equation}
\end{itemize}
\end{definition}

\begin{definition} \label{def.aads_limit}
Let $\mi{M}$, $M$, and $\ms{A}$ be as in Definition \ref{def.aads_norm}.
Moreover, let $\mf{A}$ be a tensor field on $\mi{I}$ having the same rank as $\ms{A}$.
We then define the following:
\begin{itemize}
\item $\ms{A}$ is \emph{locally bounded in $C^M$} iff for any compact coordinate system $( U, \varphi )$ on $\mi{I}$,
\begin{equation}
\label{eq.aads_bounded} \| \ms{A} \|_{ M, \varphi } < \infty \text{.}
\end{equation}

\item $\ms{A}$ \emph{locally converges in $C^M$} to $\mf{A}$, which we abbreviate as $\ms{A} \rightarrow^M \mf{A}$, iff given any compact coordinate system $( U, \varphi )$ on $\mi{I}$, we have that
\begin{equation}
\label{eq.aads_limit} \lim_{ \sigma \searrow 0 } \sup_{ \{ \sigma \} \times U } | \ms{A} - \mf{A} |_{ M, \varphi } = 0
\end{equation}

\item Moreover, we write $\ms{A} \Rightarrow^M \mf{A}$ iff for every compact coordinate system $( U, \varphi )$ on $\mi{I}$,
\begin{equation}
\label{eq.aads_limit_ex} \lim_{ \sigma \searrow 0 } \sup_{ \{ \sigma \} \times U } | \ms{A} - \mf{A} |_{ M, \varphi } = 0 \text{,} \qquad \sup_U \int_0^{ \rho_0 } \frac{1}{ \sigma } | \ms{A} - \mf{A} |_{ M, \varphi } |_\sigma d \sigma < \infty \text{.}
\end{equation}
\end{itemize}
\end{definition}

Observe that \eqref{eq.aads_limit} is a locally uniform coordinate limit in $C^M$ (involving only derivatives along $\mi{I}$) as $\rho \searrow 0$.
Moreover, the $\Rightarrow$-convergence defined in \eqref{eq.aads_limit_ex} describes a slightly stronger notion of limits that can be roughly interpreted as converging at a minimal rate.

\subsubsection{The Spacetime and Boundary Geometry} \label{sec.aads_spacetime_metric}

Having described the background manifold, as well as the types of objects we will study on them, the next step is to describe the geometry of our spacetime.
For our purposes, these will be quantified using vertical tensor fields:

\begin{definition} \label{def.aads_vertical_metric}
Let $\mi{M}$ be a aAdS region.
A vertical tensor field $\ms{h}$ of rank $( 0, 2 )$ on $\mi{M}$ is called a \emph{vertical metric} iff $\ms{h} |_\sigma$ is a Lorentzian metric on $\mi{I}$ for each $0 < \sigma \leq \rho_0$.

Furthermore, given a vertical metric $\ms{h}$ on $\mi{M}$:
\begin{itemize}
\item The metric dual of $\ms{h}$, denoted $\ms{h}^{-1}$, is the rank $( 2, 0 )$ vertical tensor field on $\mi{M}$ such that $\ms{h}^{-1} |_\sigma$ is the metric dual $( \ms{h} |_\sigma )^{-1}$ of $\ms{h} |_\sigma$ for each $0 < \sigma \leq \rho_0$.

\item Similarly, the Levi-Civita connection associated with $\ms{h}$ is the operator on vertical tensor fields such that its restriction to any level set $\{ \rho = \sigma \}$, $0 < \sigma \leq \rho_0$, coincides with the usual Levi-Civita connection of the Lorentzian metric $\ms{h} |_\sigma$.

\item The Riemann curvature associated with $\ms{h}$ is the rank $( 1, 3 )$ vertical tensor field whose restriction to any $\{ \rho = \sigma \}$, $0 < \sigma \leq \rho_0$, is the Riemann curvature associated with $\ms{h} |_\sigma$.
Analogous definitions also hold for the Ricci and scalar curvatures of $\ms{h}$.
\end{itemize}
\end{definition}

The spacetimes that we will consider in this article can now be defined as follows:

\begin{definition} \label{def.aads_metric}
$( \mi{M}, g )$ is called an \emph{FG-aAdS segment} iff the following hold:
\begin{itemize}
\item $\mi{M}$ is an aAdS region, and $g$ is a Lorentzian metric on $\mi{M}$.

\item There exists a vertical metric $\gv$ on $\mi{M}$ such that
\begin{equation}
\label{eq.aads_metric} g := \rho^{-2} ( d \rho^2 + \gv ) \text{.}
\end{equation}

\item There is a Lorentzian metric $\gm$ on $\mi{I}$ such that
\begin{equation}
\label{eq.aads_metric_limit} \gv \rightarrow^0 \gm \text{.}
\end{equation}
\end{itemize}

Moreover, in the above context, while we refer only to $( \mi{M}, g )$ as the FG-aAdS segment, we always implicitly assume the associated metrics $\gv$ and $\gm$ described in \eqref{eq.aads_metric} and \eqref{eq.aads_metric_limit}.
\end{definition}

Definition \ref{def.aads_metric} captures the minimal assumptions required for a (segment of a) spacetime to be considered ``asymptotically AdS".
We refer to the form \eqref{eq.aads_metric} for the spacetime metric $g$ as the \emph{Fefferman--Graham} (abbreviated \emph{FG}) \emph{gauge condition}, which is characterized by the $\rho$-component being (up to the conformal factor $\rho^{-2}$) trivial and fully decoupled from the remaining components.

\begin{definition} \label{def.aads_boundary}
Given an FG-aAdS segment $( \mi{M}, g )$, we refer to the $n$-dimensional (Lorentz) manifold $( \mi{I}, \gm )$ as the \emph{conformal boundary} associated with $( \mi{M}, g, \rho )$.
\end{definition}

\begin{remark}
Imposing the FG gauge condition in Definition \ref{def.aads_metric} does not result in any loss of generality.
For more general spacetime metrics $g$ having the same asymptotics as $\rho \searrow 0$, one can apply an appropriate change of coordinates to transform $g$ into the form \eqref{eq.aads_metric}.
\end{remark}

\begin{remark}
The conformal boundary $( \mi{I}, \gm )$ in Definition \ref{def.aads_boundary} is gauge-dependent, in that it depends on the choice of $\rho$.
More specifically, through coordinate transformations, one can find other boundary-defining functions $\bar{\rho}$ with respect to which $( \mi{M}, g )$ satisfies all the conditions in Definition \ref{def.aads_metric}.
Such transformations result in $\gm$ being multiplied by conformal factors; see \cite{deharo_sken_solod:holog_adscft, imbim_schwim_theis_yanki:diffeo_holog}.
\end{remark}

\begin{definition} \label{def.aads_metric_vertical}
Given an FG-aAdS segment $( \mi{M}, g )$:
\begin{itemize}
\item Let $g^{-1}$ and $\nabla$ denote the metric dual and the Levi-Civita connection, respectively, associated with $g$.
Moreover, let $R$, $Rc$, $Rs$ and $W$ denote the Riemann curvature, Ricci curvature, scalar curvature, and Weyl curvature, respectively, associated to $g$.

\item Let $\gv^{-1}$ and $\Dv$ denote the metric dual and the Levi-Civita connection, respectively, associated with the vertical metric $\gv$.
Moreover, let $\Rv$, $\Rcv$, and $\Rsv$ denote the Riemann curvature, Ricci curvature, and scalar curvature, respectively, associated to $\gv$.

\item Let $\gm^{-1}$ and $\Dm$ denote the metric dual and the Levi-Civita connection, respectively, associated with $\gm$.
In addition, let $\Rm$, $\Rcm$, and $\Rsm$ denote the Riemann curvature, Ricci curvature, and scalar curvature, respectively, that are associated to $\gm$.
\end{itemize}
Also, as is standard, we omit the superscript ``${}^{-1}$" when expressing a metric dual in index notion.
\end{definition}

\begin{remark}
To maintain consistency, we will, by default, always treat the Riemann curvature---spacetime, vertical, or boundary---as a rank $(1, 3)$ tensor field.
\end{remark}

Below, we collect some basic observations on the geometry of FG-aAdS segments:

\begin{proposition} \label{thm.g_rho}
Let $( \mi{M}, g )$ be an FG-aAdS segment.
Moreover, fix a coordinate system $( U, \varphi )$ on $\mi{I}$, and assume all indices below are with respect to $\varphi$ and $\varphi_\rho$-coordinates.
Then:
\begin{itemize}
\item The components of $g$ and $g^{-1}$ satisfy
\begin{equation}
\label{eq.g_rho} g_{ \rho \rho } = \rho^{-2} \text{,} \qquad g_{ \rho a } = 0 \text{,} \qquad g_{ a b } = \rho^{-2} \gv_{ a b } \text{,} \qquad g^{ \rho \rho } = \rho^2 \text{,} \qquad g^{ \rho a } = 0 \text{,} \qquad g^{ a b } = \rho^2 \gv^{ a b } \text{.}
\end{equation}

\item The vector field $N := \rho \partial_\rho$ satisfies
\begin{equation}
\label{eq.geod_rho} g ( N, N ) = 1 \text{,} \qquad g ( N, \partial_a ) = 0 \text{,} \qquad \nabla_N N = 0 \text{.}
\end{equation}

\item The following identity holds:
\begin{equation}
\label{eq.sff_rho} - \frac{1}{2} \rho^{-1} \mi{L}_\rho \gv_{ a b } + \rho^{-2} \gv_{ a b } = - \rho \cdot g ( \nabla_a \partial_\rho, \partial_b ) \text{.}
\end{equation}
In particular, the left-hand side of \eqref{eq.sff_rho} is precisely the second fundamental form of the level sets of $\rho$, as submanifolds of $( \mi{M}, g )$, with respect to the unit normal $N$.
\end{itemize}
\end{proposition}

\if\comp1
\begin{proof}
See Appendix \ref{sec.comp_g_rho}.
\end{proof}
\fi

Finally, we will also require the following operators later in our analysis:

\begin{definition} \label{def.aads_trace}
Let $( \mi{M}, g )$ be an FG-aAdS segment.
\begin{itemize}
\item Given a symmetric vertical tensor field $\ms{h}$ of rank $( 0, 2 )$, we define its $\gv$-trace, denoted $\trace{\gv} \ms{h}$, and its $\gv$-divergence, denoted $\Dv \cdot \ms{h}$, to be the vertical tensor fields given by
\begin{equation}
\label{eq.aads_trace_vertical} \trace{\gv} \ms{h} := \gv^{ a b } \ms{h}_{ a b } \text{,} \qquad ( \Dv \cdot \ms{h} )_a := \gv^{ b c } \Dv_b \ms{h}_{ a c } \text{.}
\end{equation}

\item Given a symmetric tensor field $\mf{h}$ on $\mi{I}$ of rank $( 0, 2 )$, we define its $\gm$-trace, denoted $\trace{\gm} \mf{h}$, and its $\gm$-divergence, denoted $\Dm \cdot \mf{h}$, to be the tensor fields on $\mi{I}$ given by
\begin{equation}
\label{eq.aads_trace_boundary} \trace{\gm} \mf{h} := \gm^{ a b } \mf{h}_{ a b } \text{,} \qquad ( \Dm \cdot \mf{h} )_a := \gm^{ b c } \Dm_b \mf{h}_{ a c } \text{.}
\end{equation}
\end{itemize}
\end{definition}

\subsection{Einstein-Vacuum Spacetimes} \label{sec.aads_vacuum}

The final assumption we will pose on our asymptotically AdS setting is that it satisfies the Einstein-vacuum equations, with negative cosmological constant:

\begin{definition} \label{def.aads_vacuum}
An FG-aAdS segment $( \mi{M}, g )$ is called a \emph{vacuum FG-aAdS segment} iff
\begin{equation}
\label{eq.einstein} Rc - \frac{1}{2} Rs \cdot g + \Lambda \cdot g = 0 \text{,} \qquad \Lambda := - \frac{ n (n - 1) }{ 2 } \text{.}
\end{equation}
\end{definition}

\begin{remark}
Definition \ref{def.aads_vacuum} normalizes the cosmological constant to a specific negative value.
This is done for convenience, in order to simplify constants from various computations.
\end{remark}

\begin{proposition} \label{thm.einstein_ex}
Suppose $( \mi{M}, g )$ is a vacuum FG-aAdS segment, and let $( U, \varphi )$ be a coordinate system on $\mi{I}$.
Then, the following hold with respect to $\varphi_\rho$-coordinates:
\begin{equation}
\label{eq.einstein_ex} Rc_{ \alpha \beta } = - n \cdot g_{ \alpha \beta } \text{,} \qquad Rs = - n (n + 1) \text{,} \qquad W_{ \alpha \beta \gamma \delta } = R_{ \alpha \beta \gamma \delta } + g_{ \alpha \gamma } g_{ \beta \delta } - g_{ \alpha \delta } g_{ \beta \gamma } \text{.}
\end{equation}
\end{proposition}

\if\comp1
\begin{proof}
See Appendix \ref{sec.comp_einstein_ex}.
\end{proof}
\fi

\subsubsection{The Vertical Formulation}

We now reformulate the Einstein-vacuum equations \eqref{eq.einstein} in terms of vertical tensor fields, in particular the vertical metric $\gv$, its $\mi{L}_\rho$-derivatives, and its associated curvature $\Rv$.
We begin by converting some standard differential geometric identities, such as the Gauss--Codazzi equations, into equations for vertical tensor fields.

\begin{proposition} \label{thm.geom_vertical}
Let $( \mi{M}, g )$ be a vacuum FG-aAdS segment, and let $( U, \varphi )$ be a coordinate system on $\mi{I}$.
Then, the following identities hold with respect to $\varphi$ and $\varphi_\rho$-coordinates:
\begin{align}
\label{eq.geom_vertical} \rho^2 W_{ \rho c a b } &= \frac{1}{2} \Dv_b \mi{L}_\rho \gv_{ a c } - \frac{1}{2} \Dv_a \mi{L}_\rho \gv_{ b c } \text{,} \\
\notag \rho^2 W_{ a b c d } &= \Rv_{ a b c d } + \frac{1}{4} ( \mi{L}_\rho \gv_{ a d } \mi{L}_\rho \gv_{ b c } - \mi{L}_\rho \gv_{ b d } \mi{L}_\rho \gv_{ a c } ) \\
\notag &\qquad + \frac{1}{2} \rho^{-1} ( \gv_{ b d } \mi{L}_\rho \gv_{ a c } + \gv_{ a c } \mi{L}_\rho \gv_{ b d } - \gv_{ a d } \mi{L}_\rho \gv_{ b c } - \gv_{ b c } \mi{L}_\rho \gv_{ a d } ) \text{,} \\
\notag \rho^2 W_{ \rho a \rho b } &= - \frac{1}{2} \mi{L}_\rho^2 \gv_{ a b } + \frac{1}{2} \rho^{-1} \mi{L}_\rho \gv_{ a b } + \frac{1}{4} \gv^{ c d } \mi{L}_\rho \gv_{ a d } \mi{L}_\rho \gv_{ b c } \text{.}
\end{align}
\end{proposition}

\begin{proof}
We begin by defining the vertical tensor fields
\begin{equation}
\label{eql.geom_vertical_0} \ms{h} := \rho^{-2} \gv \text{,} \qquad \kv := - \frac{1}{2} \rho^{-1} \mi{L}_\rho \gv + \rho^{-2} \gv \text{.}
\end{equation}
Notice that \eqref{eq.aads_metric} and \eqref{eq.sff_rho} imply that $\ms{h}$ and $\ms{k}$ are the (vertical) metrics and second fundamental forms on the level sets of $\rho$ that are induced by the spacetime metric $g$.

Observe that with respect to $\varphi$-coordinates, we have
\begin{equation}
\label{eql.geom_vertical_1} \ms{\Gamma}^a_{ b c } := \frac{1}{2} \gv^{ a d } ( \partial_b \gv_{ d c } + \partial_c \gv_{ d b } - \partial_d \gv_{ b c } ) = \frac{1}{2} \ms{h}^{ a d } ( \partial_b \ms{h}_{ d c } + \partial_c \ms{h}_{ d b } - \partial_d \ms{h}_{ b c } ) \text{,}
\end{equation}
that is, $\gv$ and $\ms{h}$ have the same Christoffel symbols $\ms{\Gamma}^a_{ b c }$.
This implies $\Dv$ is the Levi-Civita connection for both $\gv$ and $\ms{h}$.
Thus, the Codazzi equations on the level sets of $\rho$ are given by
\begin{equation}
\label{eql.geom_vertical_2} \rho R_{ \rho c a b } = \Dv_a \ms{k}_{ b c } - \Dv_b \ms{k}_{ a c } \text{.}
\end{equation}
Observe that by \eqref{eq.g_rho} and \eqref{eq.einstein_ex}, the left-hand side of \eqref{eql.geom_vertical_2} is equal to $\rho W_{ \rho c a b }$.
Thus, replacing $\ms{k}$ in \eqref{eql.geom_vertical_2} using \eqref{eql.geom_vertical_0} and noting that $\Dv \gv = 0$ yields the first identity in \eqref{eq.geom_vertical}.

Next, let $\ms{S}$ denote the Riemann curvature associated with $\ms{h}$.
Using the ($\varphi$-)coordinate representation of the Riemann curvature along with \eqref{eql.geom_vertical_1}, we have that
\[
\ms{S}^a{}_{ b c d } = ( \partial_c \ms{\Gamma}^a_{b d} - \partial_d \ms{\Gamma}^a_{b c} + \ms{\Gamma}^a_{c e} \ms{\Gamma}^e_{b d} - \ms{\Gamma}^a_{d e} \ms{\Gamma}^e_{b c} ) |_{ \rho = 0 } = \Rv^a{}_{ b c d } \text{.}
\]
Lowering the first index on both sides of the above (with respect to $\ms{h}$ on the left-hand side, and with respect to $\gv$ on the right-hand side) and recalling \eqref{eql.geom_vertical_0}, we obtain
\begin{equation}
\label{eql.geom_vertical_3} \ms{S}_{ a b c d } = \rho^{-2} \Rv_{ a b c d } \text{.}
\end{equation}
Now, the Gauss equations on the level sets of $\rho$ are given by
\[
R_{ a b c d } = \ms{S}_{ a b c d } + \ms{k}_{ a d } \ms{k}_{ b c } - \ms{k}_{ b d } \ms{k}_{ a c } \text{.}
\]
Replacing $\ms{k}$ and $\ms{S}$ in the above using \eqref{eql.geom_vertical_0} and \eqref{eql.geom_vertical_3}, we see that
\begin{align*}
\rho^2 R_{ a b c d } &= \Rv_{ a b c d } + \frac{1}{4} \mi{L}_\rho \gv_{ a d } \mi{L}_\rho \gv_{ b c } - \frac{1}{4} \mi{L}_\rho \gv_{ b d } \mi{L}_\rho \gv_{ a c } + \frac{1}{2} \rho^{-1} \gv_{ b d } \mi{L}_\rho \gv_{ a c } + \frac{1}{2} \rho^{-1} \gv_{ a c } \mi{L}_\rho \gv_{ b d } \\
\notag &\qquad - \frac{1}{2} \rho^{-1} \gv_{ a d } \mi{L}_\rho \gv_{ b c } - \frac{1}{2} \rho^{-1} \gv_{ b c } \mi{L}_\rho \gv_{ a d } + \rho^{-2} \gv_{ a d } \gv_{ b c } - \rho^{-2} \gv_{ b d } \gv_{ a c } \text{.}
\end{align*}
Replacing $R$ with $W$ using \eqref{eq.einstein_ex} and recalling \eqref{eq.aads_metric} results in the second part of \eqref{eq.geom_vertical}.

For the remaining equality in \eqref{eq.geom_vertical}, we set $N := \rho \partial_\rho$ as in Proposition \ref{thm.g_rho}, and we note that
\[
[ N, \partial_a ] = \rho [ \partial_\rho, \partial_a ] = 0 \text{.}
\]
Recalling \eqref{eq.sff_rho}, the second part of \eqref{eql.geom_vertical_0}, and the above, we then obtain
\begin{align*}
\rho \mi{L}_\rho \kv_{ a b } &= - N [ g ( \nabla_a N, \partial_b ) ] \\
&= - g ( \nabla_N ( \nabla_a N ), \partial_b ) - g ( \nabla_a N, \nabla_N \partial_b ) \\
&= [ - g ( \nabla_a ( \nabla_N N ), \partial_b ) - R ( \partial_b, N, N, \partial_a ) ] - g ( \nabla_a N, \nabla_b N ) \text{.}
\end{align*}
Since $\nabla_N N$ vanishes by \eqref{eq.geod_rho}, then \eqref{eq.sff_rho}, \eqref{eql.geom_vertical_0}, and the above imply
\begin{equation}
\label{eql.geom_vertical_4} \rho \mi{L}_\rho \ms{k}_{ a b } = R ( N, \partial_a, N, \partial_b ) - g ( \nabla_a N, \nabla_b N ) = \rho^2 R_{ \rho a \rho b } - \ms{h}^{ c d } \ms{k}_{ a d } \ms{k}_{ b c } \text{.}
\end{equation}

Moreover, turning again to the second part of \eqref{eql.geom_vertical_0}, we see that
\[
\rho \mi{L}_\rho \ms{k}_{ a b } + \ms{h}^{ c d } \ms{k}_{ a d } \ms{k}_{ b c } = - \frac{1}{2} \mi{L}_\rho^2 \gv_{ a b } + \frac{1}{2} \rho^{-1} \mi{L}_\rho \gv_{ a b } - \rho^{-2} \gv_{ a b } + \frac{1}{4} \gv^{ c d } \mi{L}_\rho \gv_{ a d } \mi{L}_\rho \gv_{ b c } \text{,}
\]
which, when combined with \eqref{eql.geom_vertical_4}, yields
\[
\rho^2 R_{ \rho a \rho b } + \rho^{-2} \gv_{ a b } = - \frac{1}{2} \mi{L}_\rho^2 \gv_{ a b } + \frac{1}{2} \rho^{-1} \mi{L}_\rho \gv_{ a b } + \frac{1}{4} \gv^{ c d } \mi{L}_\rho \gv_{ a d } \mi{L}_\rho \gv_{ b c } \text{.}
\]
The last part of \eqref{eq.geom_vertical} now follows from \eqref{eq.aads_metric}, \eqref{eq.g_rho}, and \eqref{eq.einstein_ex}, which together imply that
\[
\rho^2 R_{ \rho a \rho b } + \rho^{-2} \gv_{ a b } = \rho^2 ( R_{ \rho a \rho b } + g_{ \rho \rho } g_{ a b } - g_{ \rho a } g_{ \rho b } ) = \rho^2 W_{ \rho a \rho b } \text{.} \qedhere
\]
\end{proof}

From the identities in Proposition \ref{thm.geom_vertical}, we obtain our main relations for $\mi{L}_\rho \gv$:

\begin{proposition} \label{thm.einstein_vertical}
Let $( \mi{M}, g )$ be a vacuum FG-aAdS segment, and let $( U, \varphi )$ be a coordinate system on $\mi{I}$.
Then, the following identities hold with respect to $\varphi$-coordinates:
\begin{align}
\label{eq.einstein_vertical} 0 &= \gv^{ b c } \Dv_b \mi{L}_\rho \gv_{ a c } - \Dv_a ( \trace{\gv} \mi{L}_\rho \gv ) \text{,} \\
\notag 0 &= \rho \mi{L}_\rho^2 \gv_{ a b } - ( n - 1 ) \cdot \mi{L}_\rho \gv_{ a b } - \trace{\gv} \mi{L}_\rho \gv \cdot \gv_{ a b } - 2 \rho \cdot \Rcv_{ a b } \\
\notag &\qquad + \frac{1}{2} \rho \cdot \trace{\gv} \mi{L}_\rho \gv \cdot \mi{L}_\rho \gv_{ a b } - \rho \cdot \gv^{ c d } \mi{L}_\rho \gv_{ a c } \mi{L}_\rho \gv_{ b d } \text{.}
\end{align}
\end{proposition}

\begin{proof}
For the first part, we take a $\gv$-contraction of the first identity of \eqref{eq.geom_vertical}, in the components given by $b$ and $c$.
Since $W$ is ($g$-)trace free, the left-hand side evaluates to
\[
\gv^{ b c } \cdot \rho^2 W_{ \rho c a b } = g^{ \beta \gamma } W_{ \rho \gamma a \beta } - g^{ \rho \rho } W_{ \rho \rho a \rho } = 0 \text{,}
\]
and the first equality of \eqref{eq.einstein_vertical} follows immediately.

Next, we take a $\gv$-contraction of the second identity in \eqref{eq.geom_vertical}, with respect to the components $a$ and $c$.
Using that $W$ is ($g$-)trace free, along with \eqref{eq.g_rho}, the left-hand side becomes
\[
\gv^{ a c } \cdot \rho^2 W_{ a b c d } = - g^{ \rho \rho } W_{ \rho b \rho d } = - \rho^2 W_{ \rho b \rho d } \text{,}
\]
and hence we have
\begin{align*}
- \rho^2 W_{ \rho b \rho d } &= \Rcv_{ b d } + \frac{1}{4} ( \gv^{ a c } \mi{L}_\rho \gv_{ a d } \mi{L}_\rho \gv_{ b c } - \trace{\gv} \mi{L}_\rho \gv \cdot \mi{L}_\rho \gv_{ b d } ) \\
&\qquad + \frac{1}{2} \rho^{-1} ( \trace{\gv} \mi{L}_\rho \gv \cdot \gv_{ b d } + n \cdot \mi{L}_\rho \gv_{ b d } - \mi{L}_\rho \gv_{ b d } - \mi{L}_\rho \gv_{ b d } ) \text{.}
\end{align*}

Rearranging the above (that is, replacing the free indices $b, d$ by $a, b$, respectively) and then replacing the left-hand side using the last identity in \eqref{eq.geom_vertical}, we obtain
\begin{align*}
0 &= \frac{1}{2} \mi{L}_\rho^2 \gv_{ a b } - \frac{1}{2} ( n - 1 ) \rho^{-1} \cdot \mi{L}_\rho \gv_{ a b } - \frac{1}{2} \rho^{-1} \cdot \trace{\gv} \mi{L}_\rho \gv \cdot \gv_{ a b } - \Rcv_{ a b } \\
&\qquad + \frac{1}{4} \trace{\gv} \mi{L}_\rho \gv \cdot \mi{L}_\rho \gv_{ a b } - \frac{1}{2} \gv^{ c d } \mi{L}_\rho \gv_{ a c } \mi{L}_\rho \gv_{ b d } \text{.}
\end{align*}
Multiplying both sides by $2 \rho$ results in the second part of \eqref{eq.einstein_vertical}.
\end{proof}

We will also require corresponding differential equations for $\Rv$, as well as general commutation formulas for $\mi{L}_\rho$-derivatives.
These are somewhat standard geometric identities (for instance, in the Ricci flow literature; see \cite{chow_lu_ni:hamilton_ricci}) adapted to our current vertical setting.

\begin{proposition} \label{thm.comm_vertical}
Let $( \mi{M}, g )$ be an FG-aAdS segment, let $( U, \varphi )$ be a coordinate system on $\mi{I}$, and let $\ms{A}$ be a vertical tensor field of rank $( k, l )$.
Then, in terms of $\varphi$-coordinates, we have
\begin{align}
\label{eq.comm_vertical} [ \mi{L}_\rho, \Dv_a ] \ms{A}^{ b_1 \dots b_k }_{ c_1 \dots c_l } &= \frac{1}{2} \sum_{ j = 1 }^k \gv^{ b_j e } ( \Dv_a \mi{L}_\rho \gv_{ e d } + \Dv_d \mi{L}_\rho \gv_{ e a } - \Dv_e \mi{L}_\rho \gv_{ a d } ) \ms{A}^{ b_1 \hat{d}_j b_k }_{ c_1 \dots c_l } \\
\notag &\qquad - \frac{1}{2} \sum_{ j = 1 }^l \gv^{ d e } ( \Dv_a \mi{L}_\rho \gv_{ e c_j } + \Dv_{ c_j } \mi{L}_\rho \gv_{ e a } - \Dv_e \mi{L}_\rho \gv_{ a c_j } ) \ms{A}^{ b_1 \dots b_k }_{ c_1 \hat{d}_j c_l } \text{,}
\end{align}
where $b_1 \hat{d}_j b_k$ denotes the sequence of indices $b_1 \dots b_k$ but with $b_j$ replaced by $d$, and where $c_1 \hat{d}_j c_l$ denotes the sequence of indices $c_1 \dots c_l$ but with $c_j$ replaced by $d$.
\end{proposition}

\if\comp1
\begin{proof}
See Appendix \ref{sec.comp_comm_vertical}.
\end{proof}
\fi

\begin{proposition} \label{thm.curvature_vertical}
Let $( \mi{M}, g )$ be an FG-aAdS segment, and let $( U, \varphi )$ be a coordinate system on $\mi{I}$.
Then, the following identity holds, with respect to $\varphi$-coordinates:
\begin{equation}
\label{eq.curvature_vertical} \mi{L}_\rho \Rv^a{}_{ b c d } = \frac{1}{2} \gv^{ a e } ( \Dv_{ c d } \mi{L}_\rho \gv_{ e b } + \Dv_{ c b } \mi{L}_\rho \gv_{ e d } - \Dv_{ c e } \mi{L}_\rho \gv_{ b d } - \Dv_{ d c } \mi{L}_\rho \gv_{ e b } - \Dv_{ d b } \mi{L}_\rho \gv_{ e c } + \Dv_{ d e } \mi{L}_\rho \gv_{ b c } ) \text{.}
\end{equation}
\end{proposition}

\if\comp1
\begin{proof}
See Appendix \ref{sec.comp_curvature_vertical}.
\end{proof}
\fi

\subsubsection{Higher-Order Relations}

For our analysis, we also need higher-order equations satisfied by $\mi{L}_\rho$-derivatives of $\gv$ and $\Rv$.
However, the exact forms of these identities tend to be quite lengthy.
Fortunately, the precise forms of various terms that lie below the leading order are not important.
Thus, in order to describe and treat them more conveniently, we adopt the following schematic notations that highlight only the essential structures of ``error" terms.

\begin{definition} \label{def.schematic_trivial}
Let $( \mi{M}, g )$ be an FG-aAdS segment, and fix $k_1, l_1, k_2, l_2 \in \N$.
Let $\mc{Q}$ be a linear operator mapping vertical tensor fields of rank $( k_1, l_1 )$ to vertical tensor fields of rank $( k_2, l_2 )$.
\begin{itemize}
\item We say $\mc{Q}$ is \emph{schematically trivial} iff $\mc{Q}$ is a composition of the following operations: component permutations, (non-metric) contractions, and multiplications by a constant.

\item $\mc{Q}$ is said to be \emph{schematically $\gv$-trivial} iff $\mc{Q}$ is a composition of the following operations: component permutations, (non-metric) contractions, multiplications by a constant, tensor products with $\gv$, and tensor products with $\gv^{-1}$.
\end{itemize}
\end{definition}

\begin{definition} \label{def.schematic_vertical}
Let $( \mi{M}, g )$ be an FG-aAdS segment.
For any $N \geq 1$ and any vertical tensor fields $\ms{A}_1, \dots, \ms{A}_N$ on $\mi{M}$, we define the following schematic notations:
\begin{itemize}
\item We write $\sch{ \ms{A}_1, \dots, \ms{A}_N }$ to represent any vertical tensor field of the form
\begin{equation}
\label{eq.schematic_vertical} \sum_{ j = 1 }^T \mc{Q}_j ( \ms{A}_1 \otimes \dots \otimes \ms{A}_N ) \text{,}
\end{equation}
where $T \geq 0$, and where each $\mc{Q}_j$, $1 \leq j \leq T$, is a schematically trivial operator.

\item We write $\sch{ \gv; \ms{A}_1, \dots, \ms{A}_N }$ to represent any vertical tensor field of the form
\begin{equation}
\label{eq.schematic_vertical_metric} \sum_{ j = 1 }^T \mc{Q}_j ( \ms{A}_1 \otimes \dots \otimes \ms{A}_N ) \text{,}
\end{equation}
where $T \geq 0$, and where each $\mc{Q}_j$, $1 \leq j \leq T$, is a schematically $\gv$-trivial operator.
\end{itemize}
\end{definition}

For example, we can reformulate the equations \eqref{eq.einstein_vertical} and \eqref{eq.curvature_vertical} using schematic notations:

\begin{proposition} \label{thm.einstein_vertical_ex}
Let $( \mi{M}, g )$ be a vacuum FG-aAdS segment.
Then, the following identities hold:
\begin{align}
\label{eq.einstein_vertical_ex} \rho \mi{L}_\rho^2 \gv - ( n - 1 ) \mi{L}_\rho \gv &= \trace{\gv} \mi{L}_\rho \gv \cdot \gv + 2 \rho \cdot \Rcv + \rho \cdot \mi{S} ( \gv; \mi{L}_\rho \gv, \mi{L}_\rho \gv ) \text{,} \\
\notag \rho \mi{L}_\rho ( \trace{\gv} \mi{L}_\rho \gv ) - ( 2 n - 1 ) \trace{\gv} \mi{L}_\rho \gv &= 2 \rho \cdot \Rsv + \rho \cdot \mi{S} ( \gv; \mi{L}_\rho \gv, \mi{L}_\rho \gv ) \text{,} \\
\notag \mi{L}_\rho \ms{R} &= \mi{S} ( \gv; \Dv^2 \mi{L}_\rho \gv ) \text{.}
\end{align}
\end{proposition}

\if\comp1
\begin{proof}
See Appendix \ref{sec.comp_einstein_vertical_ex}.
\end{proof}
\fi

To extract higher-order terms of the FG expansion, we take derivatives of the equations \eqref{eq.einstein_vertical} and \eqref{eq.curvature_vertical} and commute.
The results of these computations are summarized below:

\begin{proposition} \label{thm.einstein_vertical_deriv}
Let $( \mi{M}, g )$ be a vacuum FG-aAdS segment.
Then, for any $k \geq 2$,
\begin{align}
\label{eq.einstein_vertical_deriv} 0 &= \rho \mi{L}_\rho^{ k + 1 } \gv - ( n - k ) \mi{L}_\rho^k \gv - \trace{ \gv } \mi{L}^k_\rho \gv \cdot \gv - 2 ( k - 1 ) \mi{L}_\rho^{ k - 2 } \Rcv - 2 \rho \mi{L}_\rho^{ k - 1 } \Rcv \\
\notag &\qquad + \sum_{ \substack{ j_1 + \dots + j_l = k \\ 1 \leq j_p < k } } \mi{S} ( \gv; \mi{L}_\rho^{ j_1 } \gv, \dots, \mi{L}_\rho^{ j_l } \gv ) + \sum_{ \substack{ j_1 + \dots + j_l = k + 1 \\ 1 \leq j_p \leq k } } \rho \cdot \mi{S} ( \gv; \mi{L}_\rho^{ j_1 } \gv, \dots, \mi{L}_\rho^{ j_l } \gv ) \text{,} \\
\notag 0 &= \rho \mi{L}_\rho ( \trace{\gv} \mi{L}_\rho^k \gv ) - ( 2 n - k ) \trace{\gv} \mi{L}_\rho^k \gv - 2 ( k - 1 ) \mi{L}_\rho^{ k - 2 } \Rsv - 2 \rho \mi{L}_\rho^{ k - 1 } \Rsv \\
\notag &\qquad + \sum_{ \substack{ j_1 + \dots + j_l = k \\ 1 \leq j_p < k } } \mi{S} ( \gv; \mi{L}_\rho^{ j_1 } \gv, \dots, \mi{L}_\rho^{ j_l } \gv ) + \sum_{ \substack{ j_1 + \dots + j_l = k + 1 \\ 1 \leq j_p \leq k } } \rho \cdot \mi{S} ( \gv; \mi{L}_\rho^{ j_1 } \gv, \dots, \mi{L}_\rho^{ j_l } \gv ) \text{.}
\end{align}
\end{proposition}

\begin{proof}
First, note that for any vertical tensor field $\ms{A}$, we have
\[
\mi{L}_\rho^{ k - 1 } ( \rho \ms{A} ) = ( k - 1 ) \cdot \mi{L}_\rho^{ k - 2 } \ms{A} + \rho \cdot \mi{L}_\rho^{ k - 1 } \ms{A} \text{.}
\]
Applying $\mi{L}_\rho^{ k - 1 }$ to the first part of \eqref{eq.einstein_vertical_ex} and using the above, we obtain
\begin{align}
\label{eql.einstein_vertical_deriv_1} 0 &= \rho \mi{L}_\rho^{ k + 1 } \gv - ( n - k ) \mi{L}_\rho^k \gv - \mi{L}_\rho^{ k - 1 } ( \trace{ \gv } \mi{L}_\rho \gv \cdot \gv ) - 2 ( k - 1 ) \mi{L}_\rho^{ k - 2 } \Rcv \\
\notag &\qquad - 2 \rho \mi{L}_\rho^{ k - 1 } \Rcv + \mi{L}_\rho^{ k - 1 } [ \rho \cdot \mi{S} ( \gv; \mi{L}_\rho \gv, \mi{L}_\rho \gv ) ] \text{.}
\end{align}

We now apply the Leibniz rule repeatedly to the third and last terms on the right-hand side of \eqref{eql.einstein_vertical_deriv_1}.
For the third term, if all of $\mi{L}_\rho^{ k - 1 }$ hits the $\mi{L}_\rho \gv$, we obtain the term $\trace{\gv} \mi{L}_\rho^k \gv \cdot \gv$.
Otherwise, if at least one $\mi{L}_\rho$-derivative hits elsewhere, then we obtain terms with schematic form
\[
\sum_{ \substack{ j_1 + \dots + j_l = k \\ 1 \leq j_p < k } } \mi{S} ( \gv; \mi{L}_\rho^{ j_1 } \gv, \dots, \mi{L}_\rho^{ j_l } \gv ) \text{.}
\]
Similarly, applying the Leibniz rule to the last term of \eqref{eql.einstein_vertical_deriv_1} yields terms of the forms
\[
\sum_{ \substack{ j_1 + \dots + j_l = k \\ 1 \leq j_p < k } } \mi{S} ( \gv; \mi{L}_\rho^{ j_1 } \gv, \dots, \mi{L}_\rho^{ j_l } \gv ) \text{,} \qquad \sum_{ \substack{ j_1 + \dots + j_l = k + 1 \\ 1 \leq j_p \leq k } } \rho \cdot \mi{S} ( \gv; \mi{L}_\rho^{ j_1 } \gv, \dots, \mi{L}_\rho^{ j_l } \gv ) \text{,}
\]
depending on whether one of the $\mi{L}_\rho$'s hits the factor of $\rho$.
This results in the first part of \eqref{eq.einstein_vertical_deriv}; the second part of \eqref{eq.einstein_vertical_deriv} is obtained by taking a trace of the first part.
\end{proof}

\begin{proposition} \label{thm.comm_vertical_deriv}
Let $( \mi{M}, g )$ be an FG-aAdS segment.
Then, for any $k \geq 1$,
\begin{equation}
\label{eq.comm_vertical_deriv} [ \mi{L}_\rho^k, \Dv ] \ms{A} = \sum_{ \substack{ j + j_0 + \dots + j_l = k \\ 0 \leq j < k \text{, } j_p \geq 1 } } \mi{S} ( \gv; \mi{L}_\rho^j \ms{A}, \Dv \mi{L}_\rho^{ j_0 } \gv, \mi{L}_\rho^{ j_1 } \gv, \dots, \mi{L}_\rho^{ j_l } \gv ) \text{.}
\end{equation}
\end{proposition}

\if\comp1
\begin{proof}
See Appendix \ref{sec.comp_comm_vertical_deriv}.
\end{proof}
\fi

\begin{proposition} \label{thm.curvature_vertical_deriv}
Let $( \mi{M}, g )$ be a vacuum FG-aAdS segment.
Then, for any $k \geq 2$,
\begin{align}
\label{eq.curvature_vertical_deriv} 0 &= \mi{L}_\rho^k \Rv + \sum_{ \substack{ i_1 + \dots + i_l = 2 \\ j_1 + \dots + j_l = k \\ j_p \geq 1 } } \mi{S} ( \gv; \Dv^{ i_1 } \mi{L}_\rho^{ j_1 } \gv, \dots, \Dv^{ i_l } \mi{L}_\rho^{ j_l } \gv ) \text{,} \\
\notag 0 &= \Dv \cdot \mi{L}_\rho^k \gv - \Dv ( \trace{\gv} \mi{L}^k_\rho \gv ) + \sum_{ \substack{ j_0 + \dots + j_l = k \\ 1 \leq j_p < k } } \mi{S} ( \gv; \Dv \mi{L}_\rho^{ j_0 } \gv, \mi{L}_\rho^{ j_1 } \gv, \dots, \mi{L}_\rho^{ j_l } \gv ) \text{.}
\end{align}
\end{proposition}

\if\comp1
\begin{proof}
See Appendix \ref{sec.comp_curvature_vertical_deriv}.
\end{proof}
\fi

\subsection{Some General Estimates} \label{sec.aads_estimates}

In this section, we derive some general bounds and limit properties on FG-aAdS segments that will be useful in later sections.

\subsubsection{Geometric Bounds}

We begin by establishing some bounds for various geometric quantities.
For this, we first set some notations that will make future discussions more convenient:

\begin{definition} \label{def.coord_vector}
Let $( \mi{M}, g )$ be an FG-aAdS segment, let $( U, \varphi )$ be a coordinate system on $\mi{I}$, and let $M \geq 0$.
Given a vertical tensor field $\ms{A}$, we let $[ \ms{A} ]_{ M, \varphi }$ denote the vector-valued function, with domain $( 0, \rho_0 ] \times U$, whose components are all the $\varphi$-coordinate components of $\ms{A}$ and their $\varphi$-coordinate derivatives up to and including order $M$.
\end{definition}

\begin{proposition} \label{thm.geom_bound}
Let $( \mi{M}, g )$ be an FG-aAdS segment, let $( U, \varphi )$ be a compact coordinate system on $\mi{I}$, and let $M \geq 0$.
Suppose in addition there is some constant $C > 0$ such that
\footnote{Note that the third and fourth parts of \eqref{eq.geom_bound_ass} are actually consequences of the first two parts of \eqref{eq.geom_bound_ass}.
However, we include all four conditions here in order to simplify the writing.}
\begin{equation}
\label{eq.geom_bound_ass} \| \gv \|_{ M, \varphi } \leq C \text{,} \qquad \| \gv^{-1} \|_{ 0, \varphi } \leq C \text{,} \qquad \| \gm \|_{ M, \varphi } \leq C \text{,} \qquad \| \gm^{-1} \|_{ 0, \varphi } \leq C \text{.}
\end{equation}
Then:
\begin{itemize}
\item The following estimates hold:
\begin{equation}
\label{eq.geom_bound_g} \| \gv^{-1} \|_{ M, \varphi } \lesssim_C 1 \text{,} \qquad \| \gm^{-1} \|_{ M, \varphi } \lesssim_C 1 \text{,} \qquad | \gv^{-1} - \gm^{-1} |_{ M, \varphi } \lesssim_C | \gv - \gm |_{ M, \varphi } \text{.}
\end{equation}

\item If $M \geq 2$, then the following also hold:
\begin{align}
\label{eq.geom_bound_R} \| \Rv \|_{ M - 2, \varphi } \lesssim_C 1 \text{,} \qquad \| \Rm \|_{ M - 2, \varphi } &\lesssim_C 1 \text{,} \qquad | \Rv - \Rm |_{ M - 2, \varphi } \lesssim | \gv - \gm |_{ M, \varphi } \text{,} \\
\notag \| \Rcv \|_{ M - 2, \varphi } \lesssim_C 1 \text{,} \qquad \| \Rcm \|_{ M - 2, \varphi } &\lesssim_C 1 \text{,} \qquad | \Rcv - \Rcm |_{ M - 2, \varphi } \lesssim | \gv - \gm |_{ M, \varphi } \text{,} \\
\notag \| \Rsv \|_{ M - 2, \varphi } \lesssim_C 1 \text{,} \qquad \| \Rsm \|_{ M - 2, \varphi } &\lesssim_C 1 \text{,} \qquad | \Rsv - \Rsm |_{ M - 2, \varphi } \lesssim | \gv - \gm |_{ M, \varphi } \text{.}
\end{align}
\end{itemize}
\end{proposition}

\begin{proof}
Throughout the proof, we will always index and differentiate with respect to $\varphi$-coordinates.
First, recall from basic linear algebra that we can write
\begin{equation}
\label{eql.geom_bound_1} \gv^{ b c } = \frac{1}{ \mc{V}_\varphi ( \gv ) } \mc{P}_0 ( [ \gv ]_{ 0, \varphi } ) \text{,} \qquad \gm^{ b c } = \frac{1}{ \mc{V}_\varphi ( \gm ) } \mc{P}_0 ( [ \gm ]_{ 0, \varphi } ) \text{,}
\end{equation}
for any pair of indices $b$ and $c$, where $\mc{P}_0$ is some polynomial function of its arguments, and where $\mc{V}_\varphi ( \gv )$ and $\mc{V}_\varphi ( \gm )$ denote the determinants of the matrices formed from the $\varphi$-components of $\gv$ and $\gm$, respectively.
Note that from the assumptions \eqref{eq.geom_bound_ass}, we have
\begin{equation}
\label{eql.geom_bound_2} | \mc{V}_\varphi ( \gv ) | \simeq 1 \text{,} \qquad | \mc{V}_\varphi ( \gm ) | \simeq 1 \text{.}
\end{equation}

We now take $m$ derivatives of \eqref{eql.geom_bound_1}, for every $0 \leq m \leq M$, which yields
\begin{equation}
\label{eql.geom_bound_3} \partial^m_{ d_1 \dots d_m } \gv^{ b c } = \frac{1}{ [ \mc{V}_\varphi ( \gv ) ]^{ m + 1 } } \mc{P}_m ( [ \gv ]_{ m, \varphi } ) \text{,} \qquad \partial^m_{ d_1 \dots d_m } \gm^{ b c } = \frac{1}{ [ \mc{V}_\varphi ( \gm ) ]^{ m + 1 } } \mc{P}_m ( [ \gm ]_{ m, \varphi } ) \text{.}
\end{equation}
where $\mc{P}_m$ is a polynomial function of its arguments.
Thus, applying \eqref{eq.geom_bound_ass} to \eqref{eql.geom_bound_3} yields
\[
| \partial^m_{ d_1 \dots d_m } \gv^{ b c } | \lesssim | \mc{P}_m ( [ \gv ]_{ m, \varphi } ) | \lesssim 1 \text{,} \qquad | \partial^m_{ d_1 \dots d_m } \gm^{ b c } | \lesssim | \mc{P}_m ( [ \gm ]_{ m, \varphi } ) | \lesssim 1 \text{,} \qquad 0 \leq m \leq M \text{,}
\]
which imply the first two estimates of \eqref{eq.geom_bound_g}.
Furthermore, from \eqref{eql.geom_bound_2} and \eqref{eql.geom_bound_3}, we see that
\begin{align*}
| \partial^m_{ d_1 \dots d_m } ( \gv^{ b c } - \gm^{ b c } ) | &= \left| \frac{ [ \mc{V}_\varphi ( \gm ) ]^{ m + 1 } \mc{P}_m ( [ \gv ]_{ m, \varphi } ) - [ \mc{V}_\varphi ( \gv ) ]^{ m + 1 } \mc{P}_m ( [ \gm ]_{ m, \varphi } ) }{ [ \mc{V}_\varphi ( \gv ) \cdot \mc{V}_\varphi ( \gm ) ]^{ m + 1 } } \right| \\
&\lesssim | [ \mc{V}_\varphi ( \gm ) ]^{ m + 1 } \mc{P}_m ( [ \gv ]_{ m, \varphi } ) - [ \mc{V}_\varphi ( \gv ) ]^{ m + 1 } \mc{P}_m ( [ \gm ]_{ m, \varphi } ) | \\
&\lesssim | \gv - \gm |_{ m, \varphi } \text{,}
\end{align*}
where $0 \leq m \leq M$, and where we also used \eqref{eq.geom_bound_ass} in the last step.
Summing the above inequality over all indices and all $0 \leq m \leq M$ results in the final estimate of \eqref{eq.geom_bound_g}.

Next, from standard formulas for the coordinate components of $\Rv$ and $\smash{\Rm}$ in terms of the corresponding Christoffel symbols for $\gv$ and $\gm$, respectively, we see---for any indices $b, c, d, e$---that
\[
\Rv^e{}_{ b c d } = \bar{\mc{P}}_0 ( [ \gv ]_{ 2, \varphi }, [ \gv^{-1} ]_{ 1, \varphi } ) \text{,} \qquad \Rm^e{}_{ b c d } = \bar{\mc{P}}_0 ( [ \gm ]_{ 2, \varphi }, [ \gm^{-1} ]_{ 1, \varphi } ) \text{,}
\]
where $\bar{\mc{P}}_0$ a polynomial.
If $M \geq 2$, then differentiating the above yields, for each $2 \leq m \leq M$,
\begin{align}
\label{eql.geom_bound_10} \partial^{ m - 2 }_{ f_1 \dots f_{ m - 2 } } \Rv^e{}_{ b c d } &= \bar{\mc{P}}_{ m - 2 } ( [ \gv ]_{ m, \varphi }, [ \gv^{-1} ]_{ m - 1, \varphi } ) \text{,} \\
\notag \partial^{ m - 2 }_{ f_1 \dots f_{ m - 2 } } \Rm^a{}_{ b c d } &= \bar{\mc{P}}_{ m - 2 } ( [ \gm ]_{ m, \varphi }, [ \gm^{-1} ]_{ m - 1, \varphi } ) \text{,}
\end{align}
where $\bar{\mc{P}}_{ m - 2 }$ is again a polynomial.
The first two parts of \eqref{eq.geom_bound_R} now follow from \eqref{eq.geom_bound_ass}, \eqref{eq.geom_bound_g}, and \eqref{eql.geom_bound_10}.
Moreover, applying \eqref{eq.geom_bound_ass} and \eqref{eq.geom_bound_g} to the difference of the identities in \eqref{eql.geom_bound_10} yields
\begin{align*}
| \partial^{ m - 2 }_{ f_1 \dots f_{ m - 2 } } ( \Rv^e{}_{ b c d } - \Rm^e{}_{ b c d } ) | &\lesssim | \bar{\mc{P}}_{ m - 2 } ( [ \gv ]_{ m, \varphi }, [ \gv^{-1} ]_{ m - 1, \varphi } ) - \bar{\mc{P}}_{ m - 2 } ( [ \gm ]_{ m, \varphi }, [ \gm^{-1} ]_{ m - 1, \varphi } ) | \\
&\lesssim | \gv - \gm |_{ m, \varphi } + | \gv^{-1} - \gm^{-1} |_{ m - 1, \varphi } \\
&\lesssim | \gv - \gm |_{ m, \varphi } \text{,}
\end{align*}
for any $2 \leq m \leq M$.
This implies the third estimate of \eqref{eq.geom_bound_R}.

The fourth, fifth, and sixth parts of \eqref{eq.geom_bound_R} follow trivially from the first three, since the Ricci curvature is a (non-metric) contraction of the Riemann curvature. 
The seventh and eight bounds of \eqref{eq.geom_bound_R} are consequences of the definition of scalar curvature, \eqref{eq.geom_bound_g}, and previous bounds in \eqref{eq.geom_bound_R}:
\[
| \Rsv |_{ M - 2, \varphi } \lesssim | \gv^{-1} |_{ M - 2, \varphi } | \Rcv |_{ M - 2, \varphi } \lesssim 1 \text{,} \qquad | \Rsm |_{ M - 2, \varphi } \lesssim | \gm^{-1} |_{ M - 2, \varphi } | \Rcm |_{ M - 2, \varphi } \lesssim 1 \text{.}
\]
Finally, since
\[
\Rsv - \Rsm = ( \gv^{ a b } - \gm^{ a b } ) \Rcv_{ a b } + \gm^{ a b } ( \Rcv_{ a b } - \Rcm_{ a b } ) \text{,}
\]
then \eqref{eq.geom_bound_g}, the fourth part of \eqref{eq.geom_bound_R}, and the sixth part of \eqref{eq.geom_bound_R} imply
\begin{align*}
| \Rsv - \Rsm |_{ M - 2, \varphi } &\lesssim | \gv^{-1} - \gm^{-1} |_{ M - 2, \varphi } | \Rcv |_{ M - 2, \varphi } + | \gm^{-1} |_{ M - 2, \varphi } | \Rcv - \Rcm |_{ M - 2, \varphi } \\
&\lesssim | \Rcv - \Rcm |_{ M - 2, \varphi } + | \gv^{-1} - \gm^{-1} |_{ M - 2, \varphi } \\
&\lesssim | \gv - \gm |_{ M, \varphi } \text{,}
\end{align*}
which is the last inequality of \eqref{eq.geom_bound_R}.
\end{proof}

\begin{proposition} \label{thm.geom_bound_deriv}
Let $( \mi{M}, g )$ be an FG-aAdS segment, let $( U, \varphi )$ be a compact coordinate system on $\mi{I}$, and let $M > 0$.
In addition, assume that the bounds \eqref{eq.geom_bound_ass} hold.
Then, for any vertical tensor field $\ms{A}$ on $\mi{M}$ and any tensor field $\mf{A}$ on $\mi{I}$ of the same rank, we have the following:
\begin{align}
\label{eq.geom_bound_deriv} \sum_{ m = 0 }^M | \Dv^m \ms{A} - \Dm^m \mf{A} |_{ 0, \varphi } &\lesssim_C | \ms{A} - \mf{A} |_{ M, \varphi } + ( | \ms{A} |_{ M - 1, \varphi } + | \mf{A} |_{ M - 1, \varphi } ) | \gv - \gm |_{ M, \varphi } \text{,} \\
\notag | \ms{A} - \mf{A} |_{ M, \varphi } &\lesssim_C \sum_{ m = 0 }^M | \Dv^m \ms{A} - \Dm^m \mf{A} |_{ 0, \varphi } + ( | \ms{A} |_{ M - 1, \varphi } + | \mf{A} |_{ M - 1, \varphi } ) | \gv - \gm |_{ M, \varphi } \text{.}
\end{align}
\end{proposition}

\begin{proof}
Let $(k, l)$ be the rank of $\ms{A}$ and $\mf{A}$, and fix indices $a_1, \dots, a_M, b_1 \dots, b_k, c_1, \dots, c_l$.
Then, indexing and differentiating with respect to $\varphi$-coordinates, we see from the coordinate formula for covariant derivatives that for each $0 < m \leq M$, there is a polynomial $\mc{R}_m$ such that
\begin{align}
\label{eql.geom_bound_deriv_0} \Dv_{ a_1 \dots a_m } \ms{A}^{ b_1 \dots b_k }_{ c_1 \dots c_l } &= \partial^m_{ a_1 \dots a_m } \ms{A}^{ b_1 \dots b_k }_{ c_1 \dots c_l } - \mc{R}_m ( [ \ms{A} ]_{ m - 1, \varphi }, [ \gv ]_{ m, \varphi }, [ \gv^{-1} ]_{ m - 1, \varphi } ) \text{,} \\
\notag \Dm_{ a_1 \dots a_m } \mf{A}^{ b_1 \dots b_k }_{ c_1 \dots c_l } &= \partial^m_{ a_1 \dots a_m } \mf{A}^{ b_1 \dots b_k }_{ c_1 \dots c_l } - \mc{R}_m ( [ \mf{A} ]_{ m - 1, \varphi }, [ \gm ]_{ m, \varphi }, [ \gm^{-1} ]_{ m - 1, \varphi } ) \text{.}
\end{align}
Moreover, this $\mc{R}_m$ is linear with respect to its first set (i.e., $[ \ms{A} ]_{ m - 1, \varphi }$ and $[ \mf{A} ]_{ m - 1, \varphi }$) of components.
Thus, by Proposition \ref{thm.geom_bound}, the difference of the two $\mc{R}_m$-terms in \eqref{eql.geom_bound_deriv_0} satisfies
\begin{align}
\label{eql.geom_bound_deriv_1} &| \mc{R}_m ( [ \ms{A} ]_{ m - 1, \varphi }, [ \gv ]_{ m, \varphi }, [ \gv^{-1} ]_{ m - 1, \varphi } ) - \mc{R}_m ( [ \mf{A} ]_{ m - 1, \varphi }, [ \gm ]_{ m, \varphi }, [ \gm^{-1} ]_{ m - 1, \varphi } ) | \\
\notag &\quad \lesssim | \ms{A} - \mf{A} |_{ m - 1, \varphi } + ( | \ms{A} |_{ m - 1, \varphi } + | \mf{A} |_{ m - 1, \varphi } ) | \gv - \gm |_{ m, \varphi } \text{.}
\end{align}

Consider now the difference of the two identities in \eqref{eql.geom_bound_deriv_0}.
First, recalling \eqref{eql.geom_bound_deriv_1}, we see that
\[
| \Dv_{ a_1 \dots a_m } \ms{A}^{ b_1 \dots b_k }_{ c_1 \dots c_l } - \Dm_{ a_1 \dots a_m } \mf{A}^{ b_1 \dots b_k }_{ c_1 \dots c_l } | \lesssim | \ms{A} - \mf{A} |_{ m, \varphi } + ( | \ms{A} |_{ m - 1, \varphi } + | \mf{A} |_{ m - 1, \varphi } ) | \gv - \gm |_{ m, \varphi } \text{.}
\]
Summing the above over all indices and over $0 < m \leq M$ yields the first bound of \eqref{eq.geom_bound_deriv}.
Similarly, shuffling the terms in \eqref{eql.geom_bound_deriv_0} around, we can once again use \eqref{eql.geom_bound_deriv_1} to estimate
\begin{align*}
| \partial^m_{ a_1 \dots a_m } ( \ms{A}^{ b_1 \dots b_k }_{ c_1 \dots c_l } - \mf{A}^{ b_1 \dots b_k }_{ c_1 \dots c_l } ) | &\lesssim | \Dv^m \ms{A} - \Dm^m \mf{A} |_{ 0, \varphi } + | \ms{A} - \mf{A} |_{ m - 1, \varphi } \\
&\qquad + ( | \ms{A} |_{ m - 1, \varphi } + | \mf{A} |_{ m - 1, \varphi } ) | \gv - \gm |_{ m, \varphi } \text{.}
\end{align*}
Summing the above over all indices and then adding $| \ms{A} - \mf{A} |_{ m - 1, \varphi }$ to both sides, we obtain
\[
| \ms{A} - \mf{A} |_{ m, \varphi } \lesssim | \Dv^m \ms{A} - \Dm^m \mf{A} |_{ 0, \varphi } + | \ms{A} - \mf{A} |_{ m - 1, \varphi } + ( | \ms{A} |_{ m - 1, \varphi } + | \mf{A} |_{ m - 1, \varphi } ) | \gv - \gm |_{ m, \varphi } \text{.}
\]
Combining the above with an induction argument over $m$ yields the second bound in \eqref{eq.geom_bound_deriv}.
\end{proof}

\subsubsection{Boundary Limits}

Next, we prove a number of general boundary limit properties for vertical tensor fields.
We begin by stating some trivial but useful properties:

\begin{proposition} \label{thm.limit_bound}
Let $( \mi{M}, g )$ be an FG-aAdS segment, and let $M \geq 0$.
If a vertical tensor field $\ms{A}$ on $\mi{M}$ is locally bounded in $C^M$, then $\rho^p \ms{A} \Rightarrow^M 0$ for any $p > 0$.
\end{proposition}

\if\comp1
\begin{proof}
See Appendix \ref{sec.comp_limit_bound}.
\end{proof}
\fi

\begin{proposition} \label{thm.limit_schematic}
Let $( \mi{M}, g )$ be an FG-aAdS segment, and fix $M \geq 0$ and $N \geq 1$.
Moreover, let $\ms{A}_1, \dots, \ms{A}_N$ denote vertical tensor fields, and let $\mf{A}_1, \dots, \mf{A}_N$ denote tensor fields on $\mi{I}$.
Then:
\begin{itemize}
\item If $\ms{A}_j \rightarrow^M \mf{A}_j$ for each $1 \leq j \leq N$, then $\mi{S} ( \ms{A}_1, \dots, \ms{A}_N ) \rightarrow^M \mi{S} ( \mf{A}_1, \dots, \mf{A}_N )$.

\item If $\ms{A}_j \Rightarrow^M \mf{A}_j$ for each $1 \leq j \leq N$, then $\mi{S} ( \ms{A}_1, \dots, \ms{A}_N ) \Rightarrow^M \mi{S} ( \mf{A}_1, \dots, \mf{A}_N )$.
\end{itemize}
\end{proposition}

\if\comp1
\begin{proof}
See Appendix \ref{sec.comp_limit_schematic}.
\end{proof}
\fi

Next, we derive boundary limits for vertical tensor fields satisfying differential equations.

\begin{proposition} \label{thm.limit_positive}
Let $( \mi{M}, g )$ be an FG-aAdS segment, and let $M \geq 0$.
Furthermore, let $\ms{A}$, $\ms{G}$ be vertical tensor fields on $\mi{M}$, and suppose they satisfy the following relation:
\begin{equation}
\label{eq.limit_positive_eq} \rho \mi{L}_\rho \ms{A} - c \ms{A} = \ms{G} \text{,} \qquad c > 0 \text{.}
\end{equation}
\begin{itemize}
\item For any $\rho_\ast \in ( 0, \rho_0 ]$ and any compact coordinate system $( U, \varphi )$ on $\mi{I}$, we have that
\begin{equation}
\label{eq.limit_positive_bound} | \ms{A} |_{ M, \varphi } \leq \frac{ \rho^c }{ \rho_\ast^c } \cdot | \ms{A} |_{ M, \varphi } |_{ \rho_\ast } + \rho^c \int_{ \min ( \rho, \rho_\ast ) }^{ \max ( \rho, \rho_\ast ) } \sigma^{ - c - 1 } | \ms{G} |_{ M, \varphi } |_\sigma d \sigma \text{.}
\end{equation}

\item If $\ms{G} \rightarrow^M \mf{G}$, where $\mf{G}$ is a tensor field on $\mi{I}$, then
\begin{equation}
\label{eq.limit_positive} - c \ms{A} \rightarrow^M \mf{G} \text{,} \qquad \rho \mi{L}_\rho \ms{A} \rightarrow^M 0 \text{.}
\end{equation}

\item If $\ms{G} \Rightarrow^M \mf{G}$, where $\mf{G}$ is a tensor field on $\mi{I}$, then
\begin{equation}
\label{eq.limit_positive_ex} - c \ms{A} \Rightarrow^M \mf{G} \text{,} \qquad \rho \mi{L}_\rho \ms{A} \Rightarrow^M 0 \text{.}
\end{equation}
\end{itemize}
\end{proposition}

\begin{proof}
Let $( k, l )$ be the ranks of both $\ms{A}$ and $\ms{G}$.
Moreover, we fix an arbitrary compact coordinate system $( U, \varphi )$ on $\mi{I}$, and we always index with respect to $\varphi$-coordinates.

First, observe that \eqref{eq.limit_positive_eq} can be written as
\[
\mi{L}_\rho ( \rho^{ - c } \ms{A} ) = \rho^{ - c - 1 } \ms{G} \text{.}
\]
Fixing $0 < \rho_\ast \leq \rho_0$ and integrating the above from $\rho = \rho_\ast$ yields
\[
\ms{A}^{ b_1 \dots b_k }_{ c_1 \dots c_l } = \frac{ \rho^c }{ \rho_\ast^c } \cdot \ms{A}^{ b_1 \dots b_k }_{ c_1 \dots c_l } |_{ \rho_\ast } + \rho^c \int_{ \rho_\ast }^\rho \sigma^{ - c - 1 } \ms{G}^{ b_1 \dots b_k }_{ c_1 \dots c_l } |_\sigma d \sigma \text{.}
\]
We now differentiate the above along $\mi{I}$, and we apply the mean value and dominated convergence theorems to move the derivatives under the integral.
This yields, for all $0 \leq m \leq M$,
\begin{equation}
\label{eql.limit_positive_0} \partial^m_{ a_1 \dots a_m } \ms{A}^{ b_1 \dots b_k }_{ c_1 \dots c_l } = \frac{ \rho^c }{ \rho_\ast^c } \cdot \partial^m_{ a_1 \dots a_m } \ms{A}^{ b_1 \dots b_k }_{ c_1 \dots c_l } |_{ \rho_\ast } + \rho^c \int_{ \rho_\ast }^\rho \sigma^{ - c - 1 } \partial^m_{ a_1 \dots a_m } \ms{G}^{ b_1 \dots b_k }_{ c_1 \dots c_l } |_\sigma d \sigma \text{.}
\end{equation}
Taking absolute values of both sides of \eqref{eql.limit_positive_0} and summing appropriately, we obtain \eqref{eq.limit_positive_bound}.

Suppose now that $\ms{G} \rightarrow^M \mf{G}$.
Since $\ms{G}$ converges and $( U, \varphi )$ is compact, it follows that
\begin{equation}
\label{eql.limit_positive_1} \| \ms{A} |_{ \rho_0 } \|_{ M, \varphi } < \infty \text{,} \qquad \| \ms{G} \|_{ M, \varphi } < \infty \text{.}
\end{equation}
Applying \eqref{eql.limit_positive_0}, with $\rho_\ast = \rho_0$, results in the bound
\[
| \partial^m_{ a_1 \dots a_m } \ms{A}^{ b_1 \dots b_k }_{ c_1 \dots c_l } | = \frac{ \rho^c }{ \rho_0^c } \cdot | \partial^m_{ a_1 \dots a_m } \ms{A}^{ b_1 \dots b_k }_{ c_1 \dots c_l } |_{ \rho_0 } | + \rho^c \left| \int_{ \rho_0 }^\rho \sigma^{ - c - 1 } d \sigma \right| \| \ms{G} \|_{ M, \varphi } \text{.}
\]
Summing the above over $m$ and over all the indices, we obtain that
\begin{equation}
\label{eql.limit_positive_2} \| \ms{A} \|_{ M, \varphi } \lesssim \| \ms{A} |_{ \rho_0 } \|_{ M, \varphi } + \frac{1}{c} \left| 1 - \frac{ \rho^c }{ \rho_0^c } \right| \| \ms{G} \|_{ M, \varphi } < \infty \text{.}
\end{equation}

Let $\rho_1 > 0$ be small enough so that $\rho_1^2 \leq \rho_1 \leq \rho_0$.
Applying \eqref{eql.limit_positive_0} again, at $\rho = \rho_1^2 < 0$ and with $\rho_\ast := \rho_1 < \rho_0$, we then obtain, for any $0 \leq m \leq M$, that
\begin{align}
\label{eql.limit_positive_10} \partial^m_{ a_1 \dots a_m } \ms{A}^{ b_1 \dots b_k }_{ c_1 \dots c_l } |_{ \rho_1^2 } &= \rho_1^c \cdot \partial^m_{ a_1 \dots a_m } \ms{A}^{ b_1 \dots b_k }_{ c_1 \dots c_l } |_{ \rho_1 } - \frac{1}{c} ( 1 - \rho_1^c ) \cdot \partial^m_{ a_1 \dots a_m } \mf{G}^{ b_1 \dots b_k }_{ c_1 \dots c_l } \\
\notag &\qquad + \rho_1^{ 2 c } \int_{ \rho_1 }^{ \rho_1^2 } \sigma^{ - c - 1 } ( \partial^m_{ a_1 \dots a_m } \ms{G}^{ b_1 \dots b_k }_{ c_1 \dots c_l } - \partial^m_{ a_1 \dots a_m } \mf{G}^{ b_1 \dots b_k }_{ c_1 \dots c_l } ) |_\sigma d \sigma \text{.}
\end{align}
Letting $I$ denote the last term on the right-hand side of \eqref{eql.limit_positive_10}, we see that
\[
| I | \leq \left| \rho_1^{ 2 c } \int_{ \rho_1 }^{ \rho_1^2 } \sigma^{ - c - 1 } d \sigma \right| \sup_{ ( 0, \rho_1 ] \times U } | \ms{G} - \mf{G} |_{ M, \varphi } \leq \frac{1}{c} \sup_{ ( 0, \rho_1 ] \times U } | \ms{G} - \mf{G} |_{ M, \varphi } \text{,}
\]
which, when combined with \eqref{eql.limit_positive_10}, implies
\[
\left| \left. \left( \ms{A} + \frac{1}{c} \mf{G} \right) \right|_{ \rho_1^2 } \right|_{ M, \varphi } \lesssim \rho_1^c \| \ms{A} \|_{ M, \varphi } + \frac{ \rho_1^c }{c} \| \mf{G} \|_{ M, \varphi } + \frac{1}{c} \sup_{ ( 0, \rho_1 ] \times U } | \ms{G} - \mf{G} |_{ M, \varphi } \text{.}
\]

Letting $\rho_1 \searrow 0$ in the above results in the first limit of \eqref{eq.limit_positive}; this follows immediately from \eqref{eql.limit_positive_2}, the assumption $\ms{G} \rightarrow^M \mf{G}$, and that $( U, \varphi)$ is a compact coordinate system.
The remaining limit in \eqref{eq.limit_positive} then follows from the above and from \eqref{eq.limit_positive_eq}.

Finally, assume in addition that $\ms{G} \Rightarrow^M \mf{G}$.
We rewrite the identity \eqref{eql.limit_positive_0}, with $\rho_\ast = \rho_0$, as
\begin{align*}
\partial^m_{ a_1 \dots a_m } \left( \ms{A}^{ b_1 \dots b_k }_{ c_1 \dots c_l } + \frac{1}{c} \mf{G}^{ b_1 \dots b_k }_{ c_1 \dots c_l } \right) &= \frac{ \rho^c }{ \rho_0^c } \cdot \partial^m_{ a_1 \dots a_m } \ms{A}^{ b_1 \dots b_k }_{ c_1 \dots c_l } |_{ \rho_0 } + \rho^c \int_{ \rho_0 }^\rho \sigma^{ - c - 1 } \partial^m_{ a_1 \dots a_m } \ms{G}^{ b_1 \dots b_k }_{ c_1 \dots c_l } |_\sigma d \sigma \\
&\qquad + \frac{1}{c} \cdot \frac{ \rho^c }{ \rho_0^c } \cdot \partial^m_{ a_1 \dots a_m } \mf{G}^{ b_1 \dots b_k }_{ c_1 \dots c_l } + \frac{1}{c} \left( 1 - \frac{ \rho^c }{ \rho_0^c } \right) \partial^m_{ a_1 \dots a_m } \mf{G}^{ b_1 \dots b_k }_{ c_1 \dots c_l } \\
&= \frac{ \rho^c }{ \rho_0^c } \cdot \partial^m_{ a_1 \dots a_m } \ms{A}^{ b_1 \dots b_k }_{ c_1 \dots c_l } |_{ \rho_0 } + \frac{1}{c} \cdot \frac{ \rho^c }{ \rho_0^c } \cdot \partial^m_{ a_1 \dots a_m } \mf{G}^{ b_1 \dots b_k }_{ c_1 \dots c_l } \\
&\qquad + \rho^c \int_{ \rho_0 }^\rho \sigma^{ - c - 1 } ( \partial^m_{ a_1 \dots a_m } \ms{G}^{ b_1 \dots b_k }_{ c_1 \dots c_l } - \partial^m_{ a_1 \dots a_m } \mf{G}^{ b_1 \dots b_k }_{ c_1 \dots c_l } ) |_\sigma d \sigma \text{.}
\end{align*}
Dividing the above by $\rho$ and integrating its absolute value over $( 0, \rho_0 ]$, we obtain
\begin{align*}
\int_0^{ \rho_0 } \frac{1}{ \sigma } \left| \ms{A} + \frac{1}{c} \mf{G} \right|_{ M, \varphi } |_\sigma d \sigma &\lesssim \int_0^{ \rho_0 } \tau^{ c - 1 } d \tau \cdot ( | \ms{A} |_{ \rho_0 } |_{ M, \varphi } + | \mf{G} |_{ M, \varphi } ) \\
&\qquad + \int_0^{ \rho_0 } \tau^{ c - 1 } \int_\tau^{ \rho_0 } \sigma^{ - c - 1 } | \ms{G} - \mf{G} |_{ M, \varphi } |_\sigma d \sigma d \tau \text{.}
\end{align*}
The second term on the right-hand side can be simplified using Fubini's theorem, hence we have
\[
\int_0^{ \rho_0 } \frac{1}{ \sigma } \left| \ms{A} + \frac{1}{c} \mf{G} \right|_{ M, \varphi } |_\sigma d \sigma \lesssim \frac{1}{c} ( | \ms{A} |_{ \rho_0 } |_{ M, \varphi } + | \mf{G} |_{ M, \varphi } ) + \int_0^{ \rho_0 } \sigma^{ - 1 } | \ms{G} - \mf{G} |_{ M, \varphi } |_\sigma d \sigma \text{.}
\]
Note the right-hand side of the above is indeed finite, since $\ms{G} \Rightarrow^M \mf{G}$ and $( U, \varphi )$ is compact, and thus we conclude $-c \ms{A} \Rightarrow^M \mf{G}$.
The remaining limit $\rho \mi{L}_\rho \ms{A} \Rightarrow^M 0$ now follows from \eqref{eq.limit_positive_eq}.
\end{proof}

\begin{proposition} \label{thm.limit_zero}
Let $( \mi{M}, g )$ be an FG-aAdS segment, and let $M \geq 0$.
In addition, let $\ms{A}$ be a vertical tensor field on $\mi{M}$, and suppose that $\rho \mi{L}_\rho \ms{A} \Rightarrow^M \mf{G}$ for some tensor field $\mf{G}$ on $\mi{I}$.
Then, then there exists another tensor field $\mf{H}$ on $\mi{I}$ such that
\begin{equation}
\label{eq.limit_zero} \ms{A} - ( \log \rho ) \cdot \mf{G} \rightarrow^M \mf{H} \text{.}
\end{equation}
\end{proposition}

\begin{proof}
Let $( k, l )$ be the rank of $\ms{A}$, and let $( U, \varphi )$ be a compact coordinate system on $\mi{I}$.
Applying the fundamental theorem of calculus, we obtain (with respect to $\varphi$-coordinates)
\[
\ms{A}^{ b_1 \dots b_k }_{ c_1 \dots c_l } = \ms{A}^{ b_1 \dots b_k }_{ c_1 \dots c_l } |_{ \rho_0 } + \int_{ \rho_0 }^\rho \sigma^{ - 1 } \cdot \sigma \mi{L}_\rho \ms{A}^{ b_1 \dots b_k }_{ c_1 \dots c_l } |_\sigma \cdot d \sigma \text{.}
\]
Differentiating the above yields, for any $0 \leq m \leq M$,
\begin{align*}
\partial^m_{ a_1 \dots a_m } \ms{A}^{ b_1 \dots b_k }_{ c_1 \dots c_l } &= \partial^m_{ a_1 \dots a_m } \ms{A}^{ b_1 \dots b_k }_{ c_1 \dots c_l } |_{ \rho_0 } + \int_{ \rho_0 }^\rho \sigma^{ - 1 } \cdot \partial^m_{ a_1 \dots a_m } ( \sigma \mi{L}_\rho \ms{A}^{ b_1 \dots b_k }_{ c_1 \dots c_l } ) |_\sigma \cdot d \sigma \\
&= \partial^m_{ a_1 \dots a_m } \ms{A}^{ b_1 \dots b_k }_{ c_1 \dots c_l } |_{ \rho_0 } + \int_{ \rho_0 }^\rho \sigma^{ - 1 } d \sigma \cdot \partial^m_{ a_1 \dots a_m } \mf{G}^{ b_1 \dots b_k }_{ c_1 \dots c_l } \\
&\qquad + \int_{ \rho_0 }^\rho \sigma^{ - 1 } \cdot \partial^m_{ a_1 \dots a_m } ( \sigma \mi{L}_\rho \ms{A}^{ b_1 \dots b_k }_{ c_1 \dots c_l } - \mf{G}^{ b_1 \dots b_k }_{ c_1 \dots c_l } ) |_\sigma d \sigma \text{,}
\end{align*}
Rearranging the above, we then obtain
\begin{align}
\label{eql.limit_zero_1} \partial^m_{ a_1 \dots a_m } [ \ms{A}^{ b_1 \dots b_k }_{ c_1 \dots c_l } - ( \log \rho ) \mf{G}^{ b_1 \dots b_k }_{ c_1 \dots c_l } ] &= \partial^m_{ a_1 \dots a_m } \ms{A}^{ b_1 \dots b_k }_{ c_1 \dots c_l } |_{ \rho_0 } - ( \log \rho_0 ) \cdot \partial^m_{ a_1 \dots a_m } \mf{G}^{ b_1 \dots b_k }_{ c_1 \dots c_l } \\
\notag &\qquad + \int_{ \rho_0 }^\rho \sigma^{ - 1 } \cdot \partial^m_{ a_1 \dots a_m } ( \sigma \mi{L}_\rho \ms{A}^{ b_1 \dots b_k }_{ c_1 \dots c_l } - \mf{G}^{ b_1 \dots b_k }_{ c_1 \dots c_l } ) |_\sigma d \sigma \text{.}
\end{align}
The first two terms on the right-hand side of \eqref{eql.limit_zero_1} are constant in $\rho$ and hence do not change as $\rho \searrow 0$.
Moreover, notice that for the last term in \eqref{eql.limit_zero_1}, we have
\begin{equation}
\label{eql.limit_zero_2} \int_{ \rho_0 }^\rho \sigma^{ - 1 } \left| \partial^m_{ a_1 \dots a_m } ( \sigma \mi{L}_\rho \ms{A}^{ b_1 \dots b_k }_{ c_1 \dots c_l } - \mf{G}^{ b_1 \dots b_k }_{ c_1 \dots c_l } ) |_\sigma \right| d \sigma \lesssim \int_0^{ \rho_0 } \frac{1}{ \sigma } | \rho \mi{L}_\rho \ms{A} - \mf{G} |_{ M, \varphi } |_\sigma d \sigma < \infty \text{.}
\end{equation}

Define now the tensor field $\mf{H}$ on $\mi{I}$ by
\begin{equation}
\label{eql.limit_zero_3} \mf{H}^{ b_1 \dots b_k }_{ c_1 \dots c_l } := \ms{A}^{ b_1 \dots b_k }_{ c_1 \dots c_l } |_{ \rho_0 } - ( \log \rho_0 ) \cdot \mf{G}^{ b_1 \dots b_k }_{ c_1 \dots c_l } + \int_{ \rho_0 }^0 \sigma^{ - 1 } ( \sigma \mi{L}_\rho \ms{A}^{ b_1 \dots b_k }_{ c_1 \dots c_l } - \mf{G}^{ b_1 \dots b_k }_{ c_1 \dots c_l } ) |_\sigma d \sigma \text{.}
\end{equation}
We now differentiate \eqref{eql.limit_zero_3} along $\mi{I}$.
Using the dominated convergence theorem and the assumption $\rho \mi{L}_\rho \ms{A} \Rightarrow^M \mf{G}$, we see these derivatives can be moved inside the integral in the last term of \eqref{eql.limit_zero_3}.
Therefore, combining \eqref{eql.limit_zero_1}, \eqref{eql.limit_zero_2}, and \eqref{eql.limit_zero_3}, we conclude
\[
\lim_{ \sigma \searrow 0 } \sup_{ \{ \sigma \} \times U } | \ms{A} - ( \log \rho ) \cdot \mf{G} - \mf{H} |_{ M, \varphi } \lesssim \lim_{ \rho_0 \searrow 0 } \sup_U \int_0^{ \rho_0 } \frac{1}{ \sigma } | \rho \mi{L}_\rho \ms{A} - \mf{G} |_{ M, \varphi } |_\sigma d \sigma = 0 \text{,}
\]
which completes the proof of \eqref{eq.limit_zero}.
\end{proof}

\section{Fefferman--Graham Boundary Expansions} \label{sec.fg}

In this section, we give (in Theorem \ref{thm.fg_main}) the precise statement of our main result, which characterizes the near-boundary geometry of vacuum FG-aAdS segments.
We also describe several immediate corollaries of Theorem \ref{thm.fg_main}, including the following:
\begin{itemize}
\item (Theorem \ref{thm.fg_exp}) A partial Fefferman--Graham expansion of the spacetime metric, up to $n$-th order (which, in particular, includes both undetermined parts of the metric boundary data).

\item (Corollary \ref{thm.fg_exp_W}) Similar partial expansions for the spacetime Weyl curvature.
\end{itemize}
In Section \ref{sec.fg_proof}, we give the detailed proof of Theorem \ref{thm.fg_main}.
Finally, in Section \ref{sec.fg_schw_ads}, we derive this partial expansion for a familiar family of manifolds: Schwarzschild--AdS spacetimes.

\subsection{Statements of Results} \label{sec.fg_result}

The main objective of this subsection is to state the main result of this paper.
In addition, we establish a number of its immediate consequences; in Section \ref{sec.fg_exp}, we derive partial Fefferman--Graham expansions for various geometric quantities.

\subsubsection{The Main Result}

One major feature of our main result is that the boundary limits of lower ($\rho$-)derivatives of the vertical metric depend only on the boundary metric and some finite number of its derivatives.
The following definition gives a precise characterization of this dependence:

\begin{definition} \label{def.fg_depend}
Let $( \mi{M}, g )$ be an FG-aAdS segment, let $M \geq 0$, and let $\mf{A}$ be a tensor field on $\mi{I}$.
We say $\mf{A}$ \emph{depends only on $\gm$ to order $M$}, abbreviated $\mf{A} = \mc{D}^M ( \gm )$, iff $\mf{A}$ can be written as
\begin{equation}
\label{eq.fg_depend} \mf{A} = \sum_{ j = 1 }^N \mi{S} ( \mf{A}^j_1, \dots, \mf{A}^j_{ N_j } ) \text{,} \qquad N \geq 0 \text{,}
\end{equation}
where for each $1 \leq j \leq N$, we have that $N_j \geq 1$, and that
\begin{equation}
\label{eq.fg_depend_comp} \mf{A}^j_i \in \begin{cases} \{ \gm, \gm^{-1}, \Rm, \dots, \Dm^{ M - 2 } \Rm \} & M \geq 2 \text{,} \\ \{ \gm, \gm^{-1} \} & M < 2 \text{,} \end{cases} \text{,} \qquad 1 \leq i \leq N_j \text{.}
\end{equation}
\end{definition}

Furthermore, another feature of the boundary limits in our main result is that they vanish for all odd orders (at least, up to order $n$).
Therefore, in order to simplify the technical writing, we define the following notational shorthand to capture this property:

\begin{definition} \label{def.fg_depend_even}
Let $( \mi{M}, g )$, $M$, $\mf{A}$ be as in Definition \ref{def.fg_depend}.
For $k \in \Z$, we write $\mf{A} = \mc{D}^M ( \gm; k )$ iff:
\begin{itemize}
\item $\mf{A} = \mc{D}^M ( \gm )$ whenever $k$ is even.

\item $\mf{A} = 0$ whenever $k$ is odd.
\end{itemize}
\end{definition}

With the above, we can now give a precise statement of the main result:

\begin{theorem} \label{thm.fg_main}
Let $( \mi{M}, g )$ be a vacuum FG-aAdS segment, and fix $M_0 \geq n + 2$.
Furthermore, assume $n = \dim \mi{I} > 1$, and suppose that the following conditions hold:
\begin{itemize}
\item $\gv$ is locally bounded in $C^{ M_0 + 2 }$---that is, for any compact coordinate system $( U, \varphi )$ on $\mi{I}$,
\begin{equation}
\label{eq.fg_main_ass_g} \| \gv \|_{ M_0 + 2, \varphi } < \infty \text{.}
\end{equation}

\item $\mi{L}_\rho \gv$ is ``weakly locally bounded"---for any compact coordinate system $( U, \varphi )$ on $\mi{I}$,
\begin{equation}
\label{eq.fg_main_ass_Lg} \sup_U \int_0^{ \rho_0 } | \mi{L}_\rho \gv |_{ 0, \varphi } |_\sigma d \sigma < \infty \text{.}
\end{equation}
\end{itemize}
Then, the following statements hold:
\begin{itemize}
\item $\gv$ and $\gv^{-1}$ satisfy the following boundary limits:
\begin{equation}
\label{eq.fg_main_low_g} \gv \Rightarrow^{ M_0 } \gm \text{,} \qquad \gv^{-1} \Rightarrow^{ M_0 } \gm^{-1} \text{.}
\end{equation}
Furthermore, $\Rv$, $\Rcv$, and $\Rsv$ satisfy the following boundary limits:
\begin{equation}
\label{eq.fg_main_low_R} \Rv \Rightarrow^{ M_0 - 2 } \Rm \text{,} \qquad \Rcv \Rightarrow^{ M_0 - 2 } \Rcm \text{,} \qquad \Rsv \Rightarrow^{ M_0 - 2 } \Rsm \text{.}
\end{equation}

\item For any $0 \leq k < n$ and $1 \leq l < n$, there exist tensor fields
\begin{equation}
\label{eq.fg_main_high_depend} \gb{k} = \mc{D}^k ( \gm; k ) \text{,} \qquad \Rb{k} = \mc{D}^{ k + 2 } ( \gm; k ) \text{,} \qquad \bb{l} = \mc{D}^{ l + 1 } ( \gm; l )
\end{equation}
on $\mi{I}$, such that the following limits hold for each $0 \leq k < n$ and $1 \leq l < n$:
\begin{align}
\label{eq.fg_main_high} \mi{L}_\rho^k \gv \Rightarrow^{ M_0 - k } k! \, \gb{k} \text{,} &\qquad \rho \mi{L}_\rho^{ k + 1 } \gv \Rightarrow^{ M_0 - k } 0 \text{,} \\
\notag \mi{L}_\rho^k \Rv \Rightarrow^{ M_0 - k - 2 } k! \, \Rb{k} \text{,} &\qquad \rho \mi{L}_\rho^{ k + 1 } \Rv \Rightarrow^{ M_0 - k - 2 } 0 \text{,} \\
\notag \mi{L}_\rho^{ l - 1 } ( \Dv \mi{L}_\rho \gv ) \Rightarrow^{ M_0 - l - 1 } ( l - 1 )! \, \bb{l} \text{,} &\qquad \rho \mi{L}_\rho^l ( \Dv \mi{L}_\rho \gv ) \Rightarrow^{ M_0 - l - 1 } 0 \text{.}
\end{align}

\item There exist tensor fields
\begin{equation}
\label{eq.fg_main_top_depend} \gb{\star} = \mc{D}^n ( \gm; n ) \text{,} \qquad \Rb{\star} = \mc{D}^{ n + 2 } ( \gm; n ) \text{,} \qquad \bb{\star} = \mc{D}^{ n + 1 } ( \gm; n )
\end{equation}
on $\mi{I}$, such that the following boundary limits hold:
\begin{align}
\label{eq.fg_main_top} \rho \mi{L}_\rho^{ n + 1 } \gv &\Rightarrow^{ M_0 - n } n! \, \gb{\star} \text{,} \\
\notag \rho \mi{L}_\rho^{ n + 1 } \Rv &\Rightarrow^{ M_0 - n - 2 } n! \, \Rb{\star} \text{,} \\
\notag \rho \mi{L}_\rho^n ( \Dv \mi{L}_\rho \gv ) &\Rightarrow^{ M_0 - n - 1 } ( n - 1 )! \, \bb{\star} \text{.}
\end{align}
In general, $\gb{\star}$ satisfies the following constraints:
\begin{equation}
\label{eq.fg_main_top_constraint} \trace{\gm} \gb{\star} = 0 \text{,} \qquad \Dm \cdot \gb{\star} = 0 \text{.}
\end{equation}

\item Furthermore, the following relations hold:
\begin{equation}
\label{eq.fg_main_schouten} \begin{cases} \gb{2} = - \frac{ 1 }{ n - 2 } \left[ \Rcm - \frac{ 1 }{ 2 ( n - 1 ) } \Rsm \cdot \gm \right] & n > 2 \text{,} \\ ( \gb{\star}, \Rb{\star}, \bb{\star} ) = ( 0, 0, 0 ) & n = 2 \text{.} \end{cases}
\end{equation}

\item There exist tensor fields $\gb{\dagger}$, $\Rb{\dagger}$, $\bb{\dagger}$ on $\mi{I}$ such that
\begin{align}
\label{eq.fg_main_free} \mi{L}_\rho^n \gv - n! \, ( \log \rho ) \gb{\star} &\rightarrow^{ M_0 - n } n! \, \gb{\dagger} \text{,} \\
\notag \mi{L}_\rho^n \Rv - n! \, ( \log \rho ) \Rb{\star} &\rightarrow^{ M_0 - n - 2 } n! \, \Rb{\dagger} \text{,} \\
\notag \mi{L}_\rho^{ n - 1 } ( \Dv \mi{L}_\rho \gv ) - ( n - 1 )! \, ( \log \rho ) \bb{\star} &\rightarrow^{ M_0 - n - 1 } ( n - 1 )! \, \bb{\dagger} \text{.}
\end{align}
In addition, $\gb{\dagger}$ satisfies the following constraints:
\begin{equation}
\label{eq.fg_main_free_constraint} \trace{\gm} \gb{\dagger} = \mc{D}^n ( \gm; n ) \text{,} \qquad \Dm \cdot \gb{\dagger} = \mc{D}^{ n + 1 } ( \gm; n ) \text{.}
\end{equation}
\end{itemize}
\end{theorem}

\begin{remark}
In particular, the first part of \eqref{eq.fg_main_schouten} implies that when $n > 2$, the limit $\gb{2}$ is precisely the Schouten tensor associated with the boundary metric $\gm$.
\end{remark}

We defer the proof of Theorem \ref{thm.fg_main} until Section \ref{sec.fg_proof}.
In the remainder of the present subsection, we discuss some variations and consequences of Theorem \ref{thm.fg_main}.

First, combining Proposition \ref{thm.geom_bound_deriv} with the conclusions of Theorem \ref{thm.fg_main}, we see that the boundary limits \eqref{eq.fg_main_low_R}, \eqref{eq.fg_main_high}, \eqref{eq.fg_main_top}, and \eqref{eq.fg_main_free} can also be expressed covariantly:

\begin{corollary} \label{thm.fg_covar}
Assume the hypotheses of Theorem \ref{thm.fg_main}, and let
\[
\gb{0}, \dots, \gb{n - 1}, \gb{\star}, \gb{\dagger} \text{,} \qquad \Rb{0}, \dots, \Rb{n - 1 }, \Rb{\star}, \Rb{\dagger} \text{,} \qquad \bb{1}, \dots, \bb{n - 1}, \bb{\star}, \bb{\dagger} \text{,}
\]
be as in the conclusions of Theorem \ref{thm.fg_main}.
Then:
\begin{itemize}
\item The following boundary limits hold for any $0 \leq v \leq M_0 - 2$:
\begin{equation}
\label{eq.fg_covar_low} \Dv^v \Rv \Rightarrow^0 \Dm^v \Rm \text{,} \qquad \Dv^v \Rcv \Rightarrow^0 \Dm^v \Rcm \text{,} \qquad \Dv^v \Rsv \Rightarrow^0 \Dm^v \Rsm \text{.}
\end{equation}

\item For any $0 \leq k < n$ and $1 \leq l < n$, we have the following boundary limits:
\begin{align}
\label{eq.fg_covar_high} \Dv^m \mi{L}_\rho^k \gv \Rightarrow^0 k! \, \Dm^m \gb{k} \text{,} \qquad \rho \Dv^m \mi{L}_\rho^{ k + 1 } \gv &\Rightarrow^0 0 \text{,} \qquad 0 \leq m \leq M_0 - k \text{,} \\
\notag \Dv^u \mi{L}_\rho^k \Rv \Rightarrow^0 k! \, \Dm^u \Rb{k} \text{,} \qquad \rho \Dv^u \mi{L}_\rho^{ k + 1 } \Rv &\Rightarrow^0 0 \text{,} \qquad 0 \leq u \leq M_0 - k - 2 \text{,} \\
\notag \Dv^v \mi{L}_\rho^{ l - 1 } ( \Dv \mi{L}_\rho \gv ) \Rightarrow^0 ( l - 1 )! \, \Dm^v \bb{l} \text{,} \qquad \rho \Dv^v \mi{L}_\rho^l ( \Dv \mi{L}_\rho \gv ) &\Rightarrow^0 0 \text{,} \qquad 0 \leq v \leq M_0 - l - 1 \text{.}
\end{align}

\item The following boundary limits hold:
\begin{align}
\label{eq.fg_covar_top} \rho \Dv^m \mi{L}_\rho^{ n + 1 } \gv \Rightarrow^0 n! \, \Dm^m \gb{\star} \text{,} &\qquad 0 \leq m \leq M_0 - n \text{,} \\
\notag \rho \Dv^u \mi{L}_\rho^{ n + 1 } \Rv \Rightarrow^0 n! \, \Dm^u \Rb{\star} \text{,} &\qquad 0 \leq u \leq M_0 - n - 2 \text{,} \\
\notag \rho \Dv^v \mi{L}_\rho^n ( \Dv \mi{L}_\rho \gv ) \Rightarrow^0 ( n - 1 )! \, \Dm^v \bb{\star} \text{,} &\qquad 0 \leq v \leq M_0 - n - 1 \text{.}
\end{align}

\item The following boundary limits hold:
\begin{align}
\label{eq.fg_covar_free} \Dv^m \mi{L}_\rho^n \gv - n! \, ( \log \rho ) \Dm^m \gb{\star} \rightarrow^0 n! \, \Dm^m \gb{\dagger} \text{,} &\qquad 0 \leq m \leq M_0 - n \text{,} \\
\notag \Dv^u \mi{L}_\rho^n \Rv - n! \, ( \log \rho ) \Dm^u \Rb{\star} \rightarrow^0 n! \, \Dm^u \Rb{\dagger} \text{,} &\qquad 0 \leq u \leq M_0 - n - 2 \text{,} \\
\notag \Dv^v \mi{L}_\rho^{ n - 1 } ( \Dv \mi{L}_\rho \gv ) - ( n - 1 )! \, ( \log \rho ) \Dm^v \bb{\star} \rightarrow^0 ( n - 1 )! \, \Dm^v \bb{\dagger} \text{,} &\qquad 0 \leq v \leq M_0 - n - 1 \text{.}
\end{align}
\end{itemize}
\end{corollary}

\begin{proof}
The limits \eqref{eq.fg_covar_low}, \eqref{eq.fg_covar_high}, and \eqref{eq.fg_covar_top} follow immediately by applying Proposition \ref{thm.geom_bound_deriv} and \eqref{eq.fg_main_low_g} to \eqref{eq.fg_main_low_R}, \eqref{eq.fg_main_high}, and \eqref{eq.fg_main_top}, respectively.
Thus, it remains only to prove \eqref{eq.fg_covar_free}.

Applying Proposition \ref{thm.geom_bound_deriv} and \eqref{eq.fg_main_low_g} to \eqref{eq.fg_main_free} yields
\begin{align*}
\Dv^m \mi{L}_\rho^n \gv - n! \, ( \log \rho ) \Dv^m \gb{\star} &\rightarrow^0 n! \, \Dm^m \gb{\dagger} \text{,} \\
\Dv^u \mi{L}_\rho^n \Rv - n! \, ( \log \rho ) \Dv^u \Rb{\star} &\rightarrow^0 n! \, \Dm^u \Rb{\dagger} \text{,} \\
\Dv^v \mi{L}_\rho^{ n - 1 } ( \Dv \mi{L}_\rho \gv ) - ( n - 1 )! \, ( \log \rho ) \Dv^v \bb{\star} &\rightarrow^0 ( n - 1 )! \, \Dm^v \bb{\dagger} \text{,}
\end{align*}
for $m$, $u$, $v$ as in \eqref{eq.fg_covar_free}.
Thus, it suffices to show---for the same $m$, $u$, $v$---that
\begin{align}
\label{eql.fg_covar_1} ( \log \rho ) ( \Dv^m \gb{\star} - \Dm^m \gb{\star} ) &\rightarrow^0 0 \text{,} \\
\notag ( \log \rho ) ( \Dv^u \Rb{\star} - \Dm^u \Rb{\star} ) &\rightarrow^0 0 \text{,} \\
\notag ( \log \rho ) ( \Dv^v \bb{\star} - \Dm^v \bb{\star} ) &\rightarrow^0 0 \text{.}
\end{align}

For this, we let $( U, \varphi )$ be any compact coordinate system on $\mi{I}$.
Observe that \eqref{eq.geom_bound_deriv}, the fundamental theorem of calculus, \eqref{eq.fg_main_low_g}, and \eqref{eq.fg_main_high} imply, for any $0 \leq m \leq M_0 - n$, that
\begin{align*}
| ( \log \rho ) ( \Dv^m \gb{\star} - \Dm^m \gb{\star} ) |_{ 0, \varphi } &\lesssim | \log \rho | \cdot \| \gb{\star} \|_{ M_0 - n - 1, \varphi } | \gv - \gm |_{ M_0 - n, \varphi } \\
&\lesssim | \rho \log \rho | \cdot \| \mi{L}_\rho \gv \|_{ M_0 - 1, \varphi } \\
&\lesssim | \rho \log \rho | \text{,}
\end{align*}
from which the first limit in \eqref{eql.fg_covar_1} follows.
By similar methods, we also obtain
\[
| ( \log \rho ) ( \Dv^u \Rb{\star} - \Dm^u \Rb{\star} ) |_{ 0, \varphi } \lesssim | \rho \log \rho | \text{,} \qquad | ( \log \rho ) ( \Dv^v \bb{\star} - \Dm^v \bb{\star} ) | \lesssim | \rho \log \rho | \text{,}
\]
for all $0 \leq u \leq M_0 - n - 2$ and $0 \leq v \leq M_0 - n - 1$, which yield the remaining parts of \eqref{eql.fg_covar_1}.
\end{proof}

\subsubsection{Partial Expansions} \label{sec.fg_exp}

We now use Theorem \ref{thm.fg_main} to recover the partial Fefferman--Graham expansions for $\gv$, $\Rv$, and $\Dv \mi{L}_\rho \gv$, up to and including the $n$-th order term.
All the information that is required for the expansion is already present in the boundary limits \eqref{eq.fg_main_low_g}, \eqref{eq.fg_main_low_R}, \eqref{eq.fg_main_high}, \eqref{eq.fg_main_top}, and \eqref{eq.fg_main_free}.
The remaining step is to systematically construct the expansion via Taylor's theorem (and to account for the anomalous logarithmic terms when $n$ is even).

\begin{theorem} \label{thm.fg_exp}
Let $( \mi{M}, g )$ be a vacuum FG-aAdS segment, assume $n = \dim \mi{I} > 1$, and fix $M_0 \geq n + 2$.
In addition, assume the bounds \eqref{eq.fg_main_ass_g} and \eqref{eq.fg_main_ass_Lg} hold for any compact coordinate system $( U, \varphi )$ on $\mi{I}$.
Then, $\gv$, $\Rv$, and $\Dv \mi{L}_\rho \gv$ can be expressed as partial expansions in $\rho$,
\begin{align}
\label{eq.fg_exp} \gv &= \begin{cases} \sum_{ k = 0 }^{ \frac{ n - 1 }{2} } \gb{ 2k } \rho^{ 2 k } + \gb{n} \rho^n + \ms{r}_{ \gv } \rho^n & \text{$n$ odd,} \\ \sum_{ k = 0 }^{ \frac{ n - 2 }{2} } \gb{ 2k } \rho^{ 2 k } + \gb{\star} \rho^n \log \rho + \gb{n} \rho^n + \ms{r}_{ \gv } \rho^n & \text{$n$ even,} \end{cases} \\
\notag \Rv &= \begin{cases} \sum_{ k = 0 }^{ \frac{ n - 1 }{2} } \Rb{ 2k } \rho^{ 2 k } + \Rb{n} \rho^n + \ms{r}_{ \Rv } \rho^n & \text{$n$ odd,} \\ \sum_{ k = 0 }^{ \frac{ n - 2 }{2} } \Rb{ 2k } \rho^{ 2 k } + \Rb{\star} \rho^n \log \rho + \Rb{n} \rho^n + \ms{r}_{ \Rv } \rho^n & \text{$n$ even,} \end{cases} \\
\notag \Dv \mi{L}_\rho \gv &= \begin{cases} \sum_{ l = 1 }^{ \frac{ n - 1 }{2} } \bb{ 2l } \rho^{ 2 l - 1 } + \bb{n} \rho^{ n - 1 } + \ms{r}_{ \Dv } \rho^{ n - 1 } & \text{$n$ odd,} \\ \sum_{ l = 1 }^{ \frac{ n - 2 }{2} } \bb{ 2l } \rho^{ 2 l - 1 } + \bb{\star} \rho^{ n - 1 } \log \rho + \bb{n} \rho^{ n - 1 } + \ms{r}_{ \Dv } \rho^{ n - 1 } & \text{$n$ even,} \end{cases}
\end{align}
where the following also hold:
\begin{itemize}
\item For any $0 \leq k < \frac{n}{2}$ and $1 \leq l < \frac{n}{2}$, we have that $\gb{2k}$, $\Rb{2k}$, and $\bb{2l}$ are tensor fields on $\mi{I}$ that depend only on $\gm$ to order $2 k$, $2 k + 2$, and $2 l + 1$, respectively.
In addition,
\begin{equation}
\label{eq.fg_exp_low} \gb{0} = \gm \text{,} \qquad \Rb{0} = \Rm \text{.}
\end{equation}

\item If $n$ is even, then the coefficients $\gb{\star}$, $\Rb{\star}$, and $\bb{\star}$ are tensor fields on $\mi{I}$ that depend only on $\gm$ to orders $n$, $n + 2$, and $n + 1$, respectively.
Moreover, $\gb{\star}$ satisfies
\begin{equation}
\label{eq.fg_exp_top} \trace{\gm} \gb{\star} = 0 \text{,} \qquad \Dm \cdot \gb{\star} = 0 \text{.}
\end{equation}

\item We also have the following identities:
\begin{equation}
\label{eq.fg_exp_schouten} \begin{cases} \gb{2} = - \frac{ 1 }{ n - 2 } \left[ \Rcm - \frac{ 1 }{ 2 ( n - 1 ) } \Rsm \cdot \gm \right] & n > 2 \text{,} \\ ( \gb{\star}, \Rb{\star}, \bb{\star} ) = ( 0, 0, 0 ) & n = 2 \text{.} \end{cases}
\end{equation}

\item $\gb{n}$, $\Rb{n}$, and $\bb{n}$ are $C^{ M_0 - n }$, $C^{ M_0 - n - 2 }$, and $C^{ M_0 - n - 1 }$ tensor fields on $\mi{I}$, respectively.
Furthermore, if $n$ is odd, then $\trace{\gm} \gb{n}$ and $\Dm \cdot \gb{n}$ vanish; whereas if $n$ is even, then $\trace{\gm} \gb{n}$ and $\smash{\Dm \cdot \gm}$ depend only on $\gm$ to orders $n$ and $n + 1$, respectively.

\item The ``remainders" $\ms{r}_{ \gv }$, $\ms{r}_{ \Rv }$, and $\ms{r}_{ \Dv }$ are vertical tensor fields satisfying 
\begin{equation}
\label{eq.fg_exp_error} \ms{r}_{ \gv } \rightarrow^{ M_0 - n } 0 \text{,} \qquad \ms{r}_{ \Rv } \rightarrow^{ M_0 - n - 2 } 0 \text{,} \qquad \ms{r}_{ \Dv } \rightarrow^{ M_0 - n - 1 } 0 \text{.}
\end{equation}
\end{itemize}
\end{theorem}

\begin{proof}
First, observe that the conclusions of Theorem \ref{thm.fg_main} hold in this setting, since the hypotheses of Theorems \ref{thm.fg_exp} and \ref{thm.fg_main} are the same.
As a result, we let the tensor fields
\[
\gb{0}, \dots, \gb{n - 1}, \gb{\star}, \gb{\dagger} \text{,} \qquad \Rb{0}, \dots, \Rb{n - 1}, \Rb{\star}, \Rb{\dagger} \text{,} \qquad \bb{1}, \dots, \bb{n - 1}, \bb{\star}, \bb{\dagger} \text{,}
\]
be as in the conclusions of Theorem \ref{thm.fg_main}.

Define now the vertical tensor fields
\begin{equation}
\label{eql.fg_exp_0} \ms{h} := \gv - \gb{\star} \rho^n \log \rho \text{,} \qquad \ms{S} := \Rv - \Rb{\star} \rho^n \log \rho \text{,} \qquad \ms{B} := \Dv \mi{L}_\rho \gv - \bb{\star} \rho^{ n - 1 } \log \rho \text{.}
\end{equation}
Differentiating $\ms{h}$, $\ms{S}$, and $\ms{B}$ and recalling \eqref{eq.fg_main_high}, we see that
\begin{equation}
\label{eql.fg_exp_1} \mi{L}_\rho^k \ms{h} \Rightarrow^{ M_0 - k } k! \, \gb{k} \text{,} \qquad \mi{L}_\rho^k \ms{S} \Rightarrow^{ M_0 - k - 2 } k! \, \Rb{k} \text{,} \qquad \mi{L}_\rho^{ l - 1 } \ms{B} \Rightarrow^{ M_0 - l - 1 } ( l - 1 )! \, \bb{l} \text{.}
\end{equation}
for each $0 \leq k < n$ and $1 \leq l < n$.
Moreover, taking one more $\rho$-derivative yields
\begin{align}
\label{eql.fg_exp_2} \mi{L}_\rho^n \ms{h} &= \mi{L}_\rho^n \gv - n! \, \gb{\star} \log \rho - n! \, C_n \gb{\star} \text{,} \\
\notag \mi{L}_\rho^n \ms{S} &= \mi{L}_\rho^n \Rv - n! \, \Rb{\star} \log \rho - n! \, C_n \Rb{\star} \text{,} \\
\notag \mi{L}_\rho^{ n - 1 } \ms{B} &= \mi{L}_\rho^{ n - 1 } ( \Dv \mi{L}_\rho \gv ) - ( n - 1 )! \, \bb{\star} \log \rho - ( n - 1 )! \, C_{ n - 1 } \bb{\star} \text{,}
\end{align}
where the constants $C_n$ and $C_{ n - 1 }$ are given by
\[
C_n = 1^{-1} + \dots + n^{-1} \text{,} \qquad C_{ n - 1 } = 1^{-1} + \dots + ( n - 1 )^{-1} \text{.}
\]
Defining now the tensor fields
\[
\gb{n} := \gb{\dagger} - C_n \gb{\star} \text{,} \qquad \Rb{n} := \Rb{\dagger} - C_n \Rb{\star} \text{,} \qquad \bb{n} := \bb{\dagger} - C_{ n - 1 } \gb{\star} \text{,}
\]
we then see from \eqref{eq.fg_main_free} and \eqref{eql.fg_exp_2} that
\begin{equation}
\label{eql.fg_exp_3} \mi{L}_\rho^n \ms{h} \rightarrow^{ M_0 - n } n! \, \gb{n} \text{,} \qquad \mi{L}_\rho^n \ms{S} \rightarrow^{ M_0 - n - 2 } n! \, \Rb{n} \text{,} \qquad \mi{L}_\rho^{ n - 1 } \ms{B} \rightarrow^{ M_0 - n - 1 } ( n - 1 )! \, \bb{n} \text{.}
\end{equation}

Applying Taylor's theorem, \eqref{eql.fg_exp_1}, and \eqref{eql.fg_exp_3} to $\ms{h}$, $\ms{S}$, $\ms{B}$ yields
\[
\ms{h} = \sum_{ k = 0 }^n \gb{k} \rho^k + \ms{r}_{ \gv } \rho^n \text{,} \qquad \ms{S} = \sum_{ k = 0 }^n \Rb{k} \rho^k + \ms{r}_{ \Rv } \rho^n \text{,} \qquad \ms{B} = \sum_{ l = 1 }^n \bb{l} \rho^{ l - 1 } + \ms{r}_{ \Dv } \rho^{ n - 1 } \text{,}
\]
where $\ms{r}_{ \gv }$, $\ms{r}_{ \Rv }$, and $\ms{r}_{ \Dv }$ are vertical tensor fields satisfying \eqref{eq.fg_exp_error}.
The above and \eqref{eql.fg_exp_0} then imply 
\begin{align*}
\gv &= \sum_{ k = 0 }^n \gb{k} \rho^k + \gb{\star} \rho^n \log \rho + \ms{r}_{ \gv } \rho^n \text{,} \\
\notag \Rv &= \sum_{ k = 0 }^n \Rb{k} \rho^k + \Rb{\star} \rho^n \log \rho + \ms{r}_{ \Rv } \rho^n \text{,} \\
\notag \Dv \mi{L}_\rho \gv &= \sum_{ l = 1 }^n \bb{l} \rho^{ l - 1 } + \bb{\star} \rho^{ n - 1 } \log \rho + \ms{r}_{ \Dv } \rho^{ n - 1 } \text{.}
\end{align*}
Combining the above with \eqref{eq.fg_main_high_depend} and \eqref{eq.fg_main_top_depend} results in the expansions \eqref{eq.fg_exp}.

Finally, the dependence of all the coefficients in \eqref{eq.fg_exp} on $\gm$ follows from \eqref{eq.fg_main_high_depend}, \eqref{eq.fg_main_top_depend}, and \eqref{eq.fg_main_free_constraint}.
The remaining relations \eqref{eq.fg_exp_low}--\eqref{eq.fg_exp_schouten} then follow from \eqref{eq.fg_main_low_g}, \eqref{eq.fg_main_low_R}, \eqref{eq.fg_main_top_constraint}, and \eqref{eq.fg_main_schouten}.
\end{proof}

Finally, we derive corresponding partial expansions for the spacetime Weyl curvature:

\begin{corollary} \label{thm.fg_exp_W}
Let $( \mi{M}, g )$ be a vacuum FG-aAdS segment, assume $n = \dim \mi{I} > 2$, and fix $M_0 \geq n + 2$.
In addition, assume the bounds \eqref{eq.fg_main_ass_g} and \eqref{eq.fg_main_ass_Lg} hold for any compact coordinate system $( U, \varphi )$ on $\mi{I}$.
Then, with respect to any arbitrary compact coordinate system on $\mi{I}$, the components of the spacetime Weyl curvature $W$ can be expressed as partial expansions in $\rho$,
\begin{align}
\label{eq.fg_exp_W} W_{ \rho c a b } &= \begin{cases} \sum_{ k = 1 }^{ \frac{ n - 1 }{2} } \mf{x}^{ (2k) }_{ c a b } \rho^{ 2 k - 3 } + \mf{x}^{ (n) }_{ c a b } \rho^{ n - 3 } + \ms{r}^{ \mf{x} }_{ c a b } \rho^{ n - 3 } & \text{$n$ odd,} \\ \sum_{ k = 1 }^{ \frac{ n - 2 }{2} } \mf{x}^{ (2k) }_{ c a b } \rho^{ 2 k - 3 } + \mf{x}^{ (\star) }_{ c a b } \rho^{ n - 3 } \log \rho + \mf{x}^{ (n) }_{ c a b } \rho^{ n - 3 } + \ms{r}^{ \mf{x} }_{ c a b } \rho^{ n - 3 } & \text{$n$ even,} \end{cases} \\
\notag W_{ a b c d } &= \begin{cases} \sum_{ k = 1 }^{ \frac{ n - 1 }{2} } \mf{y}^{ (2k) }_{ a b c d } \rho^{ 2 k - 4 } + \mf{y}^{ (n) }_{ a b c d } \rho^{ n - 4 } + \ms{r}^{ \mf{y} }_{ a b c d } \rho^{ n - 4 } & \text{$n$ odd,} \\ \sum_{ k = 1 }^{ \frac{ n - 2 }{2} } \mf{y}^{ (2k) }_{ a b c d } \rho^{ 2 k - 4 } + \mf{y}^{ (\star) }_{ a b c d } \rho^{ n - 4 } \log \rho + \mf{y}^{ (n) }_{ a b c d } \rho^{ n - 4 } + \ms{r}^{ \mf{y} }_{ a b c d } \rho^{ n - 4 } & \text{$n$ even,} \end{cases} \\
\notag W_{ \rho a \rho b } &= \begin{cases} \sum_{ k = 2 }^{ \frac{ n - 1 }{2} } \mf{z}^{ (2k) }_{ a b } \rho^{ 2 k - 4 } + \mf{z}^{ (n) }_{ a b } \rho^{ n - 4 } + \ms{r}^{ \mf{z} }_{ a b } \rho^{ n - 4 } & \text{$n$ odd,} \\ \sum_{ k = 2 }^{ \frac{ n - 2 }{2} } \mf{z}^{ (2k) }_{ a b } \rho^{ 2 k - 4 } + \mf{z}^{ (\star) }_{ a b } \rho^{ n - 4 } \log \rho + \mf{z}^{ (n) }_{ a b } \rho^{ n - 4 } + \ms{r}^{ \mf{z} }_{ a b } \rho^{ n - 4 } & \text{$n$ even,} \end{cases}
\end{align}
where the following properties hold:
\begin{itemize}
\item $\mf{x}^{ ( 2 k ) }$ (where $1 \leq k < \frac{n}{2}$), $\mf{y}^{ ( 2 k ) }$ (where $1 \leq k < \frac{n}{2}$), and $\mf{z}^{ ( 2 k ) }$ (where $2 \leq k < \frac{n}{2}$) are tensor fields on $\mi{I}$ that depend on $\gm$ to order $2 k + 1$, $2 k$, and $2 k$, respectively.

\item $\mf{x}^{ (\star) }$, $\mf{y}^{ (\star) }$, $\mf{z}^{ (\star) }$ are tensor fields on $\mi{I}$ depending on $\gm$ to orders $n + 1$, $n$, and $n$, respectively.

\item $\mf{x}^{ (n) }$, $\mf{y}^{ (n) }$, and $\mf{z}^{ (n) }$ are $C^{ M_0 - n - 1 }$, $C^{ M_0 - n }$, and $C^{ M_0 - n }$ tensor fields on $\mi{I}$, respectively.

\item $\ms{r}^{ \mf{x} }$, $\ms{r}^{ \mf{y} }$, and $\ms{r}^{ \mf{z} }$ are vertical tensor fields that satisfy
\begin{equation}
\label{eq.fefferman_graham_W_remainder} \ms{r}^{ \mf{x} } \rightarrow^{ M_0 - n - 1 } 0 \text{,} \qquad \ms{r}^{ \mf{y} } \rightarrow^{ M_0 - n } 0 \text{,} \qquad \ms{r}^{ \mf{z} } \rightarrow^{ M_0 - n } 0 \text{.}
\end{equation}
\end{itemize}
\end{corollary}

\begin{proof}
The first expansion of \eqref{eq.fg_exp_W} follows immediately from the third part of \eqref{eq.fg_exp} and the first equation in \eqref{eq.geom_vertical}.
For the remaining expansions, we adopt the following shorthand:~we say $( \mf{a}, \ms{r} )$ is a \emph{remainder pair} iff $\mf{a}$ is a $C^{ M_0 - n }$ tensor field on $\mi{I}$ and $\ms{r}$ is a vertical tensor field with $\ms{r} \rightarrow^{ M_0 - n } 0$.

For the second part of \eqref{eq.fg_exp_W}, we begin by applying Taylor's theorem along with \eqref{eq.fg_main_high}, \eqref{eq.fg_main_top}, and \eqref{eq.fg_main_free} as before to obtain, for some remainder pair $( \mf{a}_1, \ms{r}_1 )$, the expansion
\begin{align}
\label{eql.fg_exp_W_1} \mi{L}_\rho \gv &= \begin{cases} \sum_{ k = 1 }^{ \frac{ n - 1 }{2} } \mc{D}^{ 2 k } ( \gm ) \cdot \rho^{ 2 k - 1 } + \mf{a}_1 \rho^{ n - 1 } + \ms{r}_1 \rho^{ n - 1 } & \text{$n$ odd,} \\ \sum_{ k = 1 }^{ \frac{ n - 2 }{2} } \mc{D}^{ 2 k } ( \gm ) \cdot \rho^{ 2 k - 1 } + \mc{D}^n ( \gm ) \cdot \rho^{ n - 1 } \log \rho + \mf{a}_1 \rho^{ n - 1 } + \ms{r}_1 \rho^{ n - 1 } & \text{$n$ even.} \end{cases}
\end{align}
Moreover, multiplying \eqref{eql.fg_exp_W_1} by itself and using the first expansion in \eqref{eq.fg_exp}, we compute
\begin{align}
\label{eql.fg_exp_W_2} \mi{S} ( \mi{L}_\rho \gv, \mi{L}_\rho \gv ) &= \begin{cases} \sum_{ k = 1 }^{ \frac{ n - 1 }{2} } \mc{D}^{ 2k } ( \gm ) \cdot \rho^{ 2 k } + \mf{a}_2 \rho^n + \ms{r}_2 \rho^n & \text{$n$ odd,} \\ \sum_{ k = 1 }^{ \frac{ n - 2 }{2} } \mc{D}^{ 2k } ( \gm ) \cdot \rho^{ 2 k } + \mc{D}^n ( \gm ) \cdot \rho^n \log \rho + \mf{a}_2 \rho^n + \ms{r}_2 \rho^n & \text{$n$ even,} \end{cases} \\
\notag \mi{S} ( \rho^{-1} \gv, \mi{L}_\rho \gv ) &= \begin{cases} \sum_{ k = 1 }^{ \frac{ n - 1 }{2} } \mc{D}^{ 2k } ( \gm ) \cdot \rho^{ 2 k - 2 } + \mf{a}_3 \rho^{ n - 2 } + \ms{r}_3 \rho^{ n - 2 } & \text{$n$ odd,} \\ \sum_{ k = 1 }^{ \frac{ n - 2 }{2} } \mc{D}^{ 2k } ( \gm ) \cdot \rho^{ 2 k - 2 } + \mc{D}^n ( \gm ) \cdot \rho^{ n - 2 } \log \rho + \mf{a}_3 \rho^{ n - 2 } + \ms{r}_3 \rho^{ n - 2 } & \text{$n$ even,} \end{cases}
\end{align}
where $( \mf{a}_2, \ms{r}_2 )$ and $( \mf{a}_3, \mf{r}_3 )$ are remainder pairs.
Substituting the second expansion of \eqref{eq.fg_exp} and the expansions in \eqref{eql.fg_exp_W_2} into the right-hand side of \eqref{eq.geom_vertical} results in the second expansion of \eqref{eq.fg_exp_W}.

For the final expansion, we first obtain---using \eqref{eq.fg_main_high}, \eqref{eq.fg_main_top}, and \eqref{eq.fg_main_free}---that
\[
\gv^{-1} = \begin{cases} \sum_{ k = 0 }^{ \frac{ n - 1 }{2} } \mc{D}^{ 2k } ( \gm ) \cdot \rho^{ 2 k } + \mf{a}_4 \rho^n + \ms{r}_4 \rho^n & \text{$n$ odd,} \\ \sum_{ k = 0 }^{ \frac{ n - 2 }{2} } \mc{D}^{ 2k } ( \gm ) \cdot \rho^{ 2 k } + \mc{D}^n ( \gm ) \cdot \rho^n \log \rho + \mf{a}_4 \rho^n + \ms{r}_4 \rho^n & \text{$n$ even,} \end{cases}
\]
with $( \mf{a}_4, \ms{r}_4 )$ being a remainder pair.
From the first part of \eqref{eql.fg_exp_W_2} and the above, we then have
\begin{equation}
\label{eql.fg_exp_W_3} \mi{S} ( \gv^{-1}, \mi{L}_\rho \gv, \mi{L}_\rho \gv ) = \begin{cases} \sum_{ k = 1 }^{ \frac{ n - 1 }{2} } \mc{D}^{ 2k } ( \gm ) \cdot \rho^{ 2 k } + \mf{a}_5 \rho^n + \ms{r}_5 \rho^n & \text{$n$ odd,} \\ \sum_{ k = 1 }^{ \frac{ n - 2 }{2} } \mc{D}^{ 2k } ( \gm ) \cdot \rho^{ 2 k } + \mc{D}^n ( \gm ) \cdot \rho^n \log \rho + \mf{a}_5 \rho^n + \ms{r}_5 \rho^n & \text{$n$ even,} \end{cases} \\
\end{equation}
with $( \mf{a}_5, \mr{r}_5 )$ also being a remainder pair.
Furthermore, recalling \eqref{eq.fg_main_high}, \eqref{eq.fg_main_top}, \eqref{eq.fg_main_free} again and then applying Taylor's theorem as before, we derive the expansion
\begin{equation}
\label{eql.fg_exp_W_4} \mi{L}_\rho^2 \gv - \rho^{-1} \mi{L}_\rho \gv = \begin{cases} \sum_{ k = 2 }^{ \frac{ n - 1 }{2} } \mc{D}^{ 2k } ( \gm ) \cdot \rho^{ 2 k - 2 } + \mf{a}_6 \rho^{ n - 2 } + \ms{r}_6 \rho^{ n - 2 } & \text{$n$ odd,} \\ \sum_{ k = 2 }^{ \frac{ n - 2 }{2} } \mc{D}^{ 2k } ( \gm ) \cdot \rho^{ 2 k - 2 } + \mc{D}^n ( \gm ) \cdot \rho^{ n - 2 } \log \rho + \mf{a}_6 \rho^{ n - 2 } + \ms{r}_6 \rho^{ n - 2 } & \text{$n$ even,} \end{cases}
\end{equation}
where $( \mf{a}_6, \ms{r}_6 )$ is a remainder pair.
(Note in particular that the leading terms of the expansions for $\mi{L}_\rho^2 \gv$ and $\rho^{-1} \mi{L}_\rho \gv$ cancel.)
Applying the third formula of \eqref{eq.geom_vertical} and substituting the expansions of \eqref{eql.fg_exp_W_3} and \eqref{eql.fg_exp_W_4} into the right-hand side yields the last expansion of \eqref{eq.fg_exp_W}.
\end{proof}

\subsection{Proof of Theorem \ref{thm.fg_main}} \label{sec.fg_proof}

This subsection is dedicated to the proof of the main result.
Throughout, we will assume that the hypotheses of Theorem \ref{thm.fg_main} hold.

\subsubsection{The Initial Limits}

The first step is to establish the lower-order limits \eqref{eq.fg_main_low_g} and \eqref{eq.fg_main_low_R}.
In proving this, we will also set up the iteration process for obtaining the higher-order limits.

\begin{lemma} \label{thm.limit_low}
$\mi{L}_\rho \gv$ is locally bounded in $C^{ M_0 + 1 }$.
Furthermore,
\begin{equation}
\label{eq.limit_low} \mi{L}_\rho \gv \Rightarrow^{ M_0 } 0 \text{,} \qquad \rho \mi{L}_\rho^2 \gv \Rightarrow^{ M_0 } 0 \text{.}
\end{equation}
\end{lemma}

\begin{proof}
Let $( U, \varphi )$ denote a compact coordinate system on $\mi{I}$.
We begin by applying \eqref{eq.limit_positive_bound}, with $\rho_\ast := \rho_0$, to the second identity of \eqref{eq.einstein_vertical_ex}.
This yields the inequality
\begin{equation}
\label{eql.limit_low_11} | \trace{\gv} \mi{L}_\rho \gv | |_{ 0, \varphi } \lesssim | \trace{\gv} \mi{L}_\rho \gv |_{ 0, \varphi } |_{ \rho_0 } + \rho^{ 2 n - 1 } \int_\rho^{ \rho_0 } \sigma^{ - 2 n + 1 } [ | \Rsv |_{ 0, \varphi } + | \mi{S} ( \gv; \mi{L}_\rho \gv, \mi{L}_\rho \gv ) |_{ 0, \varphi } ] |_\sigma d \sigma \text{.}
\end{equation}
The first term on the right-hand side of \eqref{eql.limit_low_11} is bounded, since $( U, \varphi )$ is compact.
Moreover, the curvature term in \eqref{eql.limit_low_11} can be controlled using \eqref{eq.geom_bound_R} and \eqref{eq.fg_main_ass_g}.
Therefore, bounding the factors of $\gv$ and $\gv^{-1}$ in the remaining schematic term in \eqref{eql.limit_low_11} using \eqref{eq.geom_bound_g} and \eqref{eq.fg_main_ass_g}, we obtain
\begin{equation}
\label{eql.limit_low_12} | \trace{\gv} \mi{L}_\rho \gv | |_{ 0, \varphi } \lesssim 1 + \int_\rho^{ \rho_0 } | \mi{L}_\rho \gv |_{ 0, \varphi } | \mi{L}_\rho \gv |_{ 0, \varphi } |_\sigma d \sigma \text{.}
\end{equation}

Applying \eqref{eq.limit_positive_bound} to the first identity in \eqref{eq.einstein_vertical_ex}---along with \eqref{eq.geom_bound_g}, \eqref{eq.geom_bound_R}, \eqref{eq.fg_main_ass_g}, and \eqref{eql.limit_low_12}---yields
\begin{align*}
| \mi{L}_\rho \gv |_{ 0, \varphi } &\lesssim 1 + \rho^{ n - 1 } \int_\rho^{ \rho_0 } \sigma^{ -n } \left[ | \trace{\gv} \mi{L}_\rho \gv |_{ 0, \varphi } + \sigma | \Rcv |_{ 0, \varphi } + \sigma | \mi{L}_\rho \gv |_{ 0, \varphi } | \mi{L}_\rho \gv |_{ 0, \varphi } \right] |_\sigma d \sigma\\
&\lesssim 1 + \rho^{ n - 1 } \int_\rho^{ \rho_0 } \sigma^{-n} | \trace{\gv} \mi{L}_\rho \gv |_{ 0, \varphi } |_\sigma d \sigma + \int_\rho^{ \rho_0 } | \mi{L}_\rho \gv |_{ 0, \varphi } | \mi{L}_\rho \gv |_{ 0, \varphi } |_\sigma d \sigma \\
&\lesssim 1 + \int_\rho^{ \rho_0 } | \mi{L}_\rho \gv |_{ 0, \varphi } | \mi{L}_\rho \gv |_{ 0, \varphi } |_\sigma d \sigma \text{.}
\end{align*}
Applying Gronwall's inequality to the above and recalling the assumption \eqref{eq.fg_main_ass_Lg} yields
\begin{equation}
\label{eql.limit_low_13} \| \mi{L}_\rho \gv \|_{ 0, \varphi } \lesssim 1 \text{.}
\end{equation}

Fix now $0 < M \leq M_0$, and assume for induction that
\begin{equation}
\label{eql.limit_low_20} \| \mi{L}_\rho \gv \|_{ M - 1, \varphi } \lesssim 1 \text{.}
\end{equation}
Observe that the product rule, \eqref{eq.geom_bound_g}, \eqref{eq.fg_main_ass_g}, and \eqref{eql.limit_low_20} imply the schematic bound
\begin{equation}
\label{eql.limit_low_21} | \mi{S} ( \gv; \mi{L}_\rho \gv, \mi{L}_\rho \gv ) |_{ M, \varphi } \lesssim 1 + | \mi{L}_\rho \gv |_{ M, \varphi } \text{.}
\end{equation}
We now apply \eqref{eq.limit_positive_bound} to the second equation in \eqref{eq.einstein_vertical_ex}, with $M$ derivatives, to obtain
\begin{align}
\label{eql.limit_low_22} | \trace{\gv} \mi{L}_\rho \gv | |_{ M, \varphi } &\lesssim | \trace{\gv} \mi{L}_\rho \gv |_{ M, \varphi } |_{ \rho_0 } + \rho^{ 2 n - 1 } \int_\rho^{ \rho_0 } \sigma^{ - 2 n + 1 } [ | \Rsv |_{ M, \varphi } + | \mi{S} ( \gv; \mi{L}_\rho \gv, \mi{L}_\rho \gv ) |_{ M, \varphi } ] |_\sigma d \sigma \\
\notag &\lesssim 1 + \int_\rho^{ \rho_0 } | \mi{L}_\rho \gv |_{ M, \varphi } |_\sigma d \sigma \text{,}
\end{align}
where in the second step, we used \eqref{eq.geom_bound_R}, \eqref{eq.fg_main_ass_g}, \eqref{eql.limit_low_21}, and the compactness of $( U, \varphi )$.
Applying \eqref{eq.limit_positive_bound} to the first identity of \eqref{eq.einstein_vertical_ex} and recalling \eqref{eq.geom_bound_R}, \eqref{eq.fg_main_ass_g}, \eqref{eql.limit_low_21}, and \eqref{eql.limit_low_22} yields
\begin{align*}
| \mi{L}_\rho \gv |_{ M, \varphi } &\lesssim 1 + \rho^{ n - 1 } \int_{ \rho_0 }^\rho \sigma^{ -n } \left[ | \trace{\gv} \mi{L}_\rho \gv |_{ M, \varphi } + \sigma | \Rcv |_{ M, \varphi } + \sigma | \mi{L}_\rho \gv |_{ M, \varphi } \right] |_\sigma d \sigma \\
&\lesssim 1 + \int_\rho^{ \rho_0 } | \mi{L}_\rho \gv |_{ M, \varphi } |_\sigma d \sigma \text{.}
\end{align*}
Applying Gronwall's inequality to the above yields the bound
\[
\| \mi{L}_\rho \gv \|_{ M, \varphi } \lesssim 1 \text{.}
\]

As a result, iterating the argument from \eqref{eql.limit_low_20} to the above, we conclude that
\begin{equation}
\label{eql.limit_low_30} \| \mi{L}_\rho \gv \|_{ M_0, \varphi } \lesssim 1 \text{,} \qquad \| \mi{S} ( \gv; \mi{L}_\rho \gv, \mi{L}_\rho \gv ) \|_{ M_0, \varphi } \lesssim 1 \text{.}
\end{equation}
(Notice that the second bound in \eqref{eql.limit_low_30} follows from \eqref{eq.geom_bound_g}, \eqref{eq.fg_main_ass_g}, and the first part of \eqref{eql.limit_low_30}.)
Therefore, applying Proposition \ref{thm.limit_bound} to \eqref{eql.limit_low_30}, we obtain
\begin{equation}
\label{eql.limit_low_31} \rho \mi{L}_\rho \gv \Rightarrow^{ M_0 } 0 \text{,} \qquad \rho \cdot \mi{S} ( \mi{L}_\rho \gv, \mi{L}_\rho \gv ) \Rightarrow^{ M_0 } 0 \text{.}
\end{equation}
Similarly, Proposition \ref{thm.limit_bound}, \eqref{eq.geom_bound_R}, and \eqref{eq.fg_main_ass_g} imply that
\begin{equation}
\label{eql.limit_low_32} \rho \Rcv \Rightarrow^{ M_0 } 0 \text{,} \qquad \rho \Rsv \Rightarrow^{ M_0 } 0 \text{.}
\end{equation}
From \eqref{eql.limit_low_31} and \eqref{eql.limit_low_32}, we conclude that the right-hand side $\ms{I}$ of the second equation of \eqref{eq.einstein_vertical_ex} satisfies $\ms{I} \Rightarrow^{ M_0 } 0$.
Consequently, using Proposition \ref{thm.limit_positive}, we conclude that
\begin{equation}
\label{eql.limit_low_33} \trace{\gv} \mi{L}_\rho \gv \Rightarrow^{ M_0 } 0 \text{.}
\end{equation}
Also, applying Proposition \ref{thm.limit_positive} and \eqref{eql.limit_low_31}--\eqref{eql.limit_low_33} to the first part of \eqref{eq.einstein_vertical_ex} yields \eqref{eq.limit_low}.
\end{proof}

Now, by the fundamental theorem of calculus, we have, in any compact coordinate system $( U, \varphi )$,
\[
\gv_{ a b } - \gm_{ a b } = \int_0^\rho \mi{L}_\rho \gv_{ a b } |_\sigma d \sigma \text{.}
\]
Differentiating the above (in $\varphi$-coordinates) and recalling \eqref{eq.limit_low}, we see that
\[
\partial^M_{ c_1 \dots c_M } \gv_{ a b } - \partial^M_{ c_1 \dots c_M } \gm_{ a b } = \int_0^\rho \partial^M_{ c_1 \dots c_M } \mi{L}_\rho \gv_{ a b } |_\sigma d \sigma \text{,} \qquad 0 \leq M \leq M_0 \text{.}
\]
From the above and Lemma \ref{thm.limit_low}, we immediately obtain
\begin{align*}
\lim_{ \rho \searrow 0 } \sup_{ \{ \rho \} \times U } | \gv - \gm |_{ M_0, \varphi } &\lesssim \lim_{ \rho \searrow 0 } \int_0^\rho d \sigma \cdot \| \mi{L}_\rho \gv \|_{ M_0, \varphi } = 0 \text{,} \\
\int_0^{ \rho_0 } \frac{1}{ \sigma } | \gv - \gm |_{ M_0, \varphi } |_\sigma d \sigma &\lesssim \int_0^{ \rho_0 } \frac{1}{ \sigma } \int_0^\sigma d \tau d \sigma \cdot \| \mi{L}_\rho \gv \|_{ M_0, \varphi } < \infty \text{,}
\end{align*}
which yields the first limit in \eqref{eq.fg_main_low_g}.
This, along with the last part of \eqref{eq.geom_bound_g} and \eqref{eq.fg_main_ass_g}, then yields the second limit of \eqref{eq.fg_main_low_g}.
Similarly, the limits in \eqref{eq.fg_main_low_R} follow from \eqref{eq.geom_bound_R}, \eqref{eq.fg_main_ass_g}, \eqref{eq.fg_main_low_g}.

Finally, we establish some consequences of \eqref{eq.fg_main_low_g}:

\begin{lemma} \label{thm.limit_schematic_metric}
Let $0 \leq M \leq M_0$ and $N \geq 1$.
Furthermore, let $\ms{A}_1, \dots, \ms{A}_N$ denote vertical tensor fields, and let $\mf{A}_1, \dots, \mf{A}_N$ denote tensor fields on $\mi{I}$.
Then:
\begin{itemize}
\item If $\ms{A}_j \rightarrow^M \mf{A}_j$ for every $1 \leq j \leq N$, then there exists a tensor field $\mf{G}$ on $\mi{I}$ that satisfies $\mi{S} ( \gv; \ms{A}_1, \dots, \ms{A}_N ) \rightarrow^M \mf{G}$.
Moreover, $\mf{G}$ is tensor field defined by replacing every instance of $\ms{A}_1, \dots, \ms{A}_N, \gv, \gv^{-1}$ in $\mi{S} ( \gv; \ms{A}_1, \dots, \ms{A}_N )$ by $\mf{A}_1, \dots, \mf{A}_N, \gm, \gm^{-1}$, respectively.

\item If, in addition, $\ms{A}_j \Rightarrow^M \mf{A}_j$ for each $1 \leq j \leq N$, then $\mi{S} ( \gv; \ms{A}_1, \dots, \ms{A}_N ) \Rightarrow^M \mf{G}$.
\end{itemize}
\end{lemma}

\begin{proof}
This is an immediate consequence of the Definitions \ref{def.schematic_trivial} and \ref{def.schematic_vertical}, the limit properties in Proposition \ref{thm.limit_schematic}, and the limits \eqref{eq.fg_main_low_g} (which yield the restriction $M \leq M_0$).
\end{proof}

\begin{lemma} \label{thm.limit_deriv}
Let $0 \leq M \leq M_0$, let $\ms{A}$ be a vertical tensor field, and let $\mf{A}$ be a tensor field on $\mi{I}$ of the same rank as $\ms{A}$.
Then, the following statements hold:
\begin{itemize}
\item $\ms{A} \rightarrow^M \mf{A}$ if and only if $\Dv^m \ms{A} \rightarrow^0 \Dm^m \mf{A}$ for every $0 \leq m \leq M$.

\item $\ms{A} \Rightarrow^M \mf{A}$ if and only if $\Dv^m \ms{A} \Rightarrow^0 \Dm^m \mf{A}$ for every $0 \leq m \leq M$.
\end{itemize}
\end{lemma}

\begin{proof}
Since both statements are trivial when $M = 0$, we assume henceforth that $M > 0$.

First, if $\ms{A} \rightarrow^M \mf{A}$, then the first part of \eqref{eq.geom_bound_deriv} implies that
\begin{equation}
\label{eql.limit_deriv_0} \sum_{ m = 0 }^M | \Dv^m \ms{A} - \Dm^m \mf{A} |_{ 0, \varphi } \lesssim_C | \ms{A} - \mf{A} |_{ M, \varphi } + | \gv - \gm |_{ M, \varphi } \text{,}
\end{equation}
for any compact coordinate system $( U, \varphi )$ on $\mi{I}$.
Applying \eqref{eq.fg_main_low_g} to \eqref{eql.limit_deriv_0}, we obtain that $\Dv^m \ms{A} \rightarrow^0 \Dm^m \mf{A}$ for any $0 \leq m \leq M$.
Furthermore, using a similar argument with \eqref{eq.fg_main_low_g} and \eqref{eql.limit_deriv_0}, we conclude as well that $\ms{A} \Rightarrow^M \mf{A}$ implies $\Dv^m \ms{A} \Rightarrow^0 \Dm^m \mf{A}$ for every $0 \leq m \leq M$.

Similarly, for the reverse directions, assume first that $\Dv^m \ms{A} \rightarrow^0 \Dm^m \mf{A}$ for each $0 \leq m \leq M$.
Then, from the second part of \eqref{eq.geom_bound_deriv}, we obtain, for any compact coordinate system $( U, \varphi )$ on $\mi{I}$, that
\begin{align}
\label{eql.limit_deriv_1} | \ms{A} - \mf{A} |_{ m, \varphi } &\lesssim \sum_{ s = 0 }^m | \Dv^s \ms{A} - \Dm^s \mf{A} |_{ 0, \varphi } + ( | \ms{A} |_{ m - 1, \varphi } + | \mf{A} |_{ m - 1, \varphi } ) | \gv - \gm |_{ m, \varphi } \\
\notag &\lesssim \sum_{ s = 0 }^m | \Dv^s \ms{A} - \Dm^s \mf{A} |_{ 0, \varphi } + ( 1 + | \ms{A} - \mf{A} |_{ m - 1, \varphi } ) | \gv - \gm |_{ m, \varphi } \text{,}
\end{align}
for any $0 < m \leq M$.
Combining \eqref{eql.limit_deriv_1} with \eqref{eq.fg_main_low_g} and an induction argument over $m$ yields the desired limit $\ms{A} \rightarrow^M \ms{A}$.
Furthermore, if $\Dv^m \ms{A} \Rightarrow^0 \Dm^m \mf{A}$ for every $0 \leq m \leq M$, then the same estimate \eqref{eql.limit_deriv_1}, along with an analogous induction argument, implies $\ms{A} \Rightarrow^M \mf{A}$.
\end{proof}

\begin{lemma} \label{thm.limit_low_R}
The following limits hold:
\begin{equation}
\label{eq.limit_low_R} \mi{L}_\rho \Rv \Rightarrow^{ M_0 - 2 } 0 \text{,} \qquad \rho \mi{L}_\rho^2 \Rv \Rightarrow^{ M_0 - 2 } 0 \text{.}
\end{equation}
\end{lemma}

\begin{proof}
Since \eqref{eq.limit_low} and Lemma \ref{thm.limit_deriv} imply $\Dv^2 \mi{L}_\rho \gv \Rightarrow^{ M_0 - 2 } 0$, then Lemma \ref{thm.limit_schematic_metric} yields
\begin{equation}
\label{eql.limit_low_R_1} \mi{S} ( \gv; \Dv^2 \mi{L}_\rho \gv ) \Rightarrow^{ M_0 - 2 } 0 \text{.}
\end{equation}
As a result, the third equation in \eqref{eq.einstein_vertical_ex} and \eqref{eql.limit_low_R_1} yield the first limit in \eqref{eq.limit_low_R}.

For the remaining limit, we differentiate the third equation of \eqref{eq.einstein_vertical_ex} to obtain
\[
\rho \mi{L}_\rho^2 \Rv = \mi{S} ( \gv; \rho \mi{L}_\rho \Dv^2 \mi{L}_\rho \gv ) + \mi{S} ( \gv; \mi{L}_\rho \gv, \rho \Dv^2 \mi{L}_\rho \gv ) \text{.}
\]
Applying the commutation formula \eqref{eq.comm_vertical} twice to the first term on the right-hand side yields
\begin{align}
\label{eql.limit_low_R_2} \rho \mi{L}_\rho^2 \Rv &= \mi{S} ( \gv; \rho \Dv^2 \mi{L}^2_\rho \gv ) + \rho \cdot \mi{S} ( \gv; \Dv \mi{L}_\rho \gv, \Dv \mi{L}_\rho \gv ) + \rho \cdot \mi{S} ( \gv; \mi{L}_\rho \gv, \Dv^2 \mi{L}_\rho \gv ) \\
\notag &:= I_1 + I_2 + I_3 \text{.}
\end{align}
By Proposition \ref{thm.limit_bound}, \eqref{eq.limit_low}, Lemma \ref{thm.limit_schematic_metric}, and Lemma \ref{thm.limit_deriv}, the schematic terms satisfy
\[
I_1 \Rightarrow^{ M_0 - 2 } 0 \text{,} \qquad I_2 \Rightarrow^{ M_0 - 1 } 0 \text{,} \qquad I_3 \Rightarrow^{ M_0 - 2 } 0 \text{.}
\]
The final limit of \eqref{eq.limit_low_R} now follows from \eqref{eql.limit_low_R_2} and the above.
\end{proof}

\subsubsection{The Non-Anomalous Limits}

In Lemmas \ref{thm.limit_low} and \ref{thm.limit_low_R}, we obtained zeroth and first-order (in $\mi{L}_\rho$) boundary limits for $\gv$ and $\Rv$.
Next, we derive the boundary limits \eqref{eq.fg_main_high_depend} and \eqref{eq.fg_main_high}, that is, the limits up to (but not including) order $n$, before the anomalous logarithmic power.

We begin with the limits for $\gv$ and $\Rv$.
First, observe that the cases $k = 0$ and $k = 1$ in already follow from \eqref{eq.fg_main_low_g}, \eqref{eq.fg_main_low_R}, Lemma \ref{thm.limit_low}, and Lemma \ref{thm.limit_low_R}:
\begin{align}
\label{eq.limit_high_pre} \gv \Rightarrow^{ M_0 } \gm \text{,} \qquad \rho \mi{L}_\rho \gv \Rightarrow^{ M_0 } 0 \text{,} &\qquad \mi{L}_\rho \gv \Rightarrow^{ M_0 - 1 } 0 \text{,} \qquad \rho \mi{L}_\rho^2 \gv \Rightarrow^{ M_0 - 1 } 0 \text{,} \\
\notag \Rv \Rightarrow^{ M_0 - 2 } \Rm \text{,} \qquad \rho \mi{L}_\rho \Rv \Rightarrow^{ M_0 - 2 } 0 \text{,} &\qquad \mi{L}_\rho \Rv \Rightarrow^{ M_0 - 3 } 0 \text{,} \qquad \rho \mi{L}_\rho^2 \Rv \Rightarrow^{ M_0 - 3 } 0 \text{.}
\end{align}
For the remaining cases $2 \leq k < n$, we apply an induction argument over $k$:

\begin{lemma} \label{thm.limit_high}
Fix $2 \leq k < n$, and assume the following hold for all $0 \leq p < k$:
\begin{align}
\label{eq.limit_high_ass} \mi{L}_\rho^p \gv \Rightarrow^{ M_0 - p } \mc{D}^p ( \gm; p ) \text{,} &\qquad \rho \mi{L}_\rho^{ p + 1 } \gv \Rightarrow^{ M_0 - p } 0 \text{,} \\
\notag \mi{L}_\rho^p \Rv \Rightarrow^{ M_0 - p - 2 } \mc{D}^{ p + 2 } ( \gm; p ) \text{,} &\qquad \rho \mi{L}_\rho^{ p + 1 } \Rv \Rightarrow^{ M_0 - p - 2 } 0 \text{.}
\end{align}
Then, the following limits also hold:
\begin{align}
\label{eq.limit_high} \mi{L}_\rho^k \gv \Rightarrow^{ M_0 - k } \mc{D}^k ( \gm; k ) \text{,} &\qquad \rho \mi{L}_\rho^{ k + 1 } \gv \Rightarrow^{ M_0 - k } 0 \text{,} \\
\notag \mi{L}_\rho^k \Rv \Rightarrow^{ M_0 - k - 2 } \mc{D}^{ k + 2 } ( \gm; k ) \text{,} &\qquad \rho \mi{L}_\rho^{ k + 1 } \Rv \Rightarrow^{ M_0 - k - 2 } 0 \text{.}
\end{align}
\end{lemma}

\begin{proof}
We begin by claiming that
\begin{align}
\label{eql.limit_high_1} \mi{L}_\rho^{ k - 2 } \Rcv \Rightarrow^{ M_0 - k } \mc{D}^k ( \gm; k ) \text{,} &\qquad \rho \mi{L}_\rho^{ k - 1 } \Rcv \Rightarrow^{ M_0 - k } 0 \text{,} \\
\notag \mi{L}_\rho^{ k - 2 } \Rsv \Rightarrow^{ M_0 - k } \mc{D}^k ( \gm; k ) \text{,} &\qquad \rho \mi{L}_\rho^{ k - 1 } \Rsv \Rightarrow^{ M_0 - k } 0 \text{,}
\end{align}
The first two limits in \eqref{eql.limit_high_1} follow from the induction hypothesis \eqref{eq.limit_high_ass} (with $p := k - 2$) and from the observation that $\mi{L}_\rho$ commutes with (non-metric) contractions.
For $\Rsv$, we first write
\begin{align}
\label{eql.limit_high_2} \mi{L}_\rho^{ k - 2 } \Rsv &= \mi{S} ( \gv; \mi{L}_\rho^{ k - 2 } \Rv ) + \sum_{ \substack{ j + j_1 + \dots + j_l = k - 2 \\ 0 \leq j < k - 2 \\ 1 \leq j_p \leq k - 2 } } \mi{S} ( \gv; \mi{L}_\rho^j \Rv, \mi{L}_\rho^{ j_1 } \gv, \dots, \mi{L}_\rho^{ j_l } \gv ) \text{,} \\
\notag \rho \mi{L}_\rho^{ k - 1 } \Rsv &= \mi{S} ( \gv; \rho \mi{L}_\rho^{ k - 1 } \Rv ) + \sum_{ \substack{ j + j_1 + \dots + j_l = k - 1 \\ 0 \leq j < k - 1 \\ 1 \leq j_p \leq k - 1 } } \mi{S} ( \gv; \rho \mi{L}_\rho^j \Rv, \mi{L}_\rho^{ j_1 } \gv, \dots, \mi{L}_\rho^{ j_l } \gv ) \text{.}
\end{align}
Applying Lemma \ref{thm.limit_schematic_metric} and \eqref{eq.limit_high_ass} to \eqref{eql.limit_high_2} yields
\begin{equation}
\label{eql.limit_high_3} \mi{L}_\rho^{ k - 2 } \Rsv \Rightarrow^{ M_0 - k } \mc{D}^k ( \gm ) \text{,} \qquad \rho \mi{L}_\rho^{ k - 1 } \Rsv \Rightarrow^{ M_0 - k } 0 \text{,}
\end{equation}
Moreover, when $k$ is odd, Lemma \ref{thm.limit_schematic_metric} and \eqref{eq.limit_high_ass} also imply
\[
\mi{S} ( \gv; \mi{L}_\rho^{ k - 2 } \Rv ) \Rightarrow^{ M_0 - k } 0 \text{,} \qquad \sum_{ \substack{ j + j_1 + \dots + j_l = k - 2 \\ 0 \leq j < k - 2 \\ 1 \leq j_p \leq k - 2 } } \mi{S} ( \gv; \mi{L}_\rho^j \Rv, \mi{L}_\rho^{ j_1 } \gv, \dots, \mi{L}_\rho^{ j_l } \gv ) \Rightarrow^{ M_0 - k } 0 \text{.}
\]
(For the second limit, note that in each schematic term of the sum, one of $j, j_1, \dots, j_l$ must be odd.)
In particular, \eqref{eql.limit_high_2} and the above together imply the claim \eqref{eql.limit_high_1}.

Next, we claim the following limits:
\begin{align}
\label{eql.limit_high_11} \mc{S}_1 &:= \sum_{ \substack{ j_1 + \dots + j_l = k \\ 1 \leq j_p < k } } \mi{S} ( \gv; \mi{L}_\rho^{ j_1 } \gv, \dots, \mi{L}_\rho^{ j_l } \gv ) \Rightarrow^{ M_0 - k } \mc{D}^{ k - 1 } ( \gm; k ) \text{,} \\
\notag \mc{S}_2 &:= \sum_{ \substack{ j_1 + \dots + j_l = k + 1 \\ 1 \leq j_p \leq k } } \rho \cdot \mi{S} ( \gv; \mi{L}_\rho^{ j_1 } \gv, \dots, \mi{L}_\rho^{ j_l } \gv ) \Rightarrow^{ M_0 - k } 0 \text{.}
\end{align}
Similar to before, the limit for $\mc{S}_1 \Rightarrow^{ M_0 - k } \mc{D}^{ k - 1 } ( \gm )$ follows immediately from Lemma \ref{thm.limit_schematic_metric} and \eqref{eq.limit_high_ass}.
Moreover, when $k$ is odd, then within each schematic term comprising $\mc{S}_1$, one of $j_1, \dots, j_l$ must also be odd.
Therefore, Lemma \ref{thm.limit_schematic_metric} and \eqref{eq.limit_high_ass} imply that $\mc{S}_1 \Rightarrow^{ M_0 - k } \mc{D}^{ k - 1 } ( \gm; k )$.

For $\mc{S}_2$, we consider each term $\mc{S}_2^\ast$ in the corresponding summation in \eqref{eql.limit_high_11} separately.
First, if none of $j_1, \dots, j_l$ in $\mc{S}_2^\ast$ is equal to $k$, then Proposition \ref{thm.limit_bound} and the first part of \eqref{eql.limit_high_11} immediately imply $\mc{S}_2^\ast \Rightarrow^{ M_0 - k } 0$.
This leaves only the case $j_1 = k$ and $j_2 = 1$:
\[
\mc{S}_2^\ast = \mi{S} ( \gv; \rho \mi{L}_\rho^k \gv, \mi{L}_\rho \gv ) \text{.}
\]
Again, Lemma \ref{thm.limit_schematic_metric} and \eqref{eq.limit_high_ass} imply $\mc{S}_2^\ast \Rightarrow^{ M_0 - k } 0$ in this case.
The above results in the second limit of \eqref{eql.limit_high_11} and hence completes the proof of the claim \eqref{eql.limit_high_11}.

Applying Proposition \ref{thm.limit_positive} to the second equation of \eqref{eq.einstein_vertical_deriv} (with $\ms{A} := \trace{\gv} \mi{L}_\rho^k \gv$, while noting that $2 n - k > 0$), and recalling the limits \eqref{eql.limit_high_1} and \eqref{eql.limit_high_11}, we obtain that
\begin{equation}
\label{eql.limit_high_20} \trace{\gv} \mi{L}_\rho^k \gv \Rightarrow^{ M_0 - k } \mc{D}^k ( \gm; k ) \text{.}
\end{equation}
Moreover, applying Proposition \ref{thm.limit_positive} again to the first part of \eqref{eq.einstein_vertical_deriv} (with $\ms{A} := \mi{L}_\rho^k \gv$, while noting that $n - k > 0$), and recalling \eqref{eql.limit_high_1}, \eqref{eql.limit_high_11} and \eqref{eql.limit_high_20}, we obtain the first two limits of \eqref{eq.limit_high}.

Similarly, by Lemma \ref{thm.limit_schematic_metric}, Lemma \ref{thm.limit_deriv}, \eqref{eq.limit_high_ass}, and the first part of \eqref{eq.limit_high}, we obtain
\begin{equation}
\label{eql.limit_high_31} \sum_{ \substack{ i_1 + \dots + i_l = 2 \\ j_1 + \dots + j_l = k \\ j_p \geq 1 } } \mi{S} ( \gv; \Dv^{ i_1 } \mi{L}_\rho^{ j_1 } \gv, \dots, \Dv^{ i_l } \mi{L}_\rho^{ j_l } \gv ) \Rightarrow^{ M_0 - k - 2 } \mc{D}^{ k + 2 } ( \gm ) \text{.}
\end{equation}
(Both $k + 2$ and $M_0 - k - 2$ arise from schematic terms of the form $\mi{S} ( \Dv^2 \mi{L}_\rho^k \gv )$.)
Again, when $k$ is odd, then one of $j_1, \dots, j_l$ in each schematic term in \eqref{eql.limit_high_31} must also be odd, and hence Lemma \ref{thm.limit_deriv} and \eqref{eq.limit_high_ass} imply that the right-hand side of \eqref{eql.limit_high_31} vanishes.
Combining this observation with the first part of \eqref{eq.curvature_vertical_deriv} and \eqref{eql.limit_high_31} yields the third limit in \eqref{eq.limit_high}.

For the remaining limit of \eqref{eq.limit_high}, we first claim that
\begin{equation}
\label{eql.limit_high_41} \sum_{ \substack{ i_1 + \dots + i_l = 2 \\ j_1 + \dots + j_l = k + 1 \\ j_p \geq 1 } } \rho \cdot \mi{S} ( \Dv^{ i_1 } \mi{L}_\rho^{ j_1 } \gv, \dots, \Dv^{ i_l } \mi{L}_\rho^{ j_l } \gv ) \Rightarrow^{ M_0 - k - 2 } 0 \text{.}
\end{equation}
To show \eqref{eql.limit_high_41}, we consider each term $\mc{S}^\ast$ of the sum separately.
If none of $j_1, \dots, j_l$ is equal to $k + 1$, then $\mc{S}^\ast \Rightarrow^{ M_0 - k - 2 } 0$ by Proposition \ref{thm.limit_bound} and \eqref{eql.limit_high_31}.
The only remaining possibility is
\[
\mc{S}^\ast = \mi{S} ( \gv; \rho \Dv^2 \mi{L}_\rho^{ k + 1 } \gv ) \text{,}
\]
which also satisfies $\mc{S}^\ast \Rightarrow^{ M_0 - k - 2 } 0$ by Lemma \ref{thm.limit_schematic_metric} and the second limit of \eqref{eq.limit_high}.
This completes the proof of the claim \eqref{eql.limit_high_41}.
Finally, combining \eqref{eql.limit_high_41} with the first equation in \eqref{eq.curvature_vertical_deriv}---with $k + 1$ in the place of $k$---we obtain the final limit in \eqref{eq.limit_high}.
\end{proof}

Now, from the base case \eqref{eq.limit_high_pre} and the inductive case proved in Lemma \ref{thm.limit_high}, we obtain the first two relations in \eqref{eq.fg_main_high_depend} and the first four limits of \eqref{eq.fg_main_high}.

Furthermore, when $n > 2$, we can take the limit as $\rho \searrow 0$ of the first equation in \eqref{eq.einstein_vertical_deriv}, with $k = 2$.
Applying \eqref{eq.fg_main_low_g}, \eqref{eq.fg_main_low_R}, and the first two limits of \eqref{eq.fg_main_high} (with $k = 2$) results in the equation
\begin{equation}
\label{eq.limit_high_schouten} 0 = - 2 ( n - 2 ) \cdot \gb{2} - \trace{\gm} \gb{2} \cdot \gm - 2 \cdot \Rcm \text{.}
\end{equation}
Taking a ($\gm$-)trace of \eqref{eq.limit_high_schouten} yields
\[
- ( n - 1 ) \cdot \trace{\gm} \gb{2} = \Rsm \text{.}
\]
Substituting the above into \eqref{eq.limit_high_schouten} yields the first part of \eqref{eq.fg_main_schouten}.

It remains to prove the last relation in \eqref{eq.fg_main_high_depend}, along with the last two limits of \eqref{eq.fg_main_high}.
For this, we commute derivatives and rely on already established limits.
We begin by applying \eqref{eq.comm_vertical_deriv}:
\begin{align}
\label{eq.limit_high_DLg_0} \mi{L}_\rho^{ l - 1 } ( \Dv \mi{L}_\rho \gv ) &= \Dv \mi{L}_\rho^l \gv + \sum_{ \substack{ j + j_0 + \dots + j_k = l - 1 \\ 0 \leq j < l - 1 \text{, } j_p \geq 1 } } \mi{S} ( \gv; \mi{L}_\rho^j ( \mi{L}_\rho \gv ), \Dv \mi{L}_\rho^{ j_0 } \gv, \mi{L}_\rho^{ j_1 } \gv, \dots, \mi{L}_\rho^{ j_k } \gv ) \\
\notag &= \Dv \mi{L}_\rho^l \gv + \sum_{ \substack{ j_0 + \dots + j_k = l \\ 1 \leq j_p < l } } \mi{S} ( \gv; \Dv \mi{L}_\rho^{ j_0 } \gv, \mi{L}_\rho^{ j_1 } \gv, \dots, \mi{L}_\rho^{ j_k } \gv ) \text{.}
\end{align}
Applying Lemma \ref{thm.limit_schematic_metric}, Lemma \ref{thm.limit_deriv}, and the first part of \eqref{eq.fg_main_high} to the above yields
\begin{equation}
\label{eq.limit_high_DLg_1} \mi{L}_\rho^{ l - 1 } ( \Dv \mi{L}_\rho \gv ) \Rightarrow^{ M_0 - l - 1 } \mc{D}^{ l + 1 } ( \gm; l ) \text{.}
\end{equation}
(In particular, when $l$ is odd, then in each schematic term on the right-hand side of \eqref{eq.limit_high_DLg_0}, one of $j_0, \dots, j_k$ must also be odd, hence this term converges to $0$.)

Next, applying \eqref{eq.comm_vertical_deriv} again, we have that
\[
\rho \mi{L}_\rho^l ( \Dv \mi{L}_\rho \gv ) = \rho \cdot \Dv \mi{L}_\rho^{ l + 1 } \gv + \sum_{ \substack{ j_0 + \dots + j_k = l + 1 \\ 1 \leq j_p < l + 1 } } \rho \cdot \mi{S} ( \gv; \Dv \mi{L}_\rho^{ j_0 } \gv, \mi{L}_\rho^{ j_1 } \gv, \dots, \mi{L}_\rho^{ j_k } \gv ) \text{.}
\]
Again, Lemma \ref{thm.limit_schematic_metric}, Lemma \ref{thm.limit_deriv}, and the first two parts of \eqref{eq.fg_main_high} imply that
\begin{equation}
\label{eq.limit_high_DLg_2} \rho \mi{L}_\rho^l ( \Dv \mi{L}_\rho \gv ) \Rightarrow^{ M_0 - l - 1 } 0 \text{.}
\end{equation}
(In the schematic terms, $\rho$ can always be paired with the factor with the highest derivative of $\gv$.)

In particular, the limits \eqref{eq.limit_high_DLg_1} and \eqref{eq.limit_high_DLg_2} complete the proofs of \eqref{eq.fg_main_high_depend} and \eqref{eq.fg_main_high}.

\subsubsection{The Anomalous and Free Limits}

Lastly, we consider the boundary limits at order $n$.
We establish the remaining boundary limits \eqref{eq.fg_main_top_depend}, \eqref{eq.fg_main_top}, and \eqref{eq.fg_main_free}, which include both anomalous logarithmic limits (when $n$ is even) determined entirely by $\gm$ and non-logarithmic limits that need not depend on $\gm$.
In this process, we will also establish the relations \eqref{eq.fg_main_top_constraint} and \eqref{eq.fg_main_free_constraint}.

We begin by recalling the equations \eqref{eq.einstein_vertical_deriv} in the special case $k = n$:
\begin{align}
\label{eq.limit_top_g_1} 0 &= \rho \mi{L}_\rho^{ n + 1 } \gv - \trace{ \gv } \mi{L}^n_\rho \gv \cdot \gv - 2 ( n - 1 ) \mi{L}_\rho^{ n - 2 } \Rcv - 2 \rho \mi{L}_\rho^{ n - 1 } \Rcv \\
\notag &\qquad + \sum_{ \substack{ j_1 + \dots + j_l = n \\ 1 \leq j_p < n } } \mi{S} ( \gv; \mi{L}_\rho^{ j_1 } \gv, \dots, \mi{L}_\rho^{ j_l } \gv ) + \sum_{ \substack{ j_1 + \dots + j_l = n + 1 \\ 1 \leq j_p \leq n } } \rho \cdot \mi{S} ( \gv; \mi{L}_\rho^{ j_1 } \gv, \dots, \mi{L}_\rho^{ j_l } \gv ) \text{,} \\
\notag 0 &= \rho \mi{L}_\rho ( \trace{\gv} \mi{L}_\rho^n \gv ) - n \cdot \trace{\gv} \mi{L}_\rho^n \gv - 2 ( n - 1 ) \mi{L}_\rho^{ n - 2 } \Rsv - 2 \rho \mi{L}_\rho^{ n - 1 } \Rsv \\
\notag &\qquad + \sum_{ \substack{ j_1 + \dots + j_l = n \\ 1 \leq j_p < n } } \mi{S} ( \gv; \mi{L}_\rho^{ j_1 } \gv, \dots, \mi{L}_\rho^{ j_l } \gv ) + \sum_{ \substack{ j_1 + \dots + j_l = n + 1 \\ 1 \leq j_p \leq n } } \rho \cdot \mi{S} ( \gv; \mi{L}_\rho^{ j_1 } \gv, \dots, \mi{L}_\rho^{ j_l } \gv ) \text{.}
\end{align}
Applying \eqref{eq.fg_main_high_depend} and \eqref{eq.fg_main_high} (and noting that $\mi{L}_\rho$ commutes with contractions) yields
\begin{equation}
\label{eq.limit_top_g_11} \mi{L}_\rho^{ n - 2 } \Rcv \Rightarrow^{ M_0 - n } \mc{D}^n ( \gm; n ) \text{,} \qquad \rho \mi{L}_\rho^{ n - 1 } \Rcv \Rightarrow^{ M_0 - n } 0 \text{.}
\end{equation}
Moreover, applying \eqref{eq.fg_main_high_depend}, \eqref{eq.fg_main_high}, and Lemma \ref{thm.limit_schematic_metric}, we obtain
\begin{align}
\label{eq.limit_top_g_12} \mi{L}_\rho^{ n - 2 } \Rsv &= \mi{S} ( \gv; \mi{L}_\rho^{ n - 2 } \Rv ) + \sum_{ \substack{ j + j_1 + \dots + j_l = n - 2 \\ 0 \leq j < n - 2 \\ 1 \leq j_p \leq n - 2 } } \mi{S} ( \gv; \mi{L}_\rho^j \Rv, \mi{L}_\rho^{ j_1 } \gv, \dots, \mi{L}_\rho^{ j_l } \gv ) \Rightarrow^{ M_0 - n } \mc{D}^n ( \gm; n ) \text{,} \\
\notag \rho \mi{L}_\rho^{ n - 1 } \Rsv &= \mi{S} ( \gv; \rho \mi{L}_\rho^{ n - 1 } \Rv ) + \sum_{ \substack{ j + j_1 + \dots + j_l = n - 1 \\ 0 \leq j < n - 1 \\ 1 \leq j_p \leq n - 1 } } \rho \cdot \mi{S} ( \gv; \mi{L}_\rho^j \Rv, \mi{L}_\rho^{ j_1 } \gv, \dots, \mi{L}_\rho^{ j_l } \gv ) \Rightarrow^{ M_0 - n } 0 \text{.}
\end{align}
(For the schematic terms in the first part of \eqref{eq.limit_top_g_12}, note when $n$ is odd, one of $j, j_1, \dots, j_n$ must also be odd.)
By similar reasoning, \eqref{eq.fg_main_high_depend}, \eqref{eq.fg_main_high}, and Lemma \ref{thm.limit_schematic_metric} also yield the limits
\begin{align}
\label{eq.limit_top_g_13} \sum_{ \substack{ j_1 + \dots + j_l = n \\ 1 \leq j_p < n } } \mi{S} ( \gv; \mi{L}_\rho^{ j_1 } \gv, \dots, \mi{L}_\rho^{ j_l } \gv ) &\Rightarrow^{ M_0 - n } \mc{D}^{ n - 1 } ( \gm; n ) \text{,} \\
\notag \sum_{ \substack{ j_1 + \dots + j_l = n + 1 \\ 1 \leq j_p \leq n } } \rho \cdot \mi{S} ( \gv; \mi{L}_\rho^{ j_1 } \gv, \dots, \mi{L}_\rho^{ j_l } \gv ) &\Rightarrow^{ M_0 - n } 0 \text{.} 
\end{align}

Applying Proposition \ref{thm.limit_positive} to the second part of \eqref{eq.limit_top_g_1} and using \eqref{eq.limit_top_g_12} and \eqref{eq.limit_top_g_13}, we see that
\begin{equation}
\label{eq.limit_top_g_20} \trace{\gv} \mi{L}_\rho^n \gv \Rightarrow^{ M_0 - n } \mc{D}^n ( \gm; n ) \text{,} \qquad \rho \mi{L}_\rho ( \trace{\gv} \mi{L}_\rho^n \gv ) \Rightarrow^{ M_0 - n } 0 \text{.}
\end{equation}
Thus, applying \eqref{eq.fg_main_low_g}, \eqref{eq.limit_top_g_11}, \eqref{eq.limit_top_g_13}, and the first part of \eqref{eq.limit_top_g_20} to the first identity in \eqref{eq.limit_top_g_1} yields
\begin{equation}
\label{eq.limit_top_g_21} \rho \mi{L}_\rho^{ n + 1 } \gv \Rightarrow^{ M_0 - n } n! \, \gb{\star} = \mc{D}^n ( \gm; n ) \text{,}
\end{equation}
which is the first identity of \eqref{eq.fg_main_top_depend} and first limit of \eqref{eq.fg_main_top}.

Moreover, note that \eqref{eq.fg_main_high_depend}, \eqref{eq.fg_main_high}, Lemma \ref{thm.limit_schematic_metric}, and \eqref{eq.limit_top_g_20} imply
\begin{equation}
\label{eq.limit_top_g_22} \trace{\gv} ( \rho \mi{L}_\rho^{ n + 1 } \gv ) = \rho \mi{L}_\rho ( \trace{\gv} \mi{L}_\rho^n \gv ) + \mi{S} ( \gv; \rho \mi{L}_\rho^n \gv, \mi{L}_\rho \gv ) \Rightarrow^{ M_0 - n } 0 \text{.}
\end{equation}
Combining the above with \eqref{eq.fg_main_low_g}, \eqref{eq.fg_main_high}, and \eqref{eq.limit_top_g_21} results in the first identity in \eqref{eq.fg_main_top_constraint}.

In addition, applying Proposition \ref{thm.limit_zero} to \eqref{eq.limit_top_g_21}, we obtain a tensor field $\gb{\dagger}$ on $\mi{I}$ such that
\begin{equation}
\label{eq.limit_top_g_30} \mi{L}_\rho^n \gv - n! \, ( \log \rho ) \gb{\star} \rightarrow^{ M_0 - n } n! \, \gb{\dagger} \text{,}
\end{equation}
which is the first limit of \eqref{eq.fg_main_free}.
Taking a $\gv$-trace of \eqref{eq.limit_top_g_30} and recalling \eqref{eq.fg_main_low_g} then yields
\begin{equation}
\label{eq.limit_top_g_31} \trace{\gv} \mi{L}_\rho^n \gv - n! \, ( \log \rho ) \cdot \trace{\gv} \gb{\star} \rightarrow^{ M_0 - n } n! \, \trace{\gm} \gb{\dagger} \text{.}
\end{equation}
By \eqref{eq.fg_main_low_g}, \eqref{eq.fg_main_high}, the first identity in \eqref{eq.fg_main_top_constraint}, and the fundamental theorem of calculus, we have
\begin{align*}
| ( \log \rho ) \cdot \trace{\gv} \gb{\star} | &= | \log \rho | | ( \gv^{ a b } - \gm^{ a b } ) \gb{\star}_{ a b } | \\
&\lesssim | \log \rho | \int_0^\rho | \mi{L}_\rho \gv |_{ 0, \varphi } |_\sigma d \sigma \cdot \| \gb{\star} \|_{ 0, \varphi } \\
&\lesssim \rho \cdot \log \rho \text{,}
\end{align*}
with respect to any compact coordinate system $( U, \varphi )$ on $\mi{I}$.
As a result,
\[
( \log \rho ) \cdot \trace{\gv} \gb{\star} \rightarrow^0 0 \text{.}
\]
From the first part of \eqref{eq.limit_top_g_20}, \eqref{eq.limit_top_g_31}, and the above, we obtain the first identity of \eqref{eq.fg_main_free_constraint}.

Next, consider the second identity of \eqref{eq.curvature_vertical_deriv}, in the case $k = n + 1$:
\begin{equation}
\label{eq.limit_top_g_40} \rho ( \Dv \cdot \mi{L}_\rho^{ n + 1 } \gv ) = \rho \cdot \Dv ( \trace{\gv} \mi{L}^{ n + 1 }_\rho \gv ) + \sum_{ \substack{ j_0 + \dots + j_l = n + 1 \\ 1 \leq j_p < n + 1 } } \rho \cdot \mi{S} ( \gv; \Dv \mi{L}_\rho^{ j_0 } \gv, \mi{L}_\rho^{ j_1 } \gv, \dots, \mi{L}_\rho^{ j_l } \gv ) \text{.}
\end{equation}
We now let $\rho \searrow 0$ in \eqref{eq.limit_top_g_40}.
Recalling \eqref{eq.fg_main_low_g}, Lemma \ref{thm.limit_deriv}, \eqref{eq.limit_top_g_21}, and \eqref{eq.limit_top_g_22} yields
\begin{equation}
\label{eq.limit_top_g_41} \rho ( \Dv \cdot \mi{L}_\rho^{ n + 1 } \gv ) \rightarrow^0 \Dm \cdot \gb{\star} \text{,} \qquad \rho \cdot \Dv ( \trace{\gv} \mi{L}^{ n + 1 }_\rho \gv ) \rightarrow^0 0 \text{.}
\end{equation}
Furthermore, \eqref{eq.fg_main_high}, Lemma \ref{thm.limit_schematic_metric}, and Lemma \ref{thm.limit_deriv} imply that
\[
\sum_{ \substack{ j_0 + \dots + j_l = n + 1 \\ 1 \leq j_p < n + 1 } } \rho \cdot \mi{S} ( \gv; \Dv \mi{L}_\rho^j \gv, \mi{L}_\rho^{ j_1 } \gv, \dots, \mi{L}_\rho^{ j_l } \gv ) \rightarrow^0 0 \text{.}
\]
(In particular, one always pairs the factor of $\rho$ with an instance of $\mi{L}_\rho^n \gv$.)
Combining \eqref{eq.limit_top_g_40}, \eqref{eq.limit_top_g_41}, and the above then yields the second identity in \eqref{eq.fg_main_top_constraint}; this completes the proof of \eqref{eq.fg_main_top_constraint}.

Note in particular that when $n = 2$, the first equation in \eqref{eq.einstein_vertical_deriv}, with $k = 2$, simplifies to
\begin{align}
\label{eq.limit_top_g_60} 0 &= \rho \mi{L}_\rho^3 \gv - \trace{ \gv } \mi{L}^2_\rho \gv \cdot \gv - 2 \Rcv - 2 \rho \mi{L}_\rho \Rcv + \mi{S} ( \gv; \mi{L}_\rho \gv, \mi{L}_\rho \gv ) \\
\notag &\qquad + \mi{S} ( \gv; \rho \mi{L}_\rho^2 \gv, \mi{L}_\rho \gv ) + \rho \cdot \mi{S} ( \gv; \mi{L}_\rho \gv, \mi{L}_\rho \gv, \mi{L}_\rho \gv ) \text{,} \\
\notag 0 &= \rho \mi{L}_\rho ( \trace{\gv} \mi{L}_\rho^2 \gv ) - 2 \trace{\gv} \mi{L}_\rho^2 \gv - 2 \Rsv - 2 \rho \mi{L}_\rho \Rsv + \mi{S} ( \gv; \mi{L}_\rho \gv, \mi{L}_\rho \gv ) \\
\notag &\qquad + \mi{S} ( \gv; \rho \mi{L}_\rho^2 \gv, \mi{L}_\rho \gv ) + \rho \cdot \mi{S} ( \gv; \mi{L}_\rho \gv, \mi{L}_\rho \gv, \mi{L}_\rho \gv ) \text{.}
\end{align}
Applying Proposition \ref{thm.limit_positive}, \eqref{eq.fg_main_low_R}, \eqref{eq.fg_main_high}, and Lemma \ref{thm.limit_schematic_metric} to the second equation in \eqref{eq.limit_top_g_60} yields
\[
\trace{ \gv } \mi{L}^2_\rho \gv \rightarrow_0 - \Rsm \text{.}
\]
Moreover, taking $\rho \searrow 0$ in the first part of \eqref{eq.limit_top_g_60} and recalling \eqref{eq.limit_top_g_21} and the above, we obtain
\begin{equation}
\label{eq.limit_top_g_61} 0 = 2 \gb{\star} + \Rsm \cdot \gm - 2 \Rcm = 2 \gb{\star} \text{,} \qquad n = 2 \text{,}
\end{equation}
where the last equality follows since $\Rcm$ and $\Rsm$ represent the curvature on a $2$-dimensional manifold.

Consider again the second equation in \eqref{eq.curvature_vertical_deriv}, now with general $n$ and $k = n$:
\begin{equation}
\label{eq.limit_top_g_50} \Dv \cdot \mi{L}_\rho^n \gv = \Dv ( \trace{\gv} \mi{L}^n_\rho \gv ) + \sum_{ \substack{ j_0 + \dots + j_l = n \\ 1 \leq j_p < n } } \mi{S} ( \gv; \Dv \mi{L}_\rho^{ j_0 } \gv, \mi{L}_\rho^{ j_1 } \gv, \dots, \mi{L}_\rho^{ j_l } \gv ) \text{.}
\end{equation}
By Lemma \ref{thm.limit_deriv} and the first part of \eqref{eq.limit_top_g_20}, we have that
\begin{equation}
\label{eq.limit_top_g_51} \Dv ( \trace{\gv} \mi{L}^n_\rho \gv ) \Rightarrow^{ M_0 - n - 1 } \mc{D}^{ n + 1 } ( \gm; n ) \text{.}
\end{equation}
Furthermore, applying \eqref{eq.fg_main_high}, Lemma \ref{thm.limit_schematic_metric}, and Lemma \ref{thm.limit_deriv}, we observe that
\begin{equation}
\label{eq.limit_top_g_52} \sum_{ \substack{ j_0 + \dots + j_l = n \\ 1 \leq j_p < n } } \mi{S} ( \gv; \Dv \mi{L}_\rho^{ j_0 } \gv, \mi{L}_\rho^{ j_1 } \gv, \dots, \mi{L}_\rho^{ j_l } \gv ) \Rightarrow^{ M_0 - n - 1 } \mc{D}^n ( \gm; n ) \text{.}
\end{equation}
From \eqref{eq.limit_top_g_50}--\eqref{eq.limit_top_g_52}, we then conclude that
\begin{equation}
\label{eq.limit_top_g_53} \Dv \cdot \mi{L}_\rho^n \gv \Rightarrow^{ M_0 - n - 1 } \mc{D}^{ n + 1 } ( \gm; n ) \text{.}
\end{equation}

On the other hand, Lemma \ref{thm.limit_deriv} and \eqref{eq.limit_top_g_30} imply
\[
\Dv \mi{L}_\rho^n \gv - ( \log \rho ) \Dv \gb{\star} \rightarrow^0 \Dm \gb{\dagger} \text{.}
\]
Taking a $\gv$-trace of the above and recalling both \eqref{eq.fg_main_low_g} and the second part of \eqref{eq.fg_main_top_constraint}, we obtain
\begin{equation}
\label{eq.limit_top_g_54} \gv^{ b c } \Dv_b \mi{L}_\rho^n \gv_{ a c } - ( \log \rho ) \gv^{ b c } ( \Dv_b \gb{\star}_{ a c } - \Dm_b \gb{\star}_{ a c } ) - ( \log \rho ) ( \gv^{ b c } - \gm^{ b c } ) \Dm_b \gb{\star}_{ a c } \rightarrow^0 \gm^{ b c } \Dm_b \gb{\dagger}_{ a c } \text{,}
\end{equation}
where the indices are with respect to any compact coordinate system $( U, \varphi )$ on $\mi{I}$.
Again, by the fundamental theorem of calculus, the first inequality of \eqref{eq.geom_bound_deriv}, \eqref{eq.fg_main_low_g}, and \eqref{eq.fg_main_high},
\begin{align*}
| ( \log \rho ) \gv^{ b c } ( \Dv_b \gb{\star}_{ a c } - \Dm_b \gb{\star}_{ a c } ) | &\lesssim | \log \rho | \| \gv - \gm \|_{ 1, \varphi } \lesssim | \rho \cdot \log \rho | \| \mi{L}_\rho \gv \|_{ 1, \varphi } \lesssim | \rho \cdot \log \rho | \text{,} \\
| ( \log \rho ) ( \gv^{ b c } - \gm^{ b c } ) \Dm_b \gb{\star}_{ a c } | &\lesssim | \rho \cdot \log \rho | \| \mi{L}_\rho \gv \|_{ 0, \varphi } \| \Dm \gb{\star} \|_{ 0, \varphi } \lesssim | \rho \cdot \log \rho | \text{,}
\end{align*}
both of which vanish as $\rho \searrow 0$.
Finally, combining \eqref{eq.limit_top_g_53}, \eqref{eq.limit_top_g_54}, and the above yields the second identity of \eqref{eq.fg_main_free_constraint}; this now completes the proof of \eqref{eq.fg_main_free_constraint}.

For $\Rv$, we begin with the first equation of \eqref{eq.curvature_vertical_deriv}, in the case $k = n + 1$:
\begin{equation}
\label{eq.limit_top_R_0} 0 = \rho \mi{L}_\rho^{ n + 1 } \Rv + \sum_{ \substack{ i_1 + \dots + i_l = 2 \\ j_1 + \dots + j_l = n + 1 \\ j_p \geq 1 } } \rho \cdot \mi{S} ( \gv; \Dv^{ i_1 } \mi{L}_\rho^{ j_1 } \gv, \dots, \Dv^{ i_l } \mi{L}_\rho^{ j_l } \gv ) \text{.}
\end{equation}
We deal with each schematic term in \eqref{eq.limit_top_R_0}.
First, if $j_1, \dots, j_n$ are all strictly less than $n$, then
\begin{equation}
\label{eq.limit_top_R_free_1} \rho \cdot \mi{S} ( \gv; \Dv^{ i_1 } \mi{L}_\rho^{ j_1 } \gv, \dots, \Dv^{ i_l } \mi{L}_\rho^{ j_l } \gv ) \Rightarrow^{ M_0 - n - 2 } 0 \text{,}
\end{equation}
by Lemma \ref{thm.limit_schematic_metric} and \eqref{eq.fg_main_high}.
This leaves only two remaining cases, which we can then treat in the following manner using Lemma \ref{thm.limit_schematic_metric}, \eqref{eq.fg_main_high}, \eqref{eq.fg_main_top}, and \eqref{eq.limit_top_g_61} (when $n = 2$):
\begin{align}
\label{eq.limit_top_R_2} \mi{S} ( \gv; \rho \cdot \Dv^{ i_1 } \mi{L}_\rho^n \gv, \Dv^{ i_2 } \mi{L}_\rho \gv ) &\Rightarrow^{ M_0 - n - 2 } 0 \text{,} \\
\notag \mi{S} ( \gv; \rho \cdot \Dv^2 \mi{L}_\rho^{ n + 1 } \gv ) &\Rightarrow^{ M_0 - n - 2 } \begin{cases} \mc{D}^{ n + 2 } ( \gm; n ) & n > 2 \text{,} \\ 0 & n = 2 \text{.} \end{cases}
\end{align}
From Lemma \ref{thm.limit_deriv} and \eqref{eq.limit_top_R_0}--\eqref{eq.limit_top_R_2}, we obtain the second identity in \eqref{eq.fg_main_top_depend} and the second limit of \eqref{eq.fg_main_top}.
In particular, when $n = 2$, the above points also imply that
\begin{equation}
\label{eq.limit_top_R_3} \Rb{\star} = 0 \text{,} \qquad n = 2 \text{.}
\end{equation}
Moreover, applying Proposition \ref{thm.limit_zero} to this limit yields the second limit in \eqref{eq.fg_main_free}.

For $\Dv \mi{L}_\rho \gv$, we obtain, from the commutation formula \eqref{eq.comm_vertical_deriv},
\[
\rho \mi{L}_\rho^n ( \Dv \mi{L}_\rho \gv ) = \rho \cdot \Dv \mi{L}_\rho^{ n + 1 } \gv + \sum_{ \substack{ j_0 + \dots + j_l = n + 1 \\ 1 \leq j_p < n + 1 } } \rho \cdot \mi{S} ( \gv; \Dv \mi{L}_\rho^{ j_0 } \gv, \mi{L}_\rho^{ j_1 } \gv, \dots, \mi{L}_\rho^{ j_l } \gv ) \text{.}
\]
Applying Lemma \ref{thm.limit_schematic_metric}, Lemma \ref{thm.limit_deriv}, \eqref{eq.fg_main_high}, \eqref{eq.fg_main_top}, and \eqref{eq.limit_top_g_61}, we see that
\begin{equation}
\label{eq.limit_top_b} \rho \mi{L}_\rho^n ( \Dv \mi{L}_\rho \gv ) \Rightarrow^{ M_0 - n - 1 } ( n - 1 ) ! \, \bb{\star} = \begin{cases} \mc{D}^{ n + 1 } ( \gm; n ) & n > 2 \text{,} \\ 0 & n = 2 \text{.} \end{cases}
\end{equation}
This proves the last identity of \eqref{eq.fg_main_top_depend}, as well as the last limit of \eqref{eq.fg_main_top}.
Furthermore, applying Proposition \ref{thm.limit_zero} to \eqref{eq.limit_top_b} yields the final limit in \eqref{eq.fg_main_free}.

Finally, when $n = 2$, we combine \eqref{eq.limit_top_g_61}, \eqref{eq.limit_top_R_3}, and \eqref{eq.limit_top_b} to obtain the second part of \eqref{eq.fg_main_schouten}.
This completes the proof of \eqref{eq.fg_main_schouten}, as well as of Theorem \ref{thm.fg_main} altogether.

\subsection{Schwarzschild-AdS Spacetimes} \label{sec.fg_schw_ads}

In this subsection, we connect the partial expansions obtained from Theorem \ref{thm.fg_exp} to the special family of Schwarzschild-AdS spacetimes.
We first recall the definition of these spacetimes (near the conformal boundary) in standard coordinates:

\begin{definition} \label{def.schw_ads}
Fix $M \in \R$ and $n > 1$, and consider the spacetime $( \mi{M}_M, g_M )$, where:
\begin{itemize}
\item There is some $r_0 \gg 1$ such that $\mi{M}_M$ is of the form
\begin{equation}
\label{eq.schw_ads_manifold} \mi{M}_M := [ r_0, \infty ) \times \R \times \Sph^{ n - 1 } \text{.}
\end{equation}

\item The metric $g_M$ is given by
\begin{equation}
\label{eq.schw_ads_metric} g_M := \left( r^2 + 1 - \frac{ M }{ r^{ n - 2 } } \right)^{-1} dr^2 - \left( r^2 + 1 - \frac{ M }{ r^{ n - 2 } } \right) dt^2 + r^2 \mathring{\gamma} \text{,}
\end{equation}
where $r$ and $t$ are the projections to the first and second components of $\mc{M}_M$, and where $\mathring{\gamma}$ denotes the unit round metric on the last $n - 1$ components of $\mi{M}_M$.
\end{itemize}
We refer to $( \mi{M}_M, g_M )$ as an \emph{Schwarzschild-AdS outer segment}, with \emph{mass} $M$.
\end{definition}

It is well-known that $( \mi{M}_M, g_M )$ is a solution of the Einstein-vacuum equations (normalized with $\Lambda$ given by \eqref{eq.einstein}).
The goal of this discussion is to express $( \mi{M}_M, g_M )$ as a vacuum FG-aAdS segment and thus read off the partial expansion for $g_M$ from the conformal boundary.
The first of the above objectives is accomplished in the subsequent proposition:

\begin{proposition} \label{thm.schw_ads_aux}
Fix $M \in \R$ and $n > 1$, and let the spacetime $( \mi{M}_M, g_M )$ be as in Definition \ref{def.schw_ads}.
Then, $( \mi{M}_M, g_M )$ is isometric to a spacetime $( \mi{M}, g )$ satisfying the following:
\begin{itemize}
\item There is some $\rho_0 \ll 1$ such that
\begin{equation}
\label{eq.schw_ads_aux_manifold} \mi{M} = ( 0, \rho_0 ] \times \mi{I} \text{,} \qquad \mi{I} = \R \times \Sph^{ n - 1 } \text{,} \qquad \rho_0 \ll 1 \text{,}
\end{equation}

\item The metric $g$ is of the form
\begin{align}
\label{eq.schw_ads_aux_metric} \rho^2 g &= d \rho^2 - \left[ 1 + \frac{1}{2} \rho^2 + \frac{1}{16} \rho^4 - \frac{ ( n - 1 ) M }{ n } \cdot \rho^n + \alpha_1 ( \rho ) \cdot \rho^{ n + 2 } \right] dt^2 \\
\notag &\qquad + \left[ 1 - \frac{1}{2} \rho^2 + \frac{1}{16} \rho^4 + \frac{ M }{ n } \cdot \rho^n + \alpha_2 ( \rho ) \cdot \rho^{ n + 2 } \right] \mathring{\gamma} \text{,}
\end{align}
where $\rho$ and $t$ are the projections to the first and second components of $\mi{M}$, where $\mathring{\gamma}$ denotes the unit round metric on the last $n - 1$ components of $\mi{M}_M$, and where the remainders $\alpha_1, \alpha_2: [ -\rho_0, \rho_0 ] \rightarrow \R$ are real-analytic functions.

\item In particular, when $n = 2$ and $n = 4$, we have explicit expressions for $g$:
\begin{align}
\label{eq.schw_ads_aux_low} \rho^2 g = \begin{cases}
d \rho^2 - \left[ 1 - \frac{ M - 1 }{2} \cdot \rho^2 + \frac{ ( M - 1 )^2 }{ 16 } \cdot \rho^4 \right] dt^2 & \\
\qquad + \left[ 1 + \frac{ M - 1 }{2} \cdot \rho^2 + \frac{ ( M - 1 )^2 }{ 16 } \cdot \rho^4 \right] \mathring{\gamma} & n = 2 \text{,} \\
d \rho^2 - \left( 1 - \frac{ 1 + 4 M }{ 16 } \cdot \rho^4 \right)^2 \left( 1 - \frac{1}{2} \rho^2 + \frac{ 1 + 4 M }{ 16 } \cdot \rho^4 \right)^{-1} dt^2 & \\
\qquad + \left( 1 - \frac{1}{2} \rho^2 + \frac{ 1 + 4 M }{ 16 } \cdot \rho^4 \right) \mathring{\gamma} & n = 4 \text{.}
\end{cases}
\end{align}
\end{itemize}
\end{proposition}

\if\comp1
\begin{proof}
See Appendix \ref{sec.comp_schw_ads_aux}.
\end{proof}
\fi

From the Fefferman--Graham gauge derived in Proposition \ref{thm.schw_ads_aux}, we can extract the corresponding partial expansions near the boundary for the Schwarzschild-AdS metric:

\begin{corollary} \label{thm.schw_ads_fg}
Fix $M \in \R$ and $n > 1$, let $( \mi{M}_M, g_M )$ be as in Definition \ref{def.schw_ads}, and let $( \mi{M}, g )$ be as in Proposition \ref{thm.schw_ads_aux}.
Then, $( \mi{M}, g )$ is a vacuum FG-aAdS segment, with vertical metric
\begin{align}
\label{eq.schw_ads_fg} \gv &= ( - dt^2 + \mathring{\gamma} ) - \frac{1}{2} ( dt^2 + \mathring{\gamma} ) \cdot \rho^2 + \frac{1}{16} ( - dt^2 + \mathring{\gamma} ) \cdot \rho^4 \\
\notag &\qquad + \frac{ M }{ n } [ ( n - 1 ) dt^2 + \mathring{\gamma} ] \cdot \rho^n + \ms{r} \cdot \rho^{ n + 2 } \text{.}
\end{align}
where the vertical $( 0, 2 )$-tensor field $\ms{r}$ is locally bounded in $C^M$ for any $M \geq 0$.
Moreover, in the special case $n = 2$, we have the following explicit formula for $\gv$:
\begin{align}
\label{eq.schw_ads_fg_low} \gv &= ( - dt^2 + \mathring{\gamma} ) - \frac{ 1 - M }{2} ( dt^2 + \mathring{\gamma} ) \cdot \rho^2 + \frac{ ( 1 - M )^2 }{ 16 } ( - dt^2 + \mathring{\gamma} ) \cdot \rho^4 \text{.}
\end{align}
\end{corollary}

\begin{proof}
It is clear from \eqref{eq.schw_ads_aux_manifold} that $\mi{M}$ is an aAdS region (see Definition \ref{def.aads_manifold}), with $\mi{I} := \R \times \Sph^{ n - 1 }$.
Moreover, from \eqref{eq.schw_ads_aux_metric}, we see that $g$ can be written in the form \eqref{eq.aads_metric}, with vertical metric
\begin{align}
\label{eql.schw_ads_fg_1} \gv &= - \left[ 1 + \frac{1}{2} \rho^2 + \frac{1}{16} \rho^4 - \frac{ ( n - 1 ) M }{ n } \cdot \rho^n + \alpha_1 ( \rho ) \cdot \rho^{ n + 2 } \right] dt^2 \\
\notag &\qquad + \left[ 1 - \frac{1}{2} \rho^2 + \frac{1}{16} \rho^4 + \frac{ M }{ n } \cdot \rho^n + \alpha_2 ( \rho ) \cdot \rho^{ n + 2 } \right] \mathring{\gamma} \text{.}
\end{align}
Expanding the above immediately yields \eqref{eq.schw_ads_fg}.
(In particular, $\ms{r}$ depends only on $\rho$, so the $C^M$-boundedness of $\ms{r}$ for any $M \geq 0$ is trivial.)
A similar computation using the first part of \eqref{eq.schw_ads_aux_low} yields the remaining formula \eqref{eq.schw_ads_fg_low}.
Observe also that \eqref{eq.aads_metric_limit} holds, with Lorentzian limit
\[
\gm := - dt^2 + \mathring{\gamma} \text{.}
\]

The above proves that $( \mi{M}, g )$ is indeed a vacuum FG-aAdS segment.
Finally, since $( \mi{M}_M, g_M )$ is known to satisfy the Einstein-vacuum equations, $( \mi{M}, g )$ must be vacuum as well.
\end{proof}

\begin{corollary} \label{thm.schw_ads_exp}
Assume the hypotheses, and hence the conclusions, of Corollary \ref{thm.schw_ads_fg}.
Then, in the context of Theorem \ref{thm.fg_exp}, the expansion \eqref{eq.fg_exp} of the vertical metric $\gv$ in \eqref{eq.schw_ads_fg} satisfies
\begin{align}
\label{eq.schw_ads_exp} \gb{0} &= - dt^2 + \mathring{\gamma} \text{,} \\
\notag \gb{2} &= \begin{cases}
- \frac{1}{2} ( dt^2 + \mathring{\gamma} ) & n > 2 \text{,} \\
- \frac{1}{2} ( dt^2 + \mathring{\gamma} ) + \frac{ M }{2} ( dt^2 + \mathring{\gamma} ) & n = 2 \text{,}
\end{cases} \\
\notag \gb{n} &= \begin{cases}
\frac{ M }{ n } [ ( n - 1 ) dt^2 + \mathring{\gamma} ] & n \not\in \{ 2, 4 \} \text{,} \\
- \frac{1}{2} ( dt^2 + \mathring{\gamma} ) + \frac{ M }{2} ( dt^2 + \mathring{\gamma} ) & n = 2 \text{,} \\
\frac{1}{16} ( -dt^2 + \mathring{\gamma} ) + \frac{M}{4} ( 3 dt^2 + \mathring{\gamma} ) & n = 4 \text{.}
\end{cases}
\end{align}
\end{corollary}

Note in particular that $\gb{0}$ is independent of $M$.
One can also check that $\gb{2}$ satisfies the last formula of \eqref{eq.fg_main_schouten}, with $\gm = \gb{0}$.
Furthermore, $\gb{n}$ is the first coefficient that depends on the mass parameter $M$, and each value of $M$ corresponds to a unique value of $\gb{n}$.
Finally, note that the logarithmic coefficient $\gb{\star}$ always vanishes, regardless of dimension.

\appendix

\if\comp1

\section{Details and Computations} \label{sec.comp}

This appendix contains additional proofs, computations, and details for readers' convenience.
In particular, propositions which were stated but not proved in the main text are proved here.

\subsection{Proof of Proposition \ref{thm.g_rho}} \label{sec.comp_g_rho}

First, note that \eqref{eq.g_rho} and the first two parts of \eqref{eq.geod_rho} follow trivially from \eqref{eq.aads_metric}.
Moreover, from \eqref{eq.g_rho} and the fact that $\partial_\rho$ and $\partial_a$ commute, we have
\[
g ( \nabla_\rho \partial_\rho, \partial_\rho ) = \frac{1}{2} \partial_\rho g_{ \rho \rho } = - \rho^{-3} \text{,} \qquad g ( \nabla_\rho \partial_\rho, \partial_a ) = - g ( \partial_\rho, \nabla_\rho \partial_a ) = - \frac{1}{2} \partial_a g_{ \rho \rho } = 0 \text{.}
\]
Observe in particular that \eqref{eq.g_rho} and the above imply
\begin{align*}
g ( \nabla_N N, \partial_\rho ) &= \rho^2 \cdot g ( \nabla_\rho \partial_\rho, \partial_\rho ) + \rho \cdot g ( \partial_\rho, \partial_\rho ) = 0 \text{,} \\
g ( \nabla_N N, \partial_a ) &= \rho^2 \cdot g ( \nabla_\rho \partial_\rho, \partial_a ) + \rho \cdot g ( \partial_\rho, \partial_a ) = 0 \text{,}
\end{align*}
from which the last part of \eqref{eq.geod_rho} immediately follows.

Next, since $\partial_\rho$ commutes with the $\partial_a$'s, then \eqref{eq.aads_metric} implies
\[
- \frac{1}{2} \rho^{-1} \mi{L}_\rho \gv_{ a b } + \rho^{-2} \gv_{ a b } = - \frac{1}{2} \rho \partial_\rho g_{ a b } = - \frac{1}{2} \rho \cdot [ g ( \nabla_a \partial_\rho, \partial_b ) + g ( \partial_a, \nabla_b \partial_\rho ) ] \text{.}
\]
The desired identity \eqref{eq.sff_rho} now follows from the above and the fact that
\[
g ( \partial_a, \nabla_b \partial_\rho ) = - g ( \nabla_b \partial_a, \partial_\rho ) = - g ( \nabla_a \partial_b, \partial_\rho ) = g ( \nabla_a \partial_\rho, \partial_b ) \text{,}
\]
which is a consequence of \eqref{eq.g_rho}.

\subsection{Proof of Proposition \ref{thm.einstein_ex}} \label{sec.comp_einstein_ex}

The second identity in \eqref{eq.einstein_ex} is obtained by taking the $g$-trace of \eqref{eq.einstein}.
The first part of \eqref{eq.einstein_ex} then follows by substituting the second part into \eqref{eq.einstein}.
The last identity in \eqref{eq.einstein_ex} now follows from the above and from the standard decomposition formula for $R$:
\begin{align*}
W_{ \alpha \beta \gamma \delta } &= R_{ \alpha \beta \gamma \delta } - \frac{1}{ n - 1 } ( g_{ \alpha \gamma } Rc_{ \beta \delta } - g_{ \alpha \delta } Rc_{ \beta \gamma } - g_{ \beta \gamma } Rc_{ \alpha \delta } + g_{ \beta \delta } Rc_{ \gamma \alpha } ) \\
&\qquad + \frac{1}{ n (n - 1) } Rs ( g_{ \alpha \beta } g_{ \beta \delta } - g_{ \alpha \delta } g_{ \beta \gamma } ) \\
&= R_{ \alpha \beta \gamma \delta } + \frac{ 2 n }{ n - 1 } ( g_{ \alpha \gamma } g_{ \beta \delta } - g_{ \alpha \delta } g_{ \beta \gamma } ) - \frac{ n + 1 }{ n - 1 } ( g_{ \alpha \beta } g_{ \beta \delta } - g_{ \alpha \delta } g_{ \beta \gamma } ) \\
&= R_{ \alpha \beta \gamma \delta } + ( g_{ \alpha \gamma } g_{ \beta \delta } - g_{ \alpha \delta } g_{ \beta \gamma } ) \text{.}
\end{align*}

\subsection{Proof of Proposition \ref{thm.comm_vertical}} \label{sec.comp_comm_vertical}

Letting $\ms{\Gamma}^a_{ b c }$ denote the Christoffel symbols for $\gv$ with respect to $\varphi$-coordinates, then a direct computation yields the following identity:
\begin{align}
\label{eql.comm_vertical_1} [ \mi{L}_\rho, \Dv_a ] \ms{A}^{ b_1 \dots b_k }_{ c_1 \dots c_l } &= \mi{L}_\rho \Dv_a \ms{A}_b - \Dv_a \mi{L}_\rho \ms{A}^{ b_1 \dots b_k }_{ c_1 \dots c_l } \\
\notag &= \partial_\rho \partial_a \ms{A}^{ b_1 \dots b_k }_{ c_1 \dots c_l } + \sum_{ j = 1 }^k \partial_\rho ( \ms{\Gamma}_{ a d }^{ b_j } \ms{A}^{ b_1 \hat{d}_j b_k }_{ c_1 \dots c_l } ) - \sum_{ j = 1 }^l \partial_\rho ( \ms{\Gamma}_{ a c_j }^d \ms{A}^{ b_1 \dots b_k }_{ c_1 \hat{d}_j c_l } ) \\
\notag &\qquad - \partial_a \partial_\rho \ms{A}^{ b_1 \dots b_k }_{ c_1 \dots c_l } - \sum_{ j = 1 }^k \ms{\Gamma}_{ a d }^{ b_j } \partial_\rho \ms{A}^{ b_1 \hat{d}_j b_k }_{ c_1 \dots c_l } + \sum_{ j = 1 }^l \ms{\Gamma}_{ a c_j }^d \partial_\rho \ms{A}^{ b_1 \dots b_k }_{ c_1 \hat{d}_j c_l } \\
\notag &= \sum_{ j = 1 }^k \partial_\rho \ms{\Gamma}_{ a d }^{ b_j } \cdot \ms{A}^{ b_1 \hat{d}_j b_k }_{ c_1 \dots c_l } - \sum_{ j = 1 }^l \partial_\rho \ms{\Gamma}_{ a c_j }^d \cdot \ms{A}^{ b_1 \dots b_k }_{ c_1 \hat{d}_j c_l } \text{.}
\end{align}
To determine $\partial_\rho \ms{\Gamma}^b_{ a c }$, we apply \eqref{eql.comm_vertical_1}, with $\ms{A} := \gv$, to obtain
\[
\Dv_a \mi{L}_\rho \gv_{ b c } = \mi{L}_\rho \Dv_a \gv_{ b c } - [ \mi{L}_\rho, \Dv_a ] \gv_{ b c } = \partial_\rho \ms{\Gamma}_{ a b }^d \cdot \gv_{ d c } + \partial_\rho \ms{\Gamma}_{ a c }^d \cdot \gv_{ b d } \text{.}
\]
Using the above, we then conclude that
\begin{align*}
\gv^{ b d } ( \Dv_a \mi{L}_\rho \gv_{ d c } + \Dv_c \mi{L}_\rho \gv_{ d a } - \Dv_d \mi{L}_\rho \gv_{ a c } ) &= \gv^{ b d } ( \gv_{ e c } \partial_\rho \ms{\Gamma}_{ a d }^e + \gv_{ d e } \partial_\rho \ms{\Gamma}_{ a c }^e ) + \gv^{ b d } ( \gv_{ e a } \partial_\rho \ms{\Gamma}_{ c d }^e + \gv_{ d e } \partial_\rho \ms{\Gamma}_{ c a }^e ) \\
&\qquad - \gv^{ b d } ( \gv_{ e c } \partial_\rho \ms{\Gamma}_{ d a }^e + \gv_{ a e } \partial_\rho \ms{\Gamma}_{ d c }^e ) \\
&= 2 \cdot \partial_\rho \ms{\Gamma}_{ a c }^b \text{.}
\end{align*}
Combining the above with \eqref{eql.comm_vertical_1} results in \eqref{eq.comm_vertical}.

\subsection{Proof of Proposition \ref{thm.curvature_vertical}} \label{sec.comp_curvature_vertical}

Let $\ms{X}$ be a vector field on $\mi{M}$ that is both vertical and independent of $\rho$.
By the definition of the Riemann curvature and \eqref{eq.comm_vertical}, we have that
\begin{align*}
\mi{L}_\rho ( \Rv^c{}_{ d a b } \ms{X}^d ) &= \mi{L}_\rho ( \Dv_{ a b } \ms{X}^c - \Dv_{ b a } \ms{X}^c ) \\
&= \Dv_a \mi{L}_\rho \Dv_b \ms{X}^c - \Dv_b \mi{L}_\rho \Dv_a \ms{X}^c - \frac{1}{2} \gv^{ d e } ( \Dv_a \mi{L}_\rho \gv_{ e b } + \Dv_b \mi{L}_\rho \gv_{ e a } - \Dv_e \mi{L}_\rho \gv_{ a b } ) \Dv_d \ms{X}^c \\
&\qquad + \frac{1}{2} \gv^{ d e } ( \Dv_b \mi{L}_\rho \gv_{ e a } + \Dv_a \mi{L}_\rho \gv_{ e b } - \Dv_e \mi{L}_\rho \gv_{ b a } ) \Dv_d \ms{X}^c \\
&\qquad + \frac{1}{2} \gv^{ c e } ( \Dv_a \mi{L}_\rho \gv_{ e d } + \Dv_d \mi{L}_\rho \gv_{ e a } - \Dv_e \mi{L}_\rho \gv_{ a d } ) \Dv_b \ms{X}^d \\
&\qquad - \frac{1}{2} \gv^{ c e } ( \Dv_b \mi{L}_\rho \gv_{ e d } + \Dv_d \mi{L}_\rho \gv_{ e b } - \Dv_e \mi{L}_\rho \gv_{ b d } ) \Dv_a \ms{X}^d \text{.}
\end{align*}
Applying \eqref{eq.comm_vertical} again to the right-hand side and recalling that $\mi{L}_\rho \ms{X} = 0$, we obtain
\begin{align*}
\mi{L}_\rho ( \Rv^c{}_{ d a b } \ms{X}^d ) &= \frac{1}{2} \Dv_a [ \gv^{ c e } ( \Dv_b \mi{L}_\rho \gv_{ e d } + \Dv_d \mi{L}_\rho \gv_{ e b } - \Dv_e \mi{L}_\rho \gv_{ b d } ) \ms{X}^d ] \\
&\qquad - \frac{1}{2} \Dv_b [ \gv^{ c e } ( \Dv_a \mi{L}_\rho \gv_{ e d } + \Dv_d \mi{L}_\rho \gv_{ e a } - \Dv_e \mi{L}_\rho \gv_{ a d } ) \ms{X}^d ] \\
&\qquad + \frac{1}{2} \gv^{ c e } ( \Dv_a \mi{L}_\rho \gv_{ e d } + \Dv_d \mi{L}_\rho \gv_{ e a } - \Dv_e \mi{L}_\rho \gv_{ a d } ) \Dv_b \ms{X}^d \\
&\qquad - \frac{1}{2} \gv^{ c e } ( \Dv_b \mi{L}_\rho \gv_{ e d } + \Dv_d \mi{L}_\rho \gv_{ e b } - \Dv_e \mi{L}_\rho \gv_{ b d } ) \Dv_a \ms{X}^d \\
&= \frac{1}{2} \gv^{ c e } ( \Dv_{ a b } \mi{L}_\rho \gv_{ e d } + \Dv_{ a d } \mi{L}_\rho \gv_{ e b } - \Dv_{ a e } \mi{L}_\rho \gv_{ b d } ) \ms{X}^d \\ 
&\qquad - \frac{1}{2} \gv^{ c e } ( \Dv_{ b a } \mi{L}_\rho \gv_{ e d } + \Dv_{ b d } \mi{L}_\rho \gv_{ e a } - \Dv_{ b e } \mi{L}_\rho \gv_{ a d } ) \ms{X}^d \text{.} 
\end{align*}
Taking $\ms{X}$ to be the coordinate vector field $\partial_{ x^d }$ in the above yields
\[
\mi{L}_\rho \Rv^c{}_{ d a b } = \frac{1}{2} \gv^{ c e } ( \Dv_{ a b } \mi{L}_\rho \gv_{ e d } + \Dv_{ a d } \mi{L}_\rho \gv_{ e b } - \Dv_{ a e } \mi{L}_\rho \gv_{ b d } - \Dv_{ b a } \mi{L}_\rho \gv_{ e d } - \Dv_{ b d } \mi{L}_\rho \gv_{ e a } + \Dv_{ b e } \mi{L}_\rho \gv_{ a d } ) \text{,}
\]
which, after reindexing, is precisely \eqref{eq.curvature_vertical}.

\subsection{Proof of Proposition \ref{thm.einstein_vertical_ex}} \label{sec.comp_einstein_vertical_ex}

The first equation in \eqref{eq.einstein_vertical_ex} follows from the second part of \eqref{eq.einstein_vertical} and Definition \ref{def.schematic_vertical}, while the third equation in \eqref{eq.einstein_vertical_ex} is a consequence of \eqref{eq.curvature_vertical}.

Taking a $\gv$-trace of the first identity of \eqref{eq.einstein_vertical_ex} and then commuting $\mi{L}_\rho$ with $\trace{\gv}$ yields
\begin{equation}
\label{eql.einstein_vertical_ex_1} \rho \mi{L}_\rho ( \trace{\gv} \mi{L}_\rho \gv ) - \rho \mi{L}_\rho \gv^{ a b } \mi{L}_\rho \gv_{ a b } - ( 2 n - 1 ) \trace{\gv} \mi{L}_\rho \gv = 2 \rho \cdot \Rsv + \rho \cdot \mi{S} ( \gv; \mi{L}_\rho \gv, \mi{L}_\rho \gv ) \text{.}
\end{equation}
From standard identities regarding Lie derivatives of metric duals, we see that
\[
\mi{L}_\rho \gv^{ a b } = - \gv^{ a c } \gv^{ b d } \mi{L}_\rho \gv_{ c d } \text{,} \qquad \rho \mi{L}_\rho \gv^{ a b } \mi{L}_\rho \gv_{ a b } = \rho \cdot \mi{S} ( \gv; \mi{L}_\rho \gv, \mi{L}_\rho \gv ) \text{.}
\]
The second identity in \eqref{eq.einstein_vertical_ex} now follows from \eqref{eql.einstein_vertical_ex_1} and the above.

\subsection{Proof of Proposition \ref{thm.comm_vertical_deriv}} \label{sec.comp_comm_vertical_deriv}

We prove \eqref{eq.comm_vertical_deriv} via induction on $k$.
First, the case $k = 1$ is simply \eqref{eq.comm_vertical}.
Now, if \eqref{eq.comm_vertical_deriv} holds for some $k \geq 1$, then by the induction hypothesis and \eqref{eq.comm_vertical},
\begin{align*}
[ \mi{L}_\rho^{ k + 1 }, \Dv ] \ms{A} &= [ \mi{L}_\rho, \Dv ] \mi{L}_\rho^k \ms{A} + \mi{L}_\rho ( [ \mi{L}_\rho^k, \Dv ] \ms{A} ) \\
&= \mi{S} ( \gv; \mi{L}_\rho^k \ms{A}, \Dv \mi{L}_\rho \gv ) + \sum_{ \substack{ j + j_0 + \dots + j_l = k \\ 0 \leq j < k \text{, } j_p \geq 1 } } \mi{L}_\rho [ \mi{S} ( \gv; \mi{L}_\rho^j \ms{A}, \Dv \mi{L}_\rho^{ j_0 } \gv, \mi{L}_\rho^{ j_1 } \gv, \dots, \mi{L}_\rho^{ j_l } \gv ) ] \text{.}
\end{align*}
Applying the Leibniz rule and then \eqref{eq.comm_vertical} on the right-hand side yields
\begin{align*}
[ \mi{L}_\rho^{ k + 1 }, \Dv ] \ms{A} &= \sum_{ \substack{ j + j_0 + \dots + j_l = k + 1 \\ 0 \leq j < k + 1 \text{, } j_p \geq 1 } } \mi{S} ( \gv; \mi{L}_\rho^j \ms{A}, \Dv \mi{L}_\rho^{ j_0 } \gv, \mi{L}_\rho^{ j_1 } \gv, \dots, \mi{L}_\rho^{ j_l } \gv ) \\
&\qquad + \sum_{ \substack{ j + j_0 + \dots + j_l = k \\ 0 \leq j < k \text{, } j_p \geq 1 } } \mi{S} ( \gv; \mi{L}_\rho^j \ms{A}, \mi{L}_\rho \Dv \mi{L}_\rho^{ j_0 } \gv, \mi{L}_\rho^{ j_1 } \gv, \dots, \mi{L}_\rho^{ j_l } \gv ) \\
&= \sum_{ \substack{ j + j_0 + \dots + j_l = k + 1 \\ 0 \leq j < k + 1 \text{, } j_p \geq 1 } } \mi{S} ( \gv; \mi{L}_\rho^j \ms{A}, \Dv \mi{L}_\rho^{ j_0 } \gv, \mi{L}_\rho^{ j_1 } \gv, \dots, \mi{L}_\rho^{ j_l } \gv ) \text{,}
\end{align*}
which completes the proof of \eqref{eq.comm_vertical_deriv}.

\subsection{Proof of Proposition \ref{thm.curvature_vertical_deriv}} \label{sec.comp_curvature_vertical_deriv}

We first apply $\mi{L}_\rho^{ k - 1 }$ to \eqref{eq.curvature_vertical} and use the Leibniz rule:
\begin{equation}
\label{eql.curvature_vertical_deriv_0} \mi{L}_\rho^k \Rv = \mi{L}_\rho^{ k - 1 } [ \mi{S} ( \gv; \Dv^2 \mi{L}_\rho \gv ) ] = \sum_{ \substack{ j_0 + \dots + j_l = k - 1 \\ j_p \geq 1 } } \mi{S} ( \gv; \mi{L}_\rho^{ j_0 } \Dv^2 \mi{L}_\rho \gv, \mi{L}_\rho^{ j_1 } \gv, \dots, \mi{L}_\rho^{ j_l } \gv ) \text{.}
\end{equation}
For $j_0 \geq 0$, we apply \eqref{eq.comm_vertical_deriv} repeatedly to obtain
\begin{align*}
\mi{L}_\rho^{ j_0 } \Dv^2 \mi{L}_\rho \gv &= \Dv \mi{L}_\rho^{ j_0 } \Dv \mi{L}_\rho \gv + \sum_{ \substack{ j' + j_0' + \dots + j'_{ l' } = j_0 \\ 0 \leq j' < j_0 \text{, } j'_p \geq 1 } } \mi{S} ( \gv; \mi{L}_\rho^{ j' } \Dv \mi{L}_\rho \gv, \Dv \mi{L}_\rho^{ j'_0 } \gv, \mi{L}_\rho^{ j'_1 } \gv, \dots, \mi{L}_\rho^{ j'_{ l' } } \gv ) \\
&= \Dv^2 \mi{L}_\rho^{ j_0 + 1 } \gv + \sum_{ \substack{ j' + j_0' + \dots + j'_{ l' } = j_0 \\ 0 \leq j' < j_0 \text{, } j'_p \geq 1 } } \Dv [ \mi{S} ( \gv; \mi{L}_\rho^{ j' } \mi{L}_\rho \gv, \Dv \mi{L}_\rho^{ j'_0 } \gv, \mi{L}_\rho^{ j'_1 } \gv, \dots, \mi{L}_\rho^{ j'_{ l' } } \gv ) ] \\
&\qquad + \sum_{ \substack{ j_0' + \dots + j'_{ l' } = j_0 + 1 \\ j'_p \geq 1 } } \mi{S} ( \gv; \Dv \mi{L}_\rho^{ j'_0 } \gv, \Dv \mi{L}_\rho^{ j'_1 } \gv, \mi{L}_\rho^{ j'_2 } \gv, \dots, \mi{L}_\rho^{ j'_{ l' } } \gv ) \\
&= \sum_{ \substack{ i'_1 + \dots + i'_{ l' } = 2 \\ j'_1 + \dots + j'_l = j_0 + 1 \\ j'_p \geq 1 } } \mi{S} ( \gv; \Dv^{ i'_1 } \mi{L}_\rho^{ j'_1 } \gv, \dots, \Dv^{ i'_{ l' } } \mi{L}_\rho^{ j'_{ l' } } \gv ) \text{.}
\end{align*}
Combining \eqref{eql.curvature_vertical_deriv_0} with the above results in the first identity in \eqref{eq.curvature_vertical_deriv}.

Next, we apply $\mi{L}_\rho^{ k - 1 }$ to the first equation in \eqref{eq.einstein_vertical}:
\begin{equation}
\label{eql.curvature_vertical_deriv_1} 0 = \mi{L}_\rho^{ k - 1 } ( \gv^{ b c } \Dv_b \mi{L}_\rho \gv_{ a c } ) - \mi{L}_\rho^{ k - 1 } \Dv_a ( \trace{\gv} \mi{L}_\rho \gv ) \text{.}
\end{equation}
For each of the terms on the right-hand side of \eqref{eql.curvature_vertical_deriv_1}, we use the Leibniz rule to commute $\mi{L}_\rho^{ k - 1 }$ with $\gv^{-1}$, and we use \eqref{eq.comm_vertical_deriv} to commute $\mi{L}_\rho^{ k - 1 }$ with $\Dv$.
This yields
\begin{align*}
\mi{L}_\rho^{ k - 1 } ( \gv^{ b c } \Dv_b \mi{L}_\rho \gv_{ a c } ) &= \gv^{ b c } \mi{L}_\rho^{ k - 1 } \Dv_b \mi{L}_\rho \gv_{ a c } + \sum_{ \substack{ j + j_1 + \dots + j_l = k - 1 \\ 0 \leq j < k - 1 \text{, } j_p \geq 1 } } [ \mi{S} ( \gv; \mi{L}_\rho^j \Dv \mi{L}_\rho \gv, \mi{L}_\rho^{ j_1 } \gv, \dots, \mi{L}_\rho^{ j_l } \gv ) ]_a \\
&= \gv^{ b c } \Dv_b \mi{L}_\rho^k \gv_{ a c } + \sum_{ \substack{ j_0 + \dots + j_l = k \\ 1 \leq j_p < k } } [ \mi{S} ( \gv; \Dv \mi{L}_\rho^{ j_0 } \gv, \mi{L}_\rho^{ j_1 } \gv, \dots, \mi{L}_\rho^{ j_l } \gv ) ]_a \text{,} \\
\Dv_a ( \trace{\gv} \mi{L}_\rho^k \gv ) &= \Dv_a \mi{L}_\rho^{ k - 1 } ( \trace{\gv} \mi{L}_\rho \gv ) + \sum_{ \substack{ j_1 + \dots + j_l = k \\ 1 \leq j_p < k } } \Dv_a [ \mi{S} ( \gv; \mi{L}_\rho^{ j_1 } \gv, \dots, \mi{L}_\rho^{ j_l } \gv ) ] \\
&= \mi{L}_\rho^{ k - 1 } \Dv_a ( \trace{\gv} \mi{L}_\rho \gv ) + \sum_{ \substack{ j_0 + \dots + j_l = k \\ 1 \leq j_p < k } } [ \mi{S} ( \gv; \Dv \mi{L}_\rho^{ j_0 } \gv, \mi{L}_\rho^{ j_1 } \gv, \dots, \mi{L}_\rho^{ j_l } \gv ) ]_a \text{.}
\end{align*}
The second identity of \eqref{eq.curvature_vertical_deriv} now follows from \eqref{eql.curvature_vertical_deriv_1} and the above.

\subsection{Proof of Proposition \ref{thm.limit_bound}} \label{sec.comp_limit_bound}

For any compact coordinate system $( U, \varphi )$ on $\mi{I}$, we have that
\[
\lim_{ \sigma \searrow 0 } \sup_{ \{ \sigma \} \times U } | \rho^p \ms{A} |_{ M, \varphi } \leq \left( \lim_{ \sigma \searrow 0 } \sigma^p \right) \| \ms{A} \|_{ M, \varphi } = 0 \text{,}
\]
and it follows that $\rho^p \ms{A} \rightarrow^M 0$.
The desired property $\rho^p \ms{A} \Rightarrow^M 0$ also follows, since
\[
\int_0^{ \rho_0 } \frac{1}{ \sigma } | \rho^p \ms{A} |_{ M, \varphi } |_\sigma d \sigma \lesssim \int_0^{ \rho_0 } \sigma^{ -1 + p } d \sigma \cdot \| \ms{A} \|_{ M, \varphi } < \infty \text{.}
\]

\subsection{Proof of Proposition \ref{thm.limit_schematic}} \label{sec.comp_limit_schematic}

Fix a compact coordinate system $( U, \varphi )$ on $\mi{I}$, and assume that the limits $\ms{A}_j \rightarrow^M \mf{A}_j$ hold for every $1 \leq j \leq N$.
Then, by compactness,
\begin{equation}
\label{eql.limit_schematic_1} \| \ms{A}_j \|_{ M, \varphi } \leq C \text{,} \qquad \| \mf{A}_j \|_{ M, \varphi } \leq C \text{,} \qquad 1 \leq j \leq N \text{,}
\end{equation}
for some $C > 0$.
Now, given any vertical tensor field
\begin{equation}
\label{eql.limit_schematic_2} \ms{B} = \mi{S} ( \ms{A}_1, \dots, \ms{A}_N ) := \sum_{ j = 1 }^T \mc{Q}_j ( \ms{A}_1 \otimes \dots \otimes \ms{A}_N ) \text{,}
\end{equation}
where $T \geq 0$ and each $\mc{Q}_j$, $1 \leq i \leq T$, is schematically trivial, we let
\begin{equation}
\label{eql.limit_schematic_3} \mf{G} = \mi{S} ( \mf{A}_1, \dots, \mf{A}_N ) := \sum_{ j = 1 }^T \mc{Q}_j ( \mf{A}_1 \otimes \dots \otimes \mf{A}_N ) \text{.}
\end{equation}
Using \eqref{eql.limit_schematic_1}, we estimate, for any $1 \leq j \leq N$,
\[
| \mc{Q}_j ( \ms{A}_1 \otimes \dots \otimes \ms{A}_N ) - \mc{Q}_j ( \mf{A}_1 \otimes \dots \otimes \mf{A}_N ) |_{ M, \varphi } \lesssim_C \sum_{ i = 1 }^N | \ms{A}_i - \mf{A}_i |_{ M, \varphi } \text{.}
\]
Combining \eqref{eql.limit_schematic_2}, \eqref{eql.limit_schematic_3}, and the above, we conclude that
\begin{equation}
\label{eql.limit_schematic_4} | \ms{B} - \mf{G} |_{ M, \varphi } \lesssim_C \sum_{ i = 1 }^N | \ms{A}_i - \mf{A}_i |_{ M, \varphi } \text{.}
\end{equation}

In particular, \eqref{eql.limit_schematic_4} implies $\ms{B} \rightarrow^M \mf{G}$, which proves the first part of the proposition.
If $\ms{A}_j \Rightarrow^M \mf{A}_j$ for each $1 \leq j \leq N$ as well, then \eqref{eql.limit_schematic_4} also implies $\ms{B} \Rightarrow^M \mf{G}$.

\subsection{Proof of Proposition \ref{thm.schw_ads_aux}} \label{sec.comp_schw_ads_aux}

First, note that since $r_0$ is assumed to be large, the relation
\begin{equation}
\label{eql.schw_ads_aux_0} r := \frac{1}{2} ( \xi^{-1} - \xi ) \text{,} \qquad 0 < \xi \ll 1 \text{,}
\end{equation}
defines a smooth and bijective function $\xi$ on $\mi{M}_M$ whose range is an interval $( 0, \xi_0 ]$, for some $\xi_0 \ll 1$.
Furthermore, direct computations using \eqref{eql.schw_ads_aux_0} yield the following:
\begin{equation}
\label{eql.schw_ads_aux_1} dr = - \frac{1}{2} ( \xi^{-2} + 1 ) d \xi \text{,} \qquad r^2 + 1 = \frac{1}{4} ( \xi^{-1} + \xi )^2 \text{.}
\end{equation}

Combining \eqref{eql.schw_ads_aux_0} and \eqref{eql.schw_ads_aux_1}, we then see that
\begin{align*}
r^2 + 1 - \frac{M}{ r^{ n - 2 } } &= \frac{1}{4} ( \xi^{-1} + \xi )^2 - \frac{ 2^{ n - 2 } M }{ ( \xi^{-1} - \xi )^{ n - 2 } } \\
&= \frac{1}{4} \xi^{-2} [ ( 1 + \xi^2 )^2 - 2^n M \xi^n \cdot ( 1 - \xi^2 )^{ - n + 2 } ] \text{,} \\
\left( r^2 + 1 - \frac{M}{ r^{ n - 2 } } \right)^{-1} dr^2 &= \frac{ ( \xi^{-2} + 1 )^2 d \xi^2 }{ \xi^{-2} [ ( 1 + \xi^2 )^2 - 2^n M \xi^n \cdot ( 1 - \xi^2 )^{ - n + 2 } ] } \\
&= \frac{ d \xi^2 }{ \xi^2 [ 1 - 2^n M \xi^n \cdot ( 1 - \xi^2 )^{ - n + 2 } ( 1 + \xi^2 )^{-2} ] } \text{.}
\end{align*}
From \eqref{eq.schw_ads_metric}, \eqref{eql.schw_ads_aux_0}, and the above, we conclude that
\begin{align}
\label{eql.schw_ads_aux_2} g_M &= \frac{ d \xi^2 }{ \xi^2 [ 1 - 2^n M \xi^n \cdot ( 1 - \xi^2 )^{ -n + 2 } ( 1 + \xi^2 )^{-2} ] } \\
\notag &\qquad - \frac{ [ ( 1 + \xi^2 )^2 - 2^n M \xi^n \cdot ( 1 - \xi^2 )^{ -n + 2 } ] dt^2 }{ 4 \xi^2 } + \frac{ ( 1 - \xi^2 )^2 \mathring{\gamma} }{ 4 \xi^2 } \text{.}
\end{align}

From now on, we will use $\zeta_k$, for any $k \in \N$, to denote appropriately chosen real-analytic functions $\zeta_k: [ -\xi_0, \xi_0 ] \rightarrow \R$.
Next, we solve the differential relation
\begin{equation}
\label{eql.schw_ads_aux_10} \frac{ d \rho }{ \rho } = \frac{ d \xi }{ \xi [ 1 - 2^n M \xi^n \cdot ( 1 - \xi^2 )^{ -n + 2 } ( 1 + \xi^2 )^{-2} ]^\frac{1}{2} } \text{,}
\end{equation}
for a positive function $\rho$ on $\mi{M}_M$ (of only $\xi$), with the boundary condition that $\rho \searrow 0$ whenever $\xi \searrow 0$.
Since $\xi_0$ is small, the right-hand side of \eqref{eql.schw_ads_aux_10} can be expanded as
\[
\frac{ d \xi }{ \xi [ 1 - 2^n M \xi^n \cdot ( 1 - \xi^2 )^{ -n + 2 } ( 1 + \xi^2 )^{-2} ]^\frac{1}{2} } = [ \xi^{-1} + 2^{ n - 1 } M \xi^{ n - 1 } + \zeta_1 ( \xi ) \cdot \xi^{ n + 1 } ] d \xi \text{.}
\]
Integrating the above (and normalizing the constant of integration), we obtain
\[
\log \rho = \log 2 + \log \xi + \frac{ 2^{ n - 1 } M \xi^n }{ n } + \zeta_2 ( \xi ) \cdot \xi^{ n + 2 } \text{,} \qquad \rho = 2 \xi \left[ 1 + \frac{ 2^{ n - 1 } M \xi^n }{ n } + \zeta_3 ( \xi ) \cdot \xi^{ n + 2 } \right] \text{.}
\]

Inverting the second expression in the above yields
\begin{equation}
\label{eql.schw_ads_aux_11} \xi = \frac{1}{2} \rho \left[ 1 - \frac{ M \rho^n }{ 2 n } + \zeta_4 ( \xi ) \cdot \rho^{ n + 2 } \right] \text{.}
\end{equation}
Moreover, using \eqref{eql.schw_ads_aux_11}, we obtain the expansions
\begin{align*}
\xi^{-2} &= 4 \rho^{-2} \left[ 1 + \frac{ M \rho^n }{ n } + \zeta_5 ( \xi ) \cdot \rho^{ n + 2 } \right] \text{,} \\
( 1 + \xi^2 )^2 &= 1 + \frac{1}{2} \rho^2 + \frac{1}{16} \rho^4 + \zeta_6 ( \xi ) \cdot \rho^{ n + 2 } \text{,} \\
( 1 - \xi^2 )^2 &= 1 - \frac{1}{2} \rho^2 + \frac{1}{16} \rho^4 + \zeta_7 ( \xi ) \cdot \rho^{ n + 2 } \text{.}
\end{align*}
Thus, combining \eqref{eql.schw_ads_aux_11} and the above, we see that
\begin{align*}
\frac{1}{4} \xi^{-2} ( 1 - \xi^2 )^2 &= \frac{ 1 - \frac{1}{2} \rho^2 + \frac{1}{16} \rho^4 + \frac{ M }{ n } \cdot \rho^n + \zeta_8 ( \xi ) \cdot \rho^{ n + 2 } }{ \rho^2 } \text{,} \\
\frac{1}{4} \xi^{-2} [ ( 1 + \xi^2 )^2 - 2^n M \xi^n \cdot ( 1 - \xi^2 )^{ -n + 2 } ] &= \frac{ 1 + \frac{1}{2} \rho^2 + \frac{1}{16} \rho^4 - \frac{ ( n - 1 ) M }{ n } \cdot \rho^n + \zeta_9 ( \xi ) \cdot \rho^{ n + 2 } }{ \rho^2 } \text{.}
\end{align*}
Substituting \eqref{eql.schw_ads_aux_10} and the above into \eqref{eql.schw_ads_aux_2} yields
\begin{align}
\label{eql.schw_ads_aux_13} g_M &= \frac{ d \rho^2 }{ \rho^2 } - \frac{ [ 1 + \frac{1}{2} \rho^2 + \frac{1}{16} \rho^4 - \frac{ ( n - 1 ) M }{ n } \cdot \rho^n + \zeta_9 ( \xi ) \cdot \rho^{ n + 2 } ] dt^2 }{ \rho^2 } \\
\notag &\qquad + \frac{ [ 1 - \frac{1}{2} \rho^2 + \frac{1}{16} \rho^4 + \frac{ M }{ n } \cdot \rho^n + \zeta_8 ( \xi ) \cdot \rho^{ n + 2 } ] \cdot \mathring{\gamma} }{ \rho^2 } \text{.}
\end{align}

Now, the form \eqref{eq.schw_ads_aux_manifold} of the manifold $\mi{M}$ is obtained by mapping $\mi{M}_M$ through the coordinate changes $r \mapsto \xi \mapsto \rho$.
The form \eqref{eq.schw_ads_aux_metric} of the metric $g$ then follows from \eqref{eql.schw_ads_aux_13} and from expressing the remainders $\zeta_8 ( \xi )$ and $\zeta_9 ( \xi )$ as functions of $\rho$ instead.

Finally, in the cases $n = 2$ and $n = 4$, we can directly solve the differential relation
\begin{equation}
\label{eql.schw_ads_aux_20} \frac{ d \rho }{ \rho } = - \frac{ d r }{ \sqrt{ r^2 + 1 - M r^{ -n + 2 } } } = \begin{cases} - \frac{ d r }{ \sqrt{ r^2 + ( 1 - M ) } } & n = 2 \text{,} \\ - \frac{ d r }{ \sqrt{ r^2 + 1 - M r^{ -2 } } } & n = 4 \text{.} \end{cases}
\end{equation}
First, when $n = 2$, integrating \eqref{eql.schw_ads_aux_20} yields the (normalized) solution
\[
\log \rho = \log 2 - \log \left( r + \sqrt{ r^2 + ( 1 - M ) } \right) \text{,} \qquad \rho = 2 \left[ r + \sqrt{ r^2 + ( 1 - M ) } \right]^{-1} \text{.}
\]
To invert the above relation, we observe that
\begin{align*}
\frac{ ( M - 1 ) \rho }{4} + \frac{ 1 }{ \rho } &= \frac{ M - 1 }{ 2 \left[ r + \sqrt{ r^2 + ( 1 - M ) } \right] } \cdot \frac{ r - \sqrt{ r^2 + ( 1 - M ) } }{ r - \sqrt{ r^2 + ( 1 - M ) } } + \frac{1}{2} \left[ r + \sqrt{ r^2 + ( 1 - M ) } \right] \\
&= \frac{1}{2} \left[ r - \sqrt{ r^2 + ( 1 - M ) } \right] + \frac{1}{2} \left[ r + \sqrt{ r^2 + ( 1 - M ) } \right] \\
&= r \text{.}
\end{align*}
Substituting the above into \eqref{eq.schw_ads_metric}, in the case $n = 2$, yields the first part of \eqref{eq.schw_ads_aux_low}.

For the remaining case $n = 4$, we make the substitution $x = r^2$ in \eqref{eql.schw_ads_aux_20} and obtain
\[
\frac{ d \rho }{ \rho } = - \frac{ r dr }{ \sqrt{ r^4 + r^2 - M } } = - \frac{ dx }{ 2 \sqrt{ x^2 + x - M } } \text{.}
\]
The above can then be directly integrated:
\[
\log \rho = \log 2 - \log \left[ ( 2 x + 1 ) + 2 \sqrt{ x^2 + x - M } \right] \text{,} \qquad \rho^2 = 4 \left[ ( 2 r^2 + 1 ) + 2 \sqrt{ r^4 + r^2 - M } \right]^{-1} \text{.}
\]
To invert this, we note that
\begin{align*}
\frac{ ( 1 + 4 M ) \rho^2 }{16} + \frac{1}{ \rho^2 } &= \frac{ ( 1 + 4 M ) }{ 4 \left[ ( 2 r^2 + 1 ) + 2 \sqrt{ r^4 + r^2 - M } \right] } + \frac{ ( 2 r^2 + 1 ) + 2 \sqrt{ r^4 + r^2 - M } }{4} \\
&= \frac{ ( 2 r^2 + 1 ) - 2 \sqrt{ r^4 + r^2 - M } }{4} + \frac{ ( 2 r^2 + 1 ) + 2 \sqrt{ r^4 + r^2 - M } }{4} \\
&= r^2 + \frac{1}{2} \text{,} \\
- \frac{ ( 1 + 4 M ) \rho^2 }{16} + \frac{1}{ \rho^2 } &= - \frac{ ( 2 r^2 + 1 ) - 2 \sqrt{ r^4 + r^2 - M } }{4} + \frac{ ( 2 r^2 + 1 ) + 2 \sqrt{ r^4 + r^2 - M } }{4} \\
&= r \sqrt{ r^2 + 1 - M r^{-2} } \text{.}
\end{align*}
From the above, we then conclude that
\begin{align*}
r^2 &= \frac{1}{ \rho^2 } \left[ 1 - \frac{1}{2} \rho^2 + \frac{ 1 + 4 M }{ 16 } \cdot \rho^4 \right] \text{,} \\
r^2 + 1 - M r^{-2} &= \left[ \frac{1}{ \rho^2 } - \frac{ ( 1 + 4 M ) \rho^2 }{ 16 } \right]^2 \left[ \frac{1}{ \rho^2 } - \frac{1}{2} + \frac{ ( 1 + 4 M ) \rho^2 }{ 16 } \right]^{-1} \\
&= \frac{1}{ \rho^2 } \left[ 1 - \frac{ ( 1 + 4 M ) \rho^4 }{ 16 } \right]^2 \left[ 1 - \frac{1}{2} \rho^2 + \frac{ ( 1 + 4 M ) \rho^4 }{ 16 } \right]^{-1} \text{.}
\end{align*}
Putting the above into \eqref{eq.schw_ads_metric}, in the case $n = 4$, results in the second part of \eqref{eq.schw_ads_aux_low}.

\fi

\raggedright
\bibliographystyle{amsplain}
\bibliography{articles,books}

\end{document}